\documentclass[oneside, 12pt]{article}
\usepackage[english]{babel}
\usepackage{design_ASC}
\usepackage{mathtools}
\usepackage{amsmath}
\usepackage{amssymb}
\usepackage{amsthm}
\usepackage{wrapfig}
\usepackage{geometry}
\usepackage{setspace}
\usepackage{authblk}
\usepackage[toc,page]{appendix}
\usepackage{subcaption}
\usepackage{biblatex}
\usepackage[hidelinks]{hyperref}
\counterwithin{figure}{section}

\usepackage{titling}
\newcommand{\subtitle}[1]{%
  \posttitle{%
    \par\end{center}
    \begin{center}\large#1\end{center} 
    \vskip0.5em}%
}

\DeclarePairedDelimiterX\braket[2]{\langle}{\rangle}{#1 \delimsize\vert #2}
\usepackage[makeroom]{cancel}
\usepackage{array}

\numberwithin{equation}{section}

\theoremstyle{plain}
\newtheorem{thm}{Theorem}[section]
\newtheorem{lem}[thm]{Lemma}
\newtheorem{prop}[thm]{Proposition}
\newtheorem{cor}[thm]{Corollary}

\theoremstyle{definition}
\newtheorem{defn}[thm]{Definition}

\newtheorem{asmpt}[thm]{Assumptions}

\theoremstyle{remark}
\newtheorem*{rem}{Remark}
\newtheorem*{rems}{Remarks}

\newcommand{\norm}[1]{\left\lVert#1\right\rVert}

\newcommand{\supp}{\text{supp}}

\title{\vspace{-1em} Contrasting behaviour of two spherically symmetric perfect fluids near a weak null singularity in a spherically symmetric black hole\vspace{-6pt}} 
\author[*]{Raya V. Mancheva}  
\affil[*]{School of Mathematics and Maxwell Institute of Mathematical Sciences, University of Edinburgh}
\begin{document}
\maketitle
\setlength{\droptitle}{-2em}

\begin{abstract}
In this work we contrast the behaviour of two spherically symmetric matter models in a class of spherically symmetric spacetimes which feature a weak null singularity. This class in particular contains spherically symmetric perturbations of subextremal Reissner-Nordstr\"{o}m under the Einstein--Maxwell--scalar field system, a system for which a $C^2$ formulation of the strong cosmic censorship conjecture was proved by Luk-Oh \cite{Luk_Oh} and Dafermos \cite{Dafermos2012}. \\[5pt]
Firstly, we consider the Cauchy problem of spherically symmetric dust falling into the weak null singularity (WNS) where the initial dust velocity is normal to a smooth spacelike curve with certain properties. We prove that the flow of the dust velocity does not experience any shell-crossing before or at the singularity, the velocity vector remains timelike, and that the dust energy density remains bounded as matter approaches the singularity.\\[5pt]
Secondly, we consider the characteristic initial value problem for stiff perfect fluid falling into the WNS. By relating the stiff fluid velocity and energy density to a scalar field satisfying the homogeneous linear wave equation, we prove that this energy density becomes infinite as we approach the weak null singularity. Furthermore, we show that the ingoing component of the stiff fluid velocity blows up while the outgoing component approaches zero at the singularity. Therefore the velocity vector approaches an ingoing null vector tangent to the singular hypersurface.
\end{abstract}
\textbf{Acknowledgements}\\[5pt]
The author thanks her PhD advisor, Jan Sbierski, for the invaluable guidance on this work. The author also acknowledges the School of Mathematics at the University of Edinburgh, for financing this project. 
\tableofcontents 
\section{Introduction}\label{Introduction}
\subsection{Strong Cosmic Censorship in Spherical Symmetry}\label{SCSSS}
A central question in mathematical general relativity is the uniqueness of maximally extended spacetimes that solve the Einstein equations, particularly given generic, suitably regular initial data. This question, which forms the basis of the Strong Cosmic Censorship Conjecture, has proven difficult to address in full generality. However, significant advancements have been made under certain simplifying assumptions, particularly in the case of spherical symmetry.  The Strong Cosmic Censorship conjecture, originally conceived by Sir Roger Penrose states that maximal globally hyperbolic developments solving the Einstein equations with generic and suitably regular initial data are inextendible.  This is a statement of uniqueness in general relativity, because it ensures all of spacetime is determined by the initial data supplied on the initial hypersurface and no region outside of the domain of dependence of the initial data makes physical sense. Given a regularity class $\mathcal{G}$ (e.g. $\mathcal{G}=C^0, C^2$ or $H^1$ etc) the question of whether solutions to the Einstein equations with generic and suitably regular initial data are $\mathcal{G}$-extendible is referred to as the $\mathcal{G}$-formulation of the strong cosmic censorship conjecture. The $C^0$-formulation of the strong cosmic censorship has been disproven for solutions of the Einstein--Maxwell--scalar field system in spherical symmetry, see \cite{Dafermos03} and \cite{Dafermos05}. For trapped characteristic initial data, it is shown that the maximal development has a future null boundary along which curvature blows up but the metric is $C^0$-extendible. In \cite{Dafermos03} this is done for spherically symmetric data close to subextremal Reissner-Nordstr\"{o}m data, while \cite{Dafermos05} uses the assumption that a Price's law decay holds on the initial outgoing hypersurface. 
Since the Einstein Equations are second order, a metric $g$ has to be at least $C^2$ to be a classical solution, or $H^1$ to be a weak solution, suggesting that the $C^2$ and $H^1$ formulations of the Strong Cosmic Censorship conjecture are the most physically relevant formulations. It has been established that generic maximal developments of spherically symmetric Einsteinl--scalar field asymptotically flat initial data satisfying the constraint equations are unique as classical (at least $C^2$) solutions of the Einsteinl--scalar field system (see Christodolou, \cite{Chr10}). Similar results have been obtained for small scalar field perturbations of Reissner-Nordstr\"{o}m initial data under the Einstein--Maxwell--scalar field (EMSF) system (see Dafermos \cite{Dafermos2012} for 2-ended asymptotically flat Cauchy data, Luk-Oh \cite{Luk_Oh} for 1-ended developments of characteristic initial data). The abovementioned developments from initial data under the EMSF system are shown to possess a \textit{non-smooth Cauchy horizon} which is a \textit{weak null singularity} (WNS). This means the metric is extendible in a continuous manner beyond the Cauchy horizon, but curvature invariants blow up in any neighbourhood of any point of the singular boundary. Although there is no widely accepted definition of a WNS independent of particular examples, an alternative definition is given by Sbierski in \cite{Sbierski2024}, where a WNS would be a null hypersurface across which the spacetime metric admits a $C^0$ extension but no $C^{0,1}_{loc}$ extension (note that for an extension to be in $C^{0,1}_{loc}$, it has to be in  $H^1$). However, equivalence between curvature blow-up and $C^{0,1}_{loc}$-inextendibility remains an open question in mathematical general relativity, and works like \cite{Sbierski2021} and \cite{Sbierski2024} achieve significant results in proving one direction of the implication, namely that curvature blow-up implies $C^{0,1}_{loc}$ inextendibility. In any case, WNS is a terminal spacetime singularity because there cannot be a solution, even in the \textit{weak sense}, of the Einstein Equations, as noted in \cite{Luk_WNS}. In  \cite{Luk_Oh} and \cite{Dafermos2012} it is shown in particular that the Christoffel symbols blow up in a double null coordinate chart which extends continuously to the Cauchy horizon even though the metric admits a (non-unique) continuous extension past the Cauchy horizon. However, the metric admits no $C^2$ extension beyond the Cauchy horizon. This result is later improved by Sbierski in \cite{Sbierski2021}, showing that the spacetime metric has no $C^{0,1}_{loc}$-extension past the weak null singularity. \\[5pt]
An outstanding problem in general relativity and fluids is to investigate the effect of a weak null singularity on matter. This is also a gateway to the study of WNS spacetimes formed by collapsing matter. A good start for this investigation is the behaviour of perfect fluids in fixed WNS spacetimes. A subsequent stage would involve the coupling of matter to gravity in a WNS spacetime (see e.g. \cite{Song}). Certain issues need to be ruled out: for example, even for the simplest perfect fluid model - spherically symmetric dust, the spacetime trajectories of particles can cross even on a fixed, smooth background, leading to a blow-up of the fluid energy density - see for example \cite{Joshi_Saraykar}, \cite{Szekeres_Lun} and \cite{Fr_Klein}.
\subsection{Rough outline of results}
Heuristic versions of the main results of this paper are given below. The results are illustrated in Figure \ref{Rough_Illu}. 
\begin{thm}\label{master_theorem_1_rough}{\textbf{(Behaviour of spherically symmetric dust near the weak null singularity, rough version.)}} 
Let $\lambda$ be a radial smooth spacelike curve in the interior of a (fixed background\footnote{Meaning the backreaction effect of the fluid is not considered.}) spherically symmetric black hole spacetime with a weak null singularity\footnote{Precise assumptions on the spacetime are provided in Section \ref{preliminaries}, but the motivating example is the black hole of spherically symmetric scalar field perturbations of subextremal Reissner-Nordstr\"{o}m, described in  \cite{Luk_Oh} and \cite{Dafermos2012}.} $(\mathcal{R},g)$.  We assume $\lambda$ satisfies certain admissibility conditions\footnote{These are given in detail in Definition \ref{define_admissible_curve}. Roughly speaking, we require that $\lambda$ is compactly contained in a region near the Cauchy horizon, it is bounded away from null, and its curvature is bounded. There is also a condition on the position of  $\lambda$ for two reasons: firstly to ensure that every inextendible radial future-directed timelike geodesic normal to $\lambda$ terminates on the Cauchy horizon bounding the black hole interior to the future; secondly - to ensure that the variational field along the geodesics of $\Gamma$ stays bounded and bounded away from zero.}. Let $\Gamma$ be a variation of inextendible future-directed radial normalized timelike geodesics, all of which start on the admissible curve $\lambda$, and are initially normal to $\lambda$.\\[5pt]
Then there is no shell-crossing between the geodesics of $\Gamma$. 
Furthermore, a perfect pressureless fluid whose fluid particles follow the trajectories of the geodesics of $\Gamma$ has uniformly bounded energy density everywhere in the support of the fluid and the fluid velocity remains timelike and bounded away from null. In particular, this energy density does not blow up at the Cauchy horizon. 
\end{thm}

\begin{thm}\label{master_theorem_2_rough}\textbf{(Behaviour of spherically symmetric stiff fluid near the weak null singularity, rough version.)}
    Consider the characteristic initial value problem for a spherically symmetric stiff fluid in a (fixed background) weak null singularity spacetime $(\mathcal{R}, g)$, satisfying the Assumptions\footnote{The Assumptions are stated in Section \ref{preliminaries}.}. Assume the initial data for the energy density and fluid velocity is supplied on a bifurcate null hypersurface in $(\mathcal{R}, g)$ such that the closure of the outgoing null hypersurface intersects the weak null singularity at a 2-sphere. Then the fluid energy density and the ingoing null component of the fluid velocity go to infinity at the singularity while the outgoing null component goes to zero. 
\end{thm}
 
\begin{figure}[h!]
\centering
\begin{subfigure}{0.45\textwidth}
    \includegraphics[width=\textwidth]{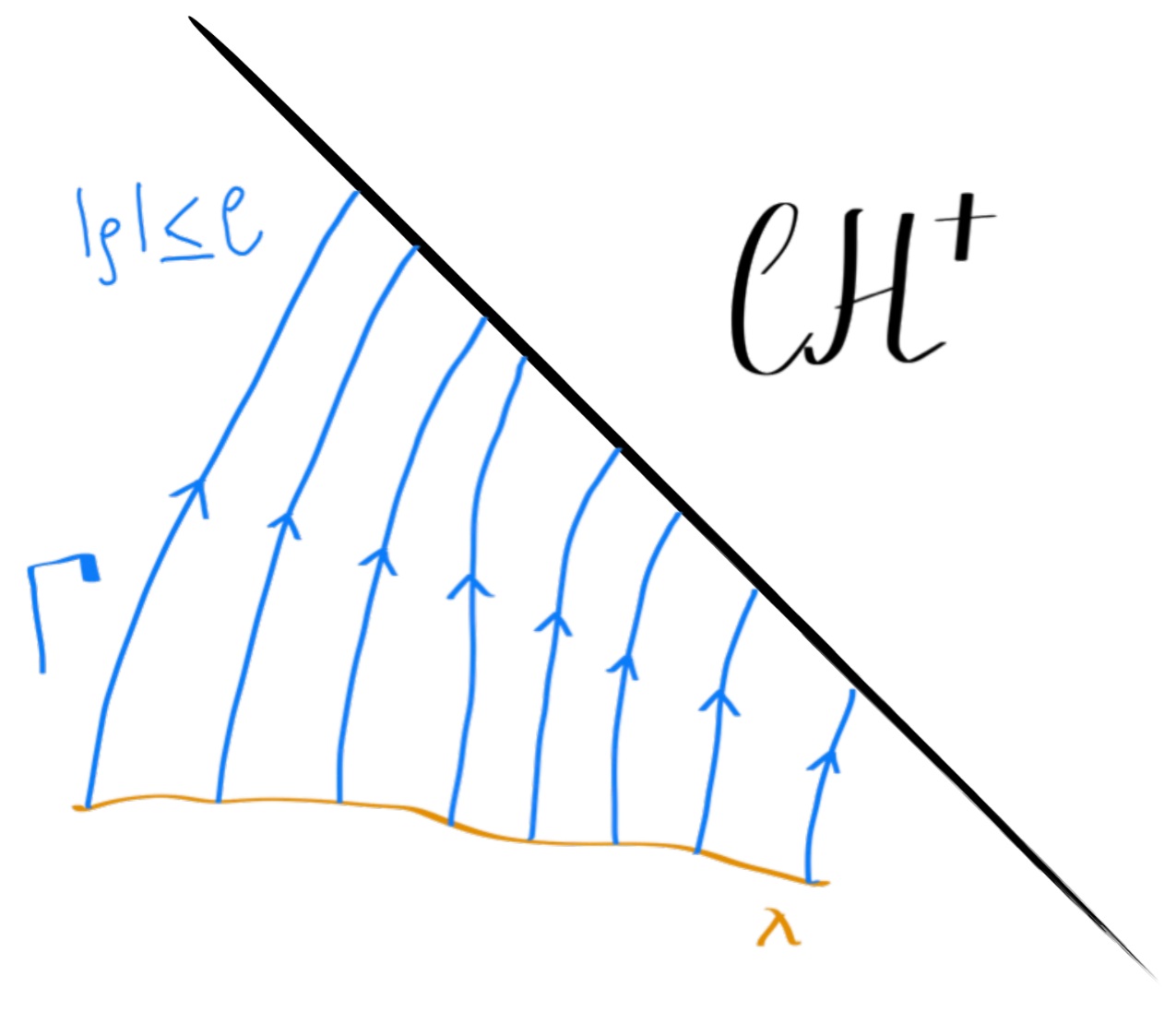} 
    \label{fig:a}
\end{subfigure}
\hfill 
\begin{subfigure}{0.45\textwidth}
    \includegraphics[width=\textwidth]{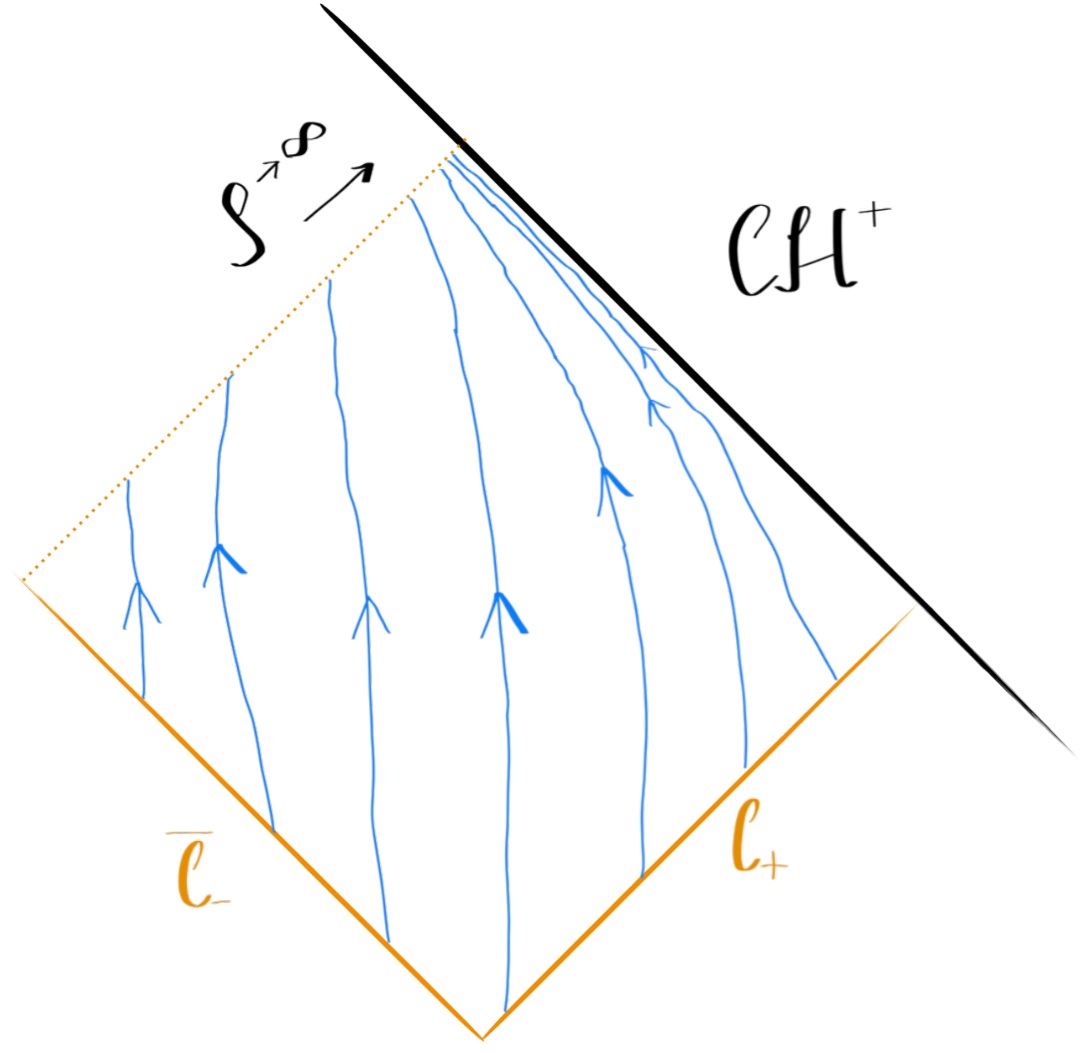} 
    \label{fig:b}
\end{subfigure}
\caption{Left: Illustration of the statement Theorem 1. The flow lines of dust, shown in blue, are initially normal to the spacelike curve $\lambda$ and extend to the singularity as timelike curves. The energy density of dust remains bounded. Right: Illustration of the statement of Theorem 2. The flow lines of the stiff fluid, in blue, emerge as timelike curves from the bifurcate null hypersurface. However, they approach ingoing null curves in the vicinity of the singularity. The energy density of the stiff fluid blows up.} 
\label{Rough_Illu}
\end{figure}

\subsection{Structure of the paper}
The structure of the paper is as follows: In Section \ref{preliminaries} we introduce the class of spacetimes considered in this paper. That is, we describe the key properties of a spherically symmetric black hole interior spacetime possessing a weak null singularity. These are stated in a double null coordinate chart regular in the vicinity of the singularity\footnote{In the sense that the coordinates extend to the WNS.} and we use this chart throughout the paper. The assumptions are roughly divided into two types: 'Stability Assumptions' and 'Instability Assumptions'. The latter ensure blow-up of the Christoffel symbols at a future null boundary of the spacetime, implying that the boundary is singular. The former constitute bounds in $L^1$ of certain geometric quantities, ensuring that the singularity is 'weak'. 
Next, we recall standard results on irrotational perfect fluids. In particular, we derive the energy-momentum conservation equations for dust and stiff fluid in a general spacetime. \\[5pt]  
In Section \ref{master_theorem_section} we state our main results in a precise way after introducing some preliminary definitions for the initial value problems considered in this paper. \\[5pt]
 Section \ref{dust_section} is dedicated to proving Main Theorem \ref{master_theorem_1_rough} (behaviour of dust near the WNS). This boils down to the following set of main steps, which rely heavily on the assumptions of subsection \ref{all_assumptions}. First we prove that future-directed timelike geodesics approaching the weak null singularity have bounded and bounded away from zero velocity components at the Cauchy horizon in double null coordinates, implying that radial timelike geodesics are extendible to the Cauchy horizon. Next we apply the stability assumptions on the spacetime to show that the Jacobi field\footnote{More precisely, the component of the Jacobi field which is transversal to the geodesics. This is the only physically significant component.} along the geodesics in the variation $\Gamma$ remains bounded and bounded away from zero for every value of the affine parameter. This allows us to prove certain expected, but non-obvious results about the total spacetime length (proper time) of the geodesics of the variation - namely that it varies monotonically and continuously with the parameter transverse to the geodesics. This is a subtle but important result which we use, together with the estimates of the variational field, to establish that no shell-crossing occurs in the variation of future-directed radial inextendible timelike geodesics. A variation of geodesics of this sort follows the trajectories of fluid particles in perfect pressureless fluid (dust). As a consequence of the energy-momentum conservation equations, the trajectories of dust are timelike geodesics and the transport equation is satisfied by the energy density and fluid velocity field for this matter model. The details of this derivation are in Section \ref{preliminaries}. We finalise Section \ref{dust_section} by proving, through analysis of the transport equation, that the energy density of the radial dust in this initial value problem stays bounded at the singularity.\\[5pt]
Section \ref{wave_section} is dedicated to proving Main Theorem \ref{master_theorem_2_rough} (behaviour of stiff fluid near the WNS). The method of proof begins by relating the stiff fluid four-velocity and energy density to a scalar field $\psi$ satisfying the linear homogeneous wave equation. This is a consequence of energy-momentum conservation for the stiff fluid and the process of obtaining it is outlined in detail in Section \ref{preliminaries}. Translating the initial data for the stiff fluid IVP into characteristic initial data for the wave equation, we find that the future-directedness of the fluid velocity is equivalent to both double null derivatives of $\psi$ being negative on the initial bifurcate null hypersurface. Through a bootstrap argument, we prove that this monotonicity is propagated throughout the domain of dependence of the characteristic data. Similar techniques are used in Dafermos \cite{Dafermos03} and Fournodavlos-Sbierski \cite{Sbierski_Fournodavlos}. The monotonicity result is in line with the physical expectation that the fluid velocity stays timelike and future-directed. It is important to note here that this monotonicity, as well as the duality between the stiff fluid and the wave equation, fails to hold in a generic background. An example for a breakdown of the duality between the stiff fluid and wave equation in finite time is given in \cite{Gavassino} for a Cauchy problem on Minkowski background. Using the assumptions listed in subsection \ref{all_assumptions}. on the spacetime, we then prove that the outgoing derivative of the scalar field blows up at the singularity but $\psi$ itself stays bounded in the domain where it is defined. When we translate these results back into the stiff fluid language, they imply that the energy density and the ingoing velocity component of the stiff fluid go to infinity at the singularity while the outgoing velocity component goes to zero. 
\\[5pt]
{In Section \ref{EMSf_chapter}, we show that a codimension-0-submanifold of the black hole interior of a spherically symmetric black hole spacetime arising from the EMSF system analysed in \cite{Luk_Oh} satisfies the geometric assumptions listed in Section \ref{preliminaries} to be an eligible spacetime where the results of this paper hold. The Stability Assumptions \ref{assumptions_on_L1_norms} on the spacetime follow from the Theorem on Stability of the Cauchy Horizon - Theorem 5.1 in \cite{Luk_Oh}. This result ensures that if the characteristic initial data for the Einstein--Maxwell--scalar field system in spherical symmetry is close to perfect Reissner-Nordstr\"{o}m data, then the solution remains close to perfect Reissner-Nordstr\"{o}m. On the other hand, the Instability Assumptions \ref{instability_assumptions} follow from the Theorem on Instability of the Cauchy Horizon - Theorem 5.5 in \cite{Luk_Oh}. With the completion of Section \ref{EMSf_chapter}, we demonstrate that there is a wide class of spacetimes on which the results of this paper apply.}

\section{Preliminaries} \label{preliminaries}

\subsection{Assumptions on the background spacetime}\label{all_assumptions}
In this subsection we describe the topological and geometrical properties of the class of Lorentzian manifolds on which the results in the rest of this paper hold.
\subsubsection{General assumptions}
 We begin by introducing the manifold  $\mathcal{R}$ with topology
\begin{align}
    \mathcal{R} = \mathcal{Q_R}\times S^2; \ \ \overline{\mathcal{R
}} = \overline{\mathcal{Q_R}}\times S^2
\end{align}
Here $\overline{\mathcal{R}}$ is to be thought of as a 4-manifold-with-boundary such that the interior of this manifold is $\mathcal{R}$. The manifold is equipped with the spherically symmetric Lorentzian metric
\begin{align}
    g = g_{\mathcal{Q_R}} + r^2 d\Omega_2^2  = -e^{2\omega} dudv + r^2(u,v) d\Omega_2^2 \label{ssmetric_general}
\end{align}
in double null coordinates $(u,v, \theta,\varphi)$. The first term in (\ref{ssmetric_general}) is the metric on the quotient $1+1$-dimensional Lorentzian manifold $(\mathcal{Q_R}, g_{\mathcal{Q_R}})$ and $d\Omega_2^2$ is the unit metric on $S^2$. The Levi-Civita connection of $g$ will be denoted by $\nabla$ and the Levi-Civita connection of $g_{\mathcal{Q_R}}$ will be denoted by $\mathring{\nabla}$. The quotient manifold is parameterised by $(u,v)$, while $(\theta, \varphi)$ are standard spherical coordinates parameterising the unit 2-spheres. In this chart, $\mathcal{Q_R}$ is given by 
\begin{align}
    \mathcal{Q_R} = \{(u,v): -\infty<u<u_s,\  v_\kappa<v<0,\  uv<u_s v_\kappa\}
\end{align}
And its closure $\overline{\mathcal{Q_R}}$, viewed as a 2-manifold-with-boundary, is given by
\begin{align}
     \overline{\mathcal{Q_R}} = \{(u,v): -\infty<u\leq u_s,\  v_\kappa\leq v\leq 0, \ uv\leq u_s v_\kappa\}
\end{align}
Where $u_s<0$ and $v_\kappa<0$ are constants\footnote{These seem like strange names for the coordinate values parameterising the boundary of the spacetime. The motivation behind this choice of notation is that there is an open subset of the spherically symmetric spacetime considered in \cite{Luk_Oh}, defined in double null coordinates by the condition $uv<u_s v_\kappa$, such that the Assumptions are satisfied in this submanifold. This will be explained in detail in {Section \ref{EMSf_chapter}}.}. We call $e^\omega:\mathcal{Q_R}\to \mathbb{R}^+$ the lapse function and 
 $r: \mathcal{Q_R} \to [0, \infty)$ is the area radius function on $\mathcal{R}$, given by $(u,v) \mapsto \sqrt{\text{Area}\big(S^2(u,v)\big)/(4\pi)}$ where $S^2_p$ is the 2-sphere in $\mathcal{R}$ attached to the point $(u,v)$ in $\mathcal{Q_R}$. The latter tells us the area of the 2-sphere attached to the point $(u,v)\in \mathcal{Q_R}$. We will denote by $\Sigma$ the past boundary of $\mathcal{R}$, the spacelike hypersurface defined by the condition $uv =u_s v_\kappa$. Since $D^+(\Sigma)/\Sigma = \mathcal{R}\cup\{u=u_s\}$, the singular boundary $\{v=0\}$ is a Cauchy horizon for $\mathcal{R}$ which we will denote by $\mathcal{CH}^+$\footnote{This notation is motivated by the nature of the Cauchy horizon as a weak null singularity in spacetimes arising from small spherically symmetric scalar field perturbations of subextremal Reissner-Nordstr\"{o}m, \cite{Luk_Oh} \cite{Dafermos2012}}.

 \begin{rem}
A straightforward calculation using the Christoffel symbols listed in Appendix \ref{Christoffel_Symbols} shows that in double null coordinates $(\mathring{\nabla}_{\mathring{X}}\mathring{Y})^\mu=(\nabla_{X}Y)^\mu$ holds for any pair of smooth vector fields $X,Y$ on $\mathcal{R}$ such that $X^\theta = X^\varphi =Y^\theta= Y^\varphi =0$, where $\mathring{X}, \mathring{Y}$ denote the pushforwards of $X, Y$ onto $\mathcal{Q_R}$ via the projection operator $\mathcal{R}\to \mathcal{Q_R},p\mapsto (u(p),v(p))$ . In coordinates,  $X^u=\mathring{X}^u, Y^u=\mathring{Y}^u$ and $X^v=\mathring{X}^v, Y^v=\mathring{Y}^v$ at every point in $\mathcal{Q_R}$. In particular, this implies that for radial geodesics $\gamma$ in the $3+1$-dimensional $\mathcal{R}$, i.e. those with $\gamma^\theta, \gamma^\varphi=$const, the geodesic equation is simply the one for the projected geodesic $\mathring{\gamma}$ in the $1+1$-dimensional $\mathcal{Q_R}$.
 \end{rem}
 \begin{asmpt} \label{general_assumptions}
The metric functions for this spacetime in these coordinates are assumed to satisfy the following conditions, which are also the ones used in the setting of \cite{Sbierski2021}, Section 4:
   \begin{itemize}
           \item $\omega\in C^\infty(\mathcal{Q_R}\to \mathbb{R})$ extends continuously to $\overline{\mathcal{Q_R}}$.
           \item $r \in C^{\infty}(\mathcal{Q_R} \to (0, \infty))$ extends continuously as a positive function to $\overline{\mathcal{Q_R}}$. 
           \item $(\mathcal{R},g)$ is time-oriented, with time orientation fixed by stipulating that $\partial_u + \partial_v$ is future-directed. 
   \end{itemize}
 \end{asmpt}

\begin{figure}[H]
{\begin{center}
    \includegraphics[width=0.45\textwidth]{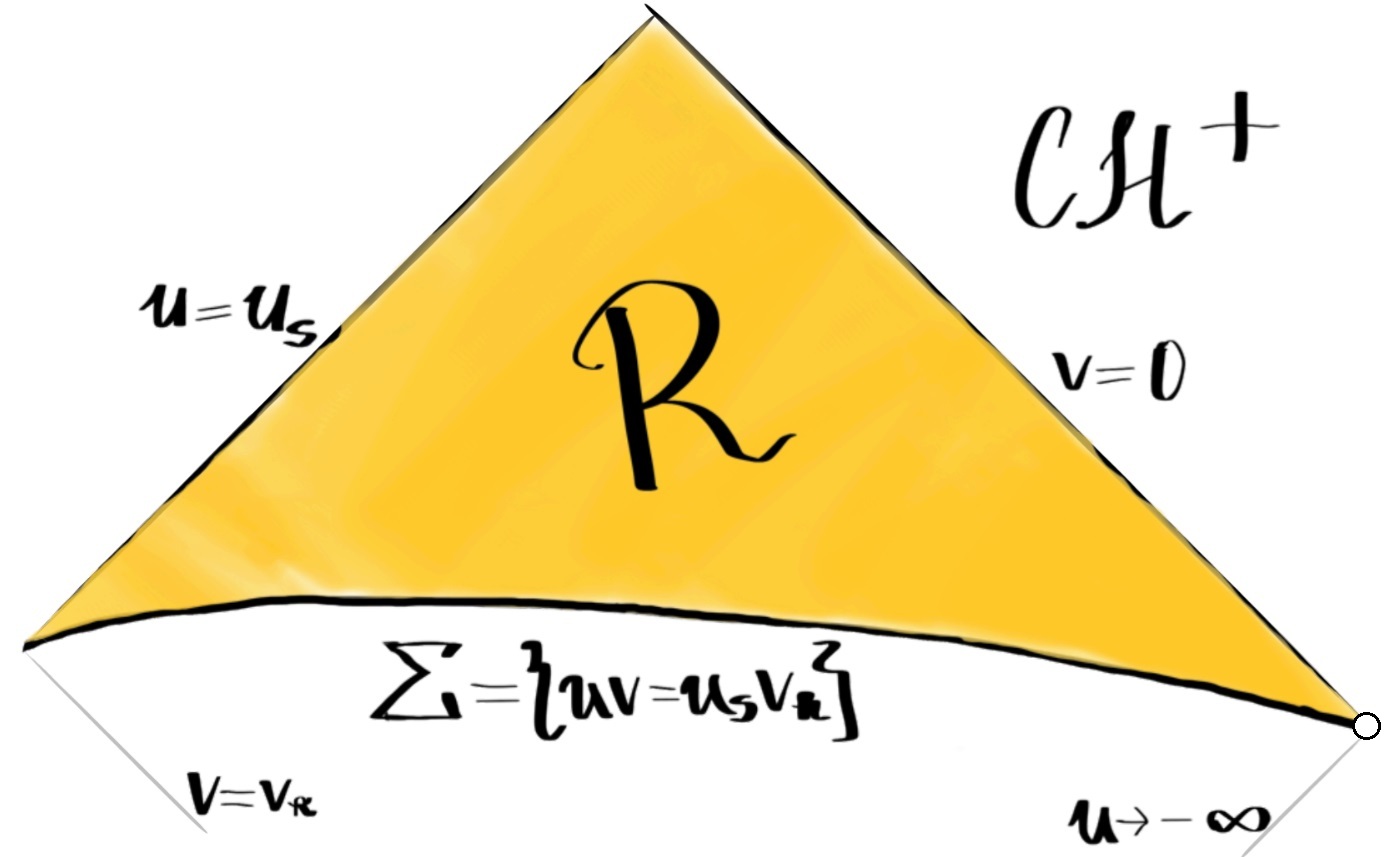}
    \caption{Penrose diagram illustrating the spacetime $\mathcal{R}$, shown in yellow. Note that $\mathcal{R}$ coincides with the domain of dependence of the spacelike hypersurface $\Sigma =\{uv=u_{s}v_{\kappa}\}$.}
    \label{restriction_to_R}
\end{center}}
\end{figure}

\subsubsection{Assumptions on spacetime quantities needed for Theorem \ref{master_theorem_1_rough} (dust behaviour)}
 In this section we state the assumptions for the $L^1$ norms of geometric quantities which are crucial to establish the proof of the first main result of this paper, Theorem \ref{master_theorem}. 
 \begin{asmpt}\label{assumptions_on_L1_norms}
 There is a (big) constant $C_g>1$ which depends only on the global geometry of the Lorentzian manifold $(\mathcal{R}, g)$, such that 
 for every $u_1<u_2$ and $v_1<v_2$ with $(u_i, v_j)\in \mathcal{Q}_\mathcal{R}$, the following $L^1$ bounds in the $(u,v)$ double null coordinate system hold
    \begin{align}
        &\norm{e^{\pm 2\omega}}_{C(\mathcal{R})} + \norm{r^{-1}}_{C(\mathcal{R})}\leq C_g\label{geom_bound_as1}\\   
        &\norm{\omega_{,u}(\cdot, v_j)}_{L^1[u_1, u_2]} + \norm{r_{,u}(\cdot, v_j)}_{L^1[u_1,u_2]} \leq C_g \label{geom_bound_as2}\\
        &\norm{\omega_{,v}(u_i,\cdot)}_{L^1[v_1, v_2]} + \norm{r_{,v}(u_i,\cdot)}_{L^1[v_1, v_2]} \leq C_g \label{geom_bound_as3}\\
        &\norm{\omega_{,uv}(u_i,\cdot)}_{L^1[v_1,v_2]}\leq C_g (1+|v_2-v_1|) \label{geom_bound_as4}
    \end{align}
 \end{asmpt}
     \begin{rem}
        We can uniformly bound what is in the parentheses in (\ref{geom_bound_as4}) by $1+|v_\kappa|$, hence 
        \begin{align}
            \norm{\omega_{,uv}(u_i, \cdot)}_{L^1[v_1,v_2]}\leq C_g (1+|v_{\kappa}|)
        \end{align}
    \end{rem}
         
\subsubsection{Assumptions on the spacetime quantities needed for Theorem \ref{master_theorem_2_rough}, (stiff fluid behaviour)}
In this section we introduce the blow-up assumptions on the Cauchy horizon $\mathcal{CH}^+=\{v=0\}$ which are essential to proving blow-up for this matter model and thus establishing Theorem \ref{masterTheorem_2}. The blow-up assumptions are equivalent to those in the setting of Section 4 of \cite{Sbierski2021}.
\begin{asmpt}\label{instability_assumptions}
For the null derivatives of the area radius function we assume:
\begin{itemize}
  \item $r_{,v}(u,v) \to -\infty$ as $v\to 0^-$ for all $u\in(-\infty,u_s)$
\item For every $u_0\in(-\infty,u_s)$, there is a $v_0=v_0(u_0)<0$ and a $\varepsilon>0$ so that $r_{,u}<0$ and $r_{,v}<0$ in $[u_0-\varepsilon,u_0+\varepsilon]\times[v_0,0)$. 
\item $\norm{r_{,u}}_{C(\mathcal{R})}<\infty$
\end{itemize}
\end{asmpt}
\begin{rems}
     \begin{enumerate}
         \item The first condition above, also featured in Section 4 of \cite{Sbierski2021}, implies that the surface $\mathcal{CH}^+=\{v=0\}$ is singular. In view of the fact that the Christoffel symbols contain $\partial_v r$, they blow up as $v\to 0$, even though the metric functions extend continuously. In the work of Sbierski \cite{Sbierski2021} it is shown in this setting that, as a consequence of the Christoffel symbols blowing up as $v\to 0$, there exists no $C^{0,1}_{loc}$-extension of $(\mathcal{R},g)$ in any neighbourhood of the future null boundary hypersurface $\{v=0\}$. We shall call a future null boundary component of the spacetime to which the metric extends continuously but not in $C^2$, a \textbf{weak null singularity} (WNS) of the spacetime. Hereafter we will occasionally refer to $\mathcal{CH}^+=\{v=0\}$ as the WNS.
         \item Physically, the second condition says that there exists a trapped region in $\mathcal{R}$. The precise wording of the condition assumed in Section 4 of
         \cite{Sbierski2021} is that for each $u_0<u_s$ there exists a $v_0'=v_0'(u_0)$ such that $v>v_0' \implies r_{,u}(u_0,v)<0$. On the other hand the first condition of Assumptions \ref{instability_assumptions} - the blow-up of $\partial_v r$ - implies that for every $u_0<u_s$ there exists a $v_0''=v_0''(u_0)$ such that $v>v_0''\implies\partial_v r(u_0,v)<0$. These two conditions combined ensure the existence of a thin trapped strip around a given $u_0$.
         \item The uniform bound on the ingoing derivative $\partial_u r$, although not essential for proving energy density blow-up for the stiff fluid, is used to prove that one component of the velocity blows up and the other goes to zero. If that bound did not hold, it may well be possible that the double null components of the fluid velocity stay bounded.
     \end{enumerate} 
\end{rems}

\subsection{Irrotational barotropic perfect fluids}
In this section we recall textbook properties of the two models of irrotational perfect fluid we investigate in this paper. We begin with more general analysis of barotropic perfect fluids.\\[5pt]
Let $\rho\in C^\infty(\mathcal{R})$ and let $U^a$ be a normalised future-directed smooth vector field on $(\mathcal{R},g)$. The \textbf{perfect fluid energy momentum tensor} with energy density function $\rho$ and four-velocity $U^a$ in covariant form is given by 
\begin{align}
    T_{ab}^{[\text{pf}]} = (p+\rho) U_a U_b + p g_{ab} \label{generic_stressenergy_tensor}
\end{align}
Where the \textbf{fluid pressure} $p=p(\rho,n)$ is generally a function of the energy density $\rho$ and the \textbf{particle density} $n$ - a smooth scalar function on $\mathcal{R}$. The relation $p=p(\rho,n)$ is called the equation of state. In this paper we will consider two models of \textbf{barotropic} perfect fluid, meaning that the pressure is a function of the energy density only, $p=p(\rho)$.\\[5pt]
We require that the energy-momentum tensor is conserved:
\begin{align}
    \nabla^bT_{ab}^{[\text{pf}]} = 0 \iff U_a \nabla_U (p+\rho) + (p+\rho) \nabla_U U_a + (p+\rho) U_a \nabla\cdot U + \nabla_a p = 0 \label{generic_stressenergy_conservation}
\end{align}
The component of (\ref{generic_stressenergy_conservation}) along $U^a$ is 
\begin{align}
    U^a \nabla^b T_{ab}^{\text{[pf]}} = 0 \iff - \nabla_U (p+\rho) - (p+\rho) \nabla\cdot U + \nabla_U p = 0 \label{generic_stressenergy_cons_parallel_compt}
\end{align}
Let $\perp^{ab} = g^{ab} + U^a U^b$. The component of (\ref{generic_stressenergy_conservation}) normal to $U^a$ is 
\begin{align}
    \perp^{ab} \nabla^c T_{bc}^{[\text{pf}]} = 0 \iff (p+\rho) \nabla_U U^a + \nabla^a p + U^a \nabla_U p = 0 \label{generic_stressenergy_cons_normal_compts}
\end{align}
These are the relativistic Euler equations.
 In this paper we will consider \textbf{irrotational} fluids. In a relativistic context, setting $(dU)_{ab} = \nabla_{[a} U_{b]} = 0$ is a too strong condition, as outlined in \cite{acoustic_spacetimes}. The appropriate relativistic condition (see \cite{Gavassino}, \cite{acoustic_spacetimes} for details of the reasons) for irrotationality is 
\begin{align}
    (U\wedge dU)_{abc} = 0 \iff U_{[a} \nabla_b U_{c]} = 0
\end{align}
By setting the 3-form $(U\wedge dU)_{abc}$ to zero, in the local rest frame of the fluid where $U^0 = 1, U^i=0$ this means $\partial_{[i} U_{j]} = 0$, as outlined also in \cite{acoustic_spacetimes}. By the Frobenius theorem, \textit{locally} there exist functions $\psi,h$ such that $U_a = \frac{1}{h}\nabla_a \psi$. Combining with the condition that $U$ is normalised, 
\begin{align}
    U^a = \frac{\nabla^a \psi}{\sqrt{|\langle\nabla \psi, \nabla \psi\rangle}|} \implies h^2 = -g^{ab} \nabla_a \psi \nabla_b \psi \label{U_intermsof_psi}
\end{align}
This is the relativistic condition for irrotational flow. The relativistic norm of $\nabla_a \psi$ is negative, because the velocity of the fluid is a timelike vector field. Notice that changing the \textit{scalar potential} $\psi\to f(\psi)$ where $f$ is a $C^2$ function with $f'\neq 0$\footnote{If $f'=0$ at a point $\psi_0$, the RHS of the first equality of (\ref{U_intermsof_psi}) has a \textit{removable singularity} at $\psi_0$. By removing the singularity, $U^a$ remains invariant. } does not affect the fluid velocity $U^a$. We will refer to this property as the \textit{gauge freedom} of the scalar potential. \\[5pt]
Let us define the vector field $V^a = hU^a = \nabla^a \psi$. So $V_a$ is a perfect gradient and $\nabla_{[a} V_{b]}=0$. Below we fix what the freedom in choosing the scalar field $\psi$ separately for the dust case versus the scalar field case. 

\subsubsection{Case I: Dust, $p=0$}
The energy-momentum tensor of pressureless perfect fluid (dust) is
\begin{align}
    T_{ab}^{[\text{dust}]} = \rho U_a U_b = \rho \partial_a \psi \partial_b \psi 
\end{align} 
The energy-momentum conservation equations for dust are: 
\begin{align}
        &\nabla_U(\rho) + \rho \nabla \cdot U =0 \ \ \ \ \ \ \text{(transport equation)}\label{dust_transport_equation}\\
        &\nabla_U U^a = 0 \ \ \ \ \ \ \ \ \text{(geodesic equation)} \label{dust_geodesic_equation}
\end{align}
In view of the geodesic equation, if we fix $\psi$ so that $h^2 = -g^{ab} \nabla_a \psi \nabla_b \psi = 1$ at a point $p\in \mathcal{R}$, then $h=1$ along the future-directed timelike geodesic describing the trajectory of the dust particle crossing through $p$. Therefore, we fix $h=1$ and $V^a = U^a$. 

\subsubsection{Case II: Stiff fluid, $p=\rho$}\label{stiff_fluid_to_wave_equation} The energy-momentum tensor for the stiff perfect fluid is 
\begin{align}
    T_{ab}^{[\text{stf}]} = 2\rho U_a U_b + \rho g_{ab}
\end{align}
We want to use the gauge freedom of the scalar field $\psi$ to set $h = \sqrt{-\langle \nabla\psi,\nabla\psi\rangle} = \sqrt{\rho}$, however it is not obvious that this is possible. We argue that it is possible to make this transformation in view of the Euler equations (\ref{generic_stressenergy_cons_normal_compts}). The argument below for fixing the gauge freedom is the one outlined in \cite{acoustic_spacetimes}. As $U^a = \nabla^a\psi / h$, 
\begin{align}
    \nabla_U U^a = \frac{1}{|h|}U^b(g^{ac} + U^a U^c) \nabla_c \nabla_b \psi= -\frac{1}{2h^2} \perp^{ac}\nabla_c(h^2) = -\perp^{ac}\nabla_c(\ln{|h|}) \label{4acceleration_formula_by_defn}
\end{align}
On the other hand, the Euler equations (\ref{generic_stressenergy_cons_normal_compts}) with $p(\rho)=\rho$ are equivalent to
\begin{align}
    \nabla_U U^a  = \frac{1}{2\rho}\big(-\nabla^a \rho - U^a\nabla_U \rho\big) = -\perp^{ab}\nabla_b (\ln \sqrt{\rho}) \label{4acceleration_formula_by_euler}
\end{align}
Comparing (\ref{4acceleration_formula_by_defn}-\ref{4acceleration_formula_by_euler}) and noticing that the function in the parentheses below is constant on the level sets of $\psi$,
\begin{align}
    \perp^{ab}\nabla_b \bigg( \ln \frac{|h|}{\sqrt{\rho}} \bigg)=0 \implies \frac{|h|}{\sqrt{\rho}} = e^{f(\psi)} \label{to_bernoulli}
\end{align}
where $f$ is an 'arbitrary function of integration'; more precisely, notice that since $\perp^{ab}\nabla_b f(\psi)$=0 for any differentiable function $f$, we can add this function inside the parentheses in the first equality of (\ref{to_bernoulli}). In relativistic fluid dynamics the second equation of (\ref{to_bernoulli}) is the Bernoulli equation. Now we use the gauge freedom of the scalar field $\psi$ to make the transformation 
\begin{align}
    &\psi \to F(\psi)= \int_{0}^\psi e^{-f(\overline{\psi})}d\overline{\psi} \\
    &\implies \nabla F(\psi) = e^{-f(\psi)}\nabla\psi
    \end{align}
    The transformation yields a corresponding transformation of $h$,
    \begin{align}
     h^2 \to h_F^2= {-g^{ab}\nabla_a F(\psi) \nabla_b F(\psi)} = -e^{-2f(\psi)} g^{ab}\nabla_a\psi \nabla_b \psi =e^{-2f(\psi)} h^2 
\end{align}
With this change, Bernoulli equation (\ref{to_bernoulli}) simplifies to 
\begin{align}
    |h_F| = \sqrt{\rho}
\end{align}
We change notation back to $h,\psi$ and choose $h>0$ so that $U^a$ is future directed if and only if $\nabla^a \psi$ is future directed. 
With this choice the energy-momentum conservation equations for the stiff fluid in terms of $V^a = \sqrt{\rho}U^a$ become 
\begin{align}
    &\nabla\cdot V = 0 \label{V_is_divergence_free}\\
    &\partial^a \rho +2 \nabla_V V^a = 0 \label{EM_conservation_normal_V}
\end{align}
Written in terms of $\psi$, equations (\ref{V_is_divergence_free} - \ref{EM_conservation_normal_V})  look like
\begin{align}
&\Box_g \psi = 0\ \ \ \ \ \ \ \ \ \ \ \text{(linear wave equation for $\psi$)} \label{lner_wave_eqn_psi}\\
& \partial^a \rho +  \nabla^a (\partial_b\psi \partial^b \psi) = 0\ \ \ \ \ \ \ \ \ \ \text{(gauge choice)}\label{EM_cons_normal_intermsof_psi}
\end{align}
Equations (\ref{EM_conservation_normal_V}) and (\ref{EM_cons_normal_intermsof_psi}) are derived from the Euler equations for the stiff fluid. They simply recover the gauge choice $h^2 = \rho$. Equation (\ref{lner_wave_eqn_psi}) is the one we will study in Section \ref{wave_section}. 
\begin{rem} \textit{(Aside, not used in the rest of this article).}
    As the fluid is irrotational, $(dV)_{ab} = 0 \iff \nabla_{[a} V_{b]} = 0$. Therefore,
\begin{align}
    \nabla_{[a} U_{b]} = \nabla_{[a} {\rho}^{-1/2} V_{b]} = {\rho}^{-1/2} \nabla_{[a} V_{b]} + V_{[b} \nabla_{a]} ({\rho}^{-1/2}) = 0 + \frac{1}{2}\bigg( V_b \nabla_a ({\rho}^{-1/2}) - V_a \nabla_b ({\rho}^{-1/2}) \bigg) \label{irrot_of_V_1}
\end{align}
Since $U_aU^a=-1$, we have
\begin{align}
    &\nabla_a \rho = -2 \sqrt{\rho} U^b \nabla_b (\sqrt{\rho} U_a) = - 2U_a \nabla_U \sqrt{\rho} - \sqrt{\rho}\nabla_a(U_b U^b) = - 2 (\nabla_U\sqrt{\rho} ) U_a \implies \\
    & \nabla_a ({\rho}^{-1/2}) = \frac{\nabla_U \sqrt{\rho}}{\rho^2} V_a =: f V_a \label{gradient_of_rho_propto_U}
\end{align}
where we used that $\langle U, U\rangle = -1$ and $f :=  {(\nabla_U \sqrt{\rho})/\rho^2}$. Hence, (\ref{gradient_of_rho_propto_U}) tells us that the gradient of the energy density is proportional to  the fluid velocity. Substituting this in (\ref{irrot_of_V_1}), we find
\begin{align}
    2\nabla_{[a} U_{b]} = V_b f V_a - V_a f V_b = 0  
\end{align}
Hence, $U$ also has vanishing exterior derivative: $(dU)_{ab} = 0$. This implies that there exists a scalar field $\phi \in C^{\infty}(\mathcal{R})$ such that $\nabla_a\phi = U_a$.
\end{rem}

\section{Main Theorems} \label{master_theorem_section}
In this section we state a comprehensive version of the main results of this paper. In the next two definitions we assume $(\mathcal{R}, g)$ is the smooth, time-oriented spherically symmetric Lorentzian manifold defined in the previous section. 

\subsection{Preliminary definitions for dust regularity}
\begin{defn}\label{1parfamily_of_geodesics} A \textbf{variation} of geodesics on $\mathcal{Q_R}$ is a smooth map 
\begin{align}
    \Gamma:\ & \mathbb{R}^2 \supseteq S \to  \mathcal{Q_R}, \nonumber\\
           &(\alpha,s)\mapsto \gamma(\alpha,s) \nonumber
\end{align}
such that for each fixed $\alpha$, the curve $\gamma(\alpha, \cdot)$ is a geodesic in $\mathcal{Q_R}$.
\end{defn}

\begin{defn}{\label{define_admissible_curve}}
Let $a>0$ and $\lambda:(-a,a) \to \mathcal{Q_R}$ be a smooth curve. Let $0<n_0<N_0, C_K>0$ Denote 
\begin{equation} u_\text{min}:= \inf_\alpha u[\lambda(\alpha)]\in [-\infty, u_s]  ,\  u_\text{max}:= \sup_\alpha u[\lambda(\alpha)]\in [-\infty, u_s]  \hspace{5pt}  \text{  and  } \hspace{5pt} v_\text{min}:= \inf_\alpha v[\lambda(\alpha)] \in [v_\kappa, 0]
\end{equation}
We say that $\lambda$ is an \textbf{${(n_0, N_0, C_K)}$--admissible curve} if 
\begin{enumerate}
    \item $\lambda$ is a nowhere constant smooth spacelike curve segment on the quotient manifold $\mathcal{Q_R}$
    \item $u_\text{min}>-\infty$
    \item For all $\alpha\in(-a, a)$ we have
     \begin{align}
        n_0<\bigg|\frac{d\lambda^u}{d\alpha}\bigg|,\bigg|\frac{d\lambda^v}{d\alpha}\bigg| < N_0\label{condd5}
    \end{align}
     where $d\lambda/d\alpha(\alpha) =: \lambda'(\alpha)$ denotes that tangent to $\lambda$.
    \item Let $\nu(\alpha)$ denote the unique future-directed radial unit normal vector to $\lambda$ at the point $\lambda(\alpha)$\footnote{Using the Christoffel symbols from Appendix \ref{appendix_riemann_components}, it is easy to calculate \begin{align}
        \nu^u = e^{-\omega} \sqrt{(\lambda')^u/(\lambda')^v} ;\ \ \nu^v= e^{-\omega}\sqrt{(\lambda')^v/(\lambda')^u}
    \end{align}}\footnote{In terms of the components of $\lambda'$, we can calculate  \begin{align}
       &\Bigg|\frac{\big\langle \lambda', {\mathring{\nabla}}_{\lambda'}\nu \big\rangle}{\big\langle \lambda', \lambda' \big\rangle} \Bigg| = \frac{1}{\sqrt{\langle \lambda' , \lambda'\rangle}}\Bigg|(\lambda')^\mu\partial_\mu\omega+ \frac{1}{2}\partial_\mu (\lambda')^\mu + \frac{1}{2}\bigg( \frac{(\lambda')^u \partial_u (\lambda')^v}{(\lambda')^v} + \frac{(\lambda')^v \partial_v (\lambda')^u}{(\lambda')^u}\bigg)\Bigg| 
        \end{align}}. Then
   \begin{align}
       &\Bigg|\frac{\big\langle \lambda', {\mathring{\nabla}}_{\lambda'}\nu \big\rangle}{\big\langle \lambda', \lambda' \big\rangle} \Bigg| \leq C_K\ \ \ \forall \alpha\in(-a, a) \label{bound_on_curvature_of_lambda}
        \end{align}
        \item The following is satisfied by $v_\text{min}$
        \begin{align}
            -v_{\text{min}} <(u_s - u_{\text{max}})\frac{n_0}{N_0 C_g e^{4C_g}}
        \end{align} where $C_g$ is the constant defined in Section \ref{assumptions_on_L1_norms}. 
         \item Furthermore, $-v_{\text{min}} < x_*$ where $x_*$ is the unique positive solution of the transcendental equation
         \begin{align}
             A x e^{B x} = 1\label{transc_eqn}
         \end{align}
         where the parameters $A$ and $B$ are given by
         \begin{align}
             &A = \sqrt{\frac{C_g^3 (1+C_K^2) N_0}{n_0}} \exp\Bigg[ \sqrt{\frac{C_g^7 N_0}{n_0}}e^{2C_g}\Bigg] \label{A} \\
             &B =  \sqrt{\frac{C_g^3 N_0}{n_0}} e^{2C_g}\bigg(\frac{1}{2} + C_g^2 \bigg) \label{B}
         \end{align}
\end{enumerate}
\end{defn}
\begin{figure}[h]
    \centering
    \includegraphics[width=0.4\linewidth]{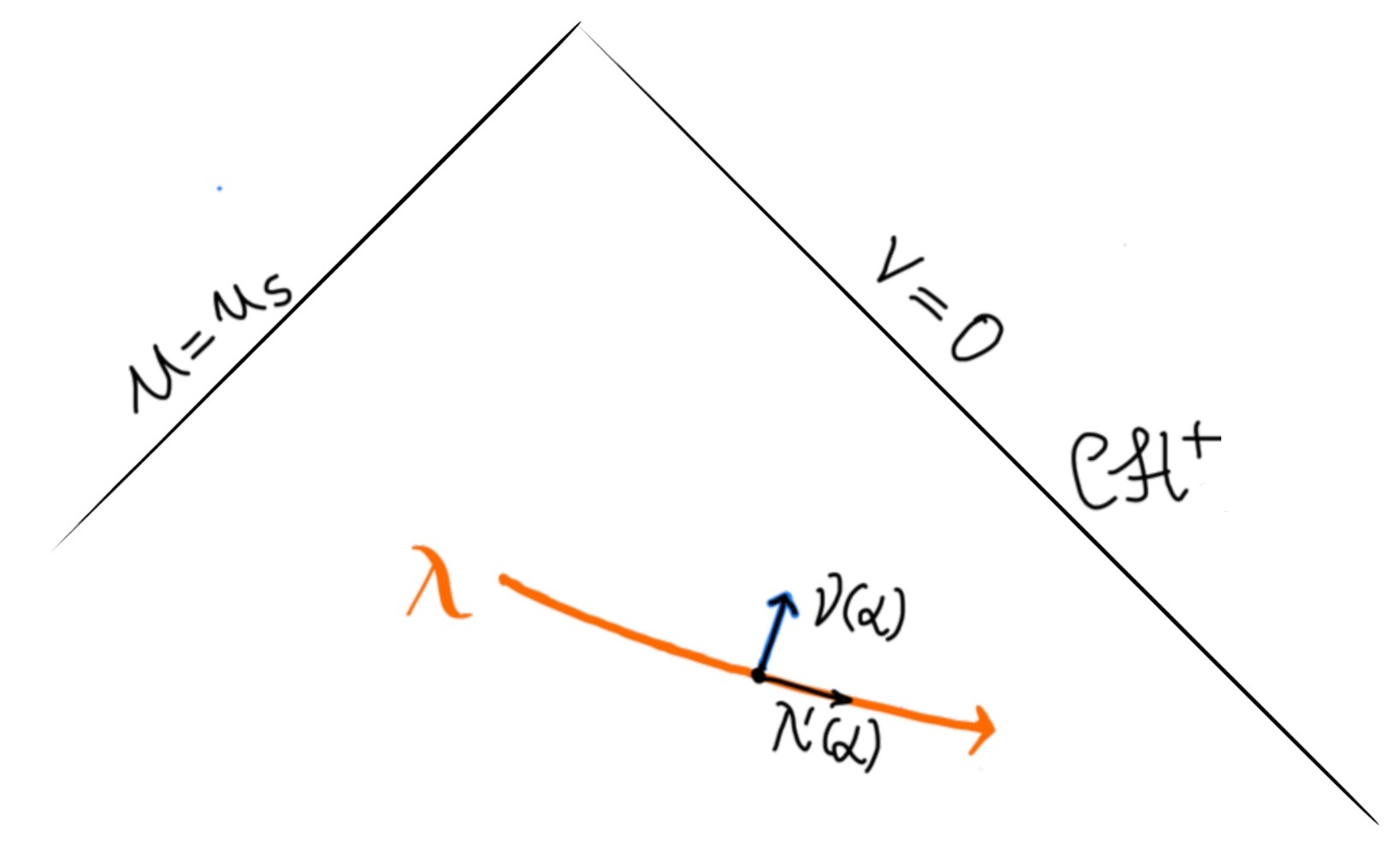}
    \caption{Admissible curve $\lambda$ with future-directed unit normal $\nu$ and tangent $\lambda'$.}
    \end{figure}
\begin{rems}
\begin{enumerate}
    \item Condition 1. ensures the tangent to $\lambda$ is nonvanishing. 
    \item Condition 2. implies that Image$(\lambda)$ is compactly supported inside $\mathcal{Q}_\mathcal{R}$. 
    \item Condition 3. means that $\lambda$ is uniformly bounded away from null. This implies that timelike geodesics locally normal to $\lambda$ are uniformly bounded away from null.
    \item The LHS of Equation (\ref{bound_on_curvature_of_lambda}) is the (only nonvanishing component of) the second fundamental form of the curve $\lambda$ as a 1-manifold, embedded in $\mathcal{Q_R}$. The covariant derivative in (\ref{bound_on_curvature_of_lambda}) is with respect to $g_{\mathcal{Q_R}}$. The condition on the curvature of lambda prevents early focusing of timelike geodesics emanating from $\lambda$. Such a condition is needed even if no singularity is present.
    \item  We will see that Condition 5 ensures that future-directed timelike geodesics emanating from the curve $\lambda$ will reach the Cauchy horizon $\{v=0\}$ rather than escape through the outgoing future null boundary $\{u=u_s\}$ of $\mathcal{Q}_\mathcal{R}$. This will be shown in Lemma \ref{lambda_sends_to_CH}.
     \item We will see that Condition 6 is essential for showing that the variational field along the trajectories of dust particles falling into the weak null singularity can be bounded away from zero. This is what makes possible to prove that no shell-crossing between these trajectories occurs. The reason for the very specific choice of the parameters $A,B$ will become clear in the proof of Theorem \ref{bounds_for_jacobi_field_thm}. 
     \item We need to show that the transcendental equation (\ref{transc_eqn}) indeed has a unique positive solution. To achieve this, note that the function $x\mapsto A x e^{B x}$ is smooth, has a unique zero at $x=0$, and is strictly increasing on $\mathbb{R}$. So the equation $Axe^{Bx}=1$ has a unique solution $x_*>0$. Furthermore, $0<-v_{\text{min}}<x_*$ implies $0<A|v_{\text{min}}|e^{B|v_\text{min}|}<1$. 
     \item Even though Conditions 5 and 6 both restrict the image of $\lambda$ to a rectangle close to the Cauchy horizon, we have deliberately separated them to emphasize their different role in the proof of the dust regularity result. 
    \end{enumerate}
\end{rems}
In Lemma \ref{existence_of_lambda} we will prove that the abovementioned conditions are compatible. That is, we will show existence of an admissible curve $\lambda$ in $(\mathcal{Q_R}, g_{\mathcal{Q_R}
})$.
\begin{defn}{\label{define_geodesic_variation_based_on_lambda}}
    Given a $(n_0, N_0, C_K)$--admissible curve $\lambda$, define \textbf{the geodesic variation based on $\lambda$} to be the map $\overline{\Gamma}$ - a variation of affine, normalised, future-directed, future-inextendible timelike geodesics on the quotient manifold-with-boundary $\overline{\mathcal{Q_R}}$, given by 
 \begin{align}
     \overline{\Gamma}(\alpha, \tau) := \exp_{\lambda(\alpha)}(\tau \nu(\alpha))\text{  on  }   D = \{ (\alpha,\tau) \in \mathbb{R}^2: -a < \alpha < a;\ 0 \leq \tau\leq T(\alpha)\}
 \end{align}
 where the proper time function $\alpha \mapsto T(\alpha)$ is defined by the condition that the $\alpha$-th geodesic\footnote{i.e. the geodesic which 'begins' at the point $\lambda(\alpha)$.}\footnote{We reiterate that, as a consequence of Condition 5 of admissibility, all inextendible future-directed timelike geodesics based on $\lambda$ terminate on $\mathcal{CH}^+$. This is shown in Lemma \ref{lambda_sends_to_CH}. So $\overline{\Gamma}$ is well-defined.} reaches $\mathcal{CH}^+$ when $\tau \to T(\alpha)$ i.e.
 \begin{equation}
     v\big[\overline{\Gamma}(\alpha,T(\alpha))\big] = 0
 \end{equation}
    Let 
    \begin{align}
        D^0 = \{(\alpha,\tau)\in \mathbb{R}^2: -a<\alpha<a; 0<\tau<T(\alpha)\}
    \end{align}
 and $\Gamma:=\overline{\Gamma}|_{D^0}$ be the restriction of $\overline{\Gamma}$ to $D^0$.\\[5pt]
 We say that $\gamma:(0,T)\to \mathcal{Q_R}$ is \textbf{a geodesic based on $\lambda$} if $\gamma$ is an affine, normalised, future-directed, future-inextendible timelike geodesic on $\mathcal{Q_R}$ such that $\gamma(0) = \lambda(\alpha)$ and $\dot{\gamma}(0) = \nu(\alpha)$ for some $\alpha\in(-a,a)$.
\end{defn}
\begin{figure}[h]
    \centering
    \includegraphics[width=0.45\linewidth]{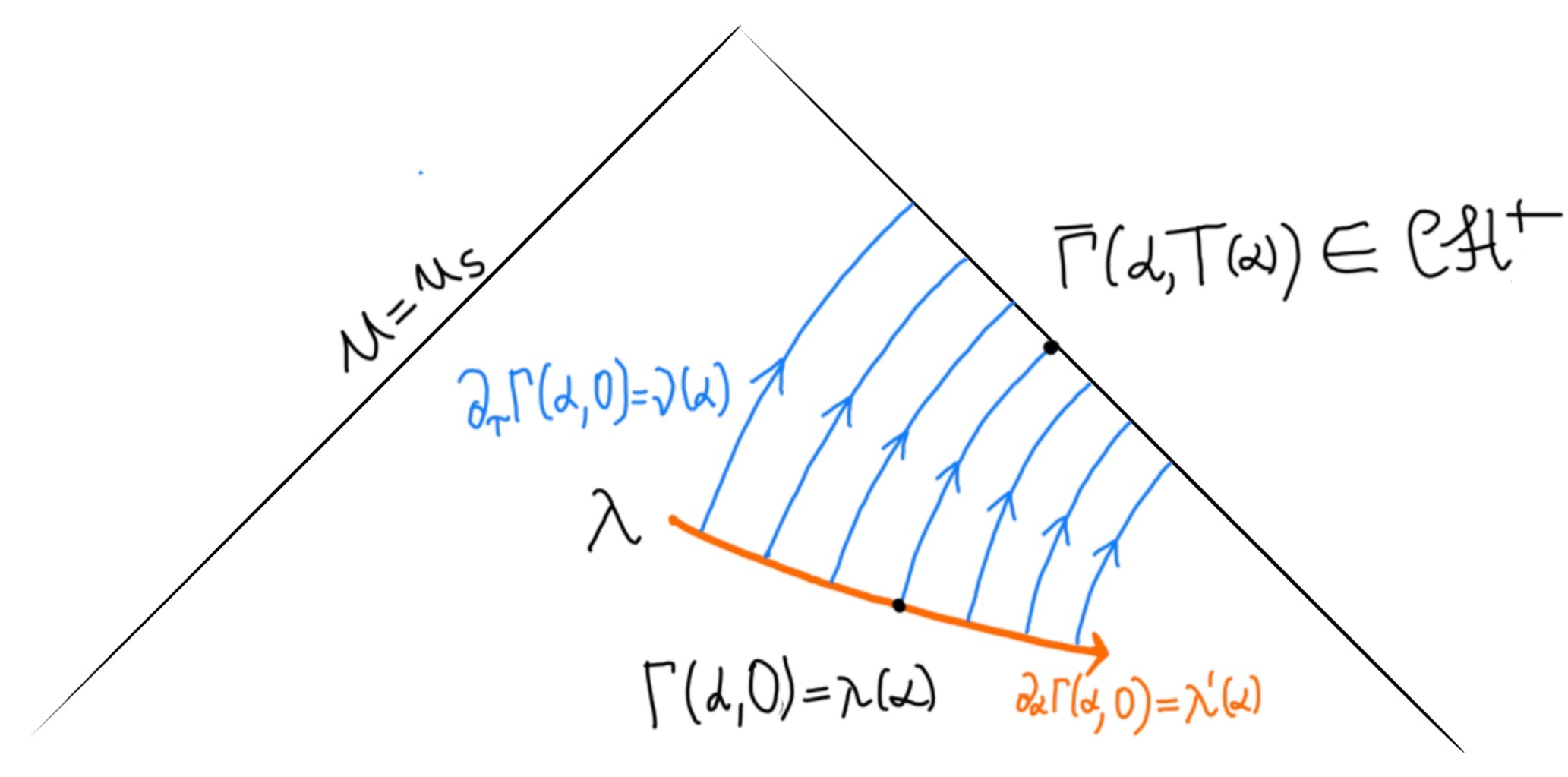}
    \caption{Illustrating the geodesic variation based on the admissible curve $\lambda$. Although in this picture there is no shell-crossing between the geodesics, this is a nontrivial result we are yet to prove. Key features in the illustration are that $\alpha\to \overline{\Gamma}(0,\alpha)$ is the curve $\lambda$, while by definition of the proper time function $T(\alpha)$, $\alpha\mapsto\overline{\Gamma}(\alpha,T(\alpha))$ maps to the weak null singularity. We will see that the restriction $\Gamma$ to $D^0$ maps away from the singularity i.e. away from the boundary of $\overline{\mathcal{Q_R}}$.}
    \end{figure}

\begin{rems}
\begin{enumerate}
    \item By properties of the exponential map and continuous dependence of solutions of the geodesic equation on initial conditions, $\overline{\Gamma}$ is well-defined, continuous in $D$ and smooth in the interior of $D$. In fact, it will be shown that the set $D^0$ is open and coincides with the interior of $D$.
    \item {The no shell-crossing result is false for spherically symmetric dust near a spacelike terminal singularity of Schwarzschild-like strength. See for example \cite{Szekeres_Lun}. Even more, according to \cite{GC_CS_ITF} the tidal forces experienced by test particles in a class of spherically symmetric spacetimes such that curvature invariants blow up at the origin but the metric is continuously extendible to the origin, e.g. wormholes with a singularity at the throat, are divergent.} 
\end{enumerate}
\end{rems}

\subsection{Behaviour of spherically symmetric dust near the WNS, precise statement of the theorem}
Below we state the first main result of this paper. We maintain the convention that the fluid lives in the $3+1$-dimensional spacetime $(\mathcal{R}, g)$ while the geodesic variation based on $\lambda$ lives in the $1+1$-dimensional quotient spacetime $\mathcal{Q_R}$. We denote by $\rho, U$ the fluid variables - a vector field and a scalar field on $\mathcal{R}$, as well as their pushforward into $\mathcal{Q_R}$. In the global double null coordinate chart, the $u$ and $v$ components of the radial geodesics of $\mathcal{R}$ (i.e. those with constant $\theta,\varphi$) coincide with the components of the geodesics of $\mathcal{Q_R}$, ensuring that the simplification to a 2-dimensional problem is valid.

\begin{thm}\label{master_theorem}
Given a $(n_0, N_0, C_K)$--admissible curve $\lambda$, let $\overline{\Gamma}: D\to \overline{\Gamma}(D) \subset \overline{\mathcal{Q}_\mathcal{R}}$ be the geodesic variation based on $\lambda$, as defined in Definition \ref{define_geodesic_variation_based_on_lambda}. Let $U$ be a spherically symmetric vector field defined in $\Gamma(D^0)\times S^2 \subset\mathcal{R}$ with components in double null coordinates given by $U^\mu(u,v) = \partial_\tau\Gamma^\mu\big(\Gamma^{-1}(u,v)\big)$ for $\mu\in\{u,v\}$. Let $\rho = \rho(u,v)$ be a  spherically symmetric scalar function defined on $\overline{\Gamma}(D)\times S^2 \subset \overline{\mathcal{R}}$ such that the restriction of $\rho$ to the spacelike hypersurface Image$(\lambda)\times S^2$ is smooth and positive, and extends to a positive function on $\overline{Image(\lambda)}\times S^2$. Assume that the pair $(\rho, U)$ satisfy the transport equation (\ref{dust_transport_equation}) in $\Gamma(D^0)\times S^2$. \\[5pt]
 Then 
    \begin{enumerate}
         \item There exist constants $0<c<C$, depending on $C_g,n_0,N_0$ such that vector field components $U^\mu$ for $\mu=u,v$ satisfy $c\leq U^\mu\leq C$ in Image$(\overline{\Gamma})\times S^2$.
        \item Image$(\Gamma) \subset \mathcal{Q_R}$
         \item The proper time $T$ is a strictly monotonic, continuous function of the parameter $\alpha$.
        \item $\Gamma$ is a ($C^\infty$-) diffeomorphism onto its image.  
        \item $\overline{\Gamma}$ is a homeomorphism onto its image.
       \item There exists a constant $C>0$, depending on $\norm{\rho}_{C^0(\lambda)}$, $n_0$, $N_0$, $v_{\text{min}}$ and $C_g$ such that $\rho(u,v) \leq C$ for all $(u,v) \in $ Image$(\overline{\Gamma})\times S^2$.
   \end{enumerate}
\end{thm}
\begin{rems}
    \begin{enumerate}
        \item The pair $(\rho, U)$ are the fluid variables of  spherically symmetric dust in $\mathcal{R}$. $\rho$ is the energy density of the dust and $U$ is the fluid velocity. As $U$ is spherically symmetric, $U^\theta$ and $U^\varphi$ are zero. By construction of the vector field $U$ as tangent to the geodesics of $\Gamma$, $U$  satisfies the geodesic equation in $\Gamma(D^0)\times S^2$. By the smoothness of $\lambda$ and standard properties of  the transport equation, $\rho$ and $U$ are smooth on $\Gamma(D^0)\times S^2$.
        \item We reiterate that admissible curves exist, as will be proved in Lemma \ref{existence_of_lambda}, ensuring that our theorem is not a statement about the empty set. 
    \end{enumerate}
\end{rems}

\subsection{Steps of the Proof of Theorem \ref{master_theorem}}
The proof of the Theorem \ref{master_theorem} relies on proving the following results under the same hypotheses:
\begin{enumerate}
    \item Components (in double null coordinates) of velocities of radial future-directed timelike geodesics are bounded and bounded away from zero on $\mathcal{CH}^+$. 
    \item The Jacobi field along an infalling radial timelike geodesic 'emanating' from a $(n_0, N_0, C_K)$--admissible curve $\lambda$ in $\mathcal{Q}_\mathcal{R}$ stays bounded and bounded away from zero. Proving this heavily relies on Condition 6 of Definition \ref{define_admissible_curve} i.e. that the image of $\lambda$ is contained close enough to the Cauchy horizon. Note that the initial data $(J,K=\nabla_\nu J)$ for the Jacobi equation is fully determined by the choice of $\lambda$. In particular, the initial value $J$ of the Jacobi field along a geodesic $\overline{\Gamma}(\alpha,\cdot)$ equals the tangent to $\lambda$ at the point $\lambda(\alpha)$, and the initial normal derivative $K$ of the Jacobi field is derived from the second fundamental form of $\lambda$ (where $\lambda$ is thought of as a 1-dimensional submanifold of $\mathcal{Q}_\mathcal{R}$).
    \item Through a bootstrap argument, and using the bounds on the geodesic components and the Jacobi field components, monotonicity of $T(\cdot)$ is proved. Again applying those bounds and the previous result of monotonicity, we prove that $T$ cannot have discontinuities using a constructive argument.  
    \item The variation $\overline{\Gamma}$ of  radially infalling timelike geodesics based on the  $(n_0, N_0, C_K)$--admissible curve $\lambda$ does not experience any shell-crossing before, or at, the Cauchy horizon.
    \item As a consequence of the transport equation, the energy density of the dust modelled by the abovementioned variation of geodesics based on an $(n_0, N_0, C_K)$--admissible curve $\lambda$ is bounded in the region $\supp(U)\subset\overline{\mathcal{Q}}_\mathcal{R}$, where it is supported.
\end{enumerate}
\subsection{Preliminary definitions for the stiff fluid}
\begin{defn}\label{define_admissible_hypersurface} Let $u_0,u_1\in(-\infty, u_s)$, $u_0<u_1$ and let $v_0\in(v_\kappa, 0)$ be such that $(u_0, v_0)\in \mathcal{Q_R}$ and such that $r_{,u}<0$ and $r_{,v}<0$ for all $(u,v)\in[u_0,u_1]\times [v_0, 0)\subset\mathcal{Q_R} $. \\[5pt]
Define the bifurcate null hypersurface $\mathcal{C}:= \mathcal{C}_+ \cup \overline{\mathcal{C}_-}$ in $\mathcal{R}$, where
\begin{align}
    &\mathcal{C}_+ = \{u_0\}\times [v_0, 0)\times S^2; \\
    &\overline{\mathcal{C}_-} = [u_0, u_1] \times \{v_0\}\times S^2
\end{align}
We will refer to $\mathcal{C}$ as a \textbf{$u_0$--admissible bifurcate null hypersurface}. 
\begin{rem}
    It is possible to choose a trapped characteristic rectangle as in Definition \ref{define_admissible_hypersurface} in view of Assumptions \ref{instability_assumptions}.
\end{rem}
\end{defn}

\subsection{Behaviour of spherically symmetric stiff fluid, precise statement of the theorem}

\begin{thm} \label{masterTheorem_2}
    Given a $u_0$-\textbf{admissible bifurcate null hypersurface} $\mathcal{C}=\mathcal{C}_+\cup \mathcal{\overline{C_-}}$ in $\mathcal{R}$, let $\mathring{\rho}$ be a continuous, positive, uniformly bounded away from zero  function on $\mathcal{C}$ such that the restrictions of $\mathring{\rho}$ to $\mathcal{C}_+$ and $\overline{\mathcal{C}_-}$ are smooth. Let $\mathring{U}$ be a normalised future-directed radial timelike continuous vector field defined on $\mathcal{C}$, such that the restrictions of $\mathring{U}$ to $\mathcal{C}_+$ and $\overline{\mathcal{C}_-}$ are smooth.  Assume furthermore that the following holds on the outgoing hypersurface $\mathcal{C}_+$:
    \begin{align}
    \norm{\sqrt{\mathring{\rho}}\mathring{U}^u(u_0,\cdot)}_{L^1(\mathcal{C_+})}:=4\pi\int_{v_0}^0 \sqrt{\mathring{\rho}}\mathring{U}^u(u_0,v)dv < \infty \label{new_assumed_L1_bound_masterThm2}
    \end{align}
    Let $(\rho, U)$ be the smooth scalar function and smooth vector field such that the stiff fluid energy-momentum conservation equations (\ref{V_is_divergence_free}-\ref{EM_conservation_normal_V}) are satisfied by $\rho$ and $V=\sqrt{\rho}U$ in $[u_0, u_1] \times [v_0, 0)$ with energy density $\rho$ and fluid velocity $U$ such that $\rho|_\mathcal{C} = \mathring{\rho}$ and $U|_\mathcal{C} = \mathring{U}$. Then
    \begin{enumerate}
        \item $\rho \to \infty$ as $v\to 0$ for every $u \in (u_0, u_1]$
        \item $U^u \to \infty$ and $U^v \to 0 $  as $v\to 0$ for every $u\in(u_0, u_1]$. 
    \end{enumerate}
\end{thm}
\begin{rem}
    About the assumption that a solution to the Euler equations (\ref{V_is_divergence_free}-\ref{EM_conservation_normal_V}) exists: Expanding the spherically symmetric wave equation $\Box_g\psi =0$ in double null coordinates gives a linear wave equation with first order terms. Given smooth characteristic data on $\mathcal{C}$, a unique smooth solution to the corresponding characteristic IVP exists in the future domain of dependence $D^+(\mathcal{C})$ by a combination of results by Rendall \cite{Rendall} and Ringstr\"{o}m \cite{Ringstrom}. In Section \ref{wave_section} we show that the solution $\psi$ has a future-directed timelike gradient throughout $D^+(\mathcal{C})$. Therefore it is in one-to-one correspondence with a solution $(\rho, U)$ of the Euler equations for a spherically symmetric stiff fluid in $(\mathcal{R}, g)$. The fluid variables $(\rho, U)$ can be reconstructed from $\psi$ by (\ref{U_intermsof_psi}) with $h^2=\rho$.
\end{rem}

\subsection{Steps of the Proof of Theorem \ref{masterTheorem_2}}
The proof of Theorem \ref{masterTheorem_2} is contained in Section \ref{wave_section} and relies on the following main steps.
\begin{enumerate}
    \item Relating $V=\sqrt{\rho}U$ to the scalar field $\psi$ such that $\nabla^a\psi = V^a$. Hence translating the stiff fluid IVP to a characteristic IVP for the wave equation, satisfied by $\psi$. 
    \item Using the future-directedness of $\mathring{U}^a$ to establish that the null derivatives of $\psi$ must be negative on the initial hypersurface. Performing a bootstrap argument to show that the negative sign of $\partial_u \psi$ and $\partial_v \psi$ is propagated inside the domain of dependence of $\mathcal{C}$. This step also uses the assumption (from Assumptions \ref{instability_assumptions}) that $\partial_u r$ and $\partial_v r$ are negative in the domain of dependence of $\mathcal{C}$. 
    \item Applying the instability assumption $\partial_v r \to - \infty$ (from Assumptions \ref{instability_assumptions}) to conclude that the energy density $\rho$  blows up at the weak null singularity.
    \item Using the assumption $\norm{\partial_u r}_{C(\mathcal{R})}<\infty$ (from Assumptions \ref{instability_assumptions}) to prove that the ratio $\partial_u \psi/\partial_v\psi$ vanishes as $v\to 0$. Back in stiff fluid terms, this translates to $U^u \to \infty$ and $U^v\to 0$ as $v\to 0$.
\end{enumerate}

\subsection{Scalar field perturbation of Reissner-Nordstr\"{o}m as an eligible spacetime}
We finalise this section with the motivating example for the class of spacetimes considered in this paper - namely spherically symmetric scalar field perturbations of the subextremal Reissner-Nordstr\"{o}m black hole solution. The following proposition, proved in Section \ref{EMSf_chapter}, shows that the Theorems \ref{master_theorem} and \ref{masterTheorem_2} hold on a nonempty (and in fact wide) class of spacetimes.  
\begin{prop} \label{EMSf_satisfies_assumptions}
 There exists a codimension-0 submanifold $(\mathcal{R}, g)$ of the spherically symmetric black hole spacetime arising from small generic spherically symmetric perturbations of subextremal Reissner-Nordstr\"{o}m initial data under the Einstein--Maxwell--scalar field system, described in \cite{Luk_Oh, Dafermos2012}, which satisfies Assumptions \ref{assumptions_on_L1_norms}-\ref{instability_assumptions}. 
\end{prop}

 \section{Spherical Dust and the Jacobi Equation near the WNS}\label{dust_section}
 This Section proves Theorem \ref{master_theorem} in full detail.
\subsection{Estimates for the geodesic velocity at the Cauchy horizon\label{estimates_for_geodesic_velocities}} 
This section begins with showing that the components of the velocity vector of radial future-directed timelike geodesics remain bounded in $\mathcal{R}$. Specifically, they remain bounded at the weak null singularity, if the geodesic in question extends that far.   
 \begin{thm} \label{main_thm_sec2} Let $\gamma: (0, T) \to \mathcal{R}, \mathcal{\tau}\to T$ be a future-directed radial timelike geodesic. Then the components of the geodesic velocity $\dot{\gamma}^\mu = d\gamma^\mu/ds$ in the $(u,v)$ double null coordinate chart described in Section \ref{preliminaries} remain bounded as $\tau \to T$. More precisely, there exists a $C>0$, depending only on the constant $C_g$ introduced in Assumptions \ref{assumptions_on_L1_norms} such that
 \begin{align}
     \sup_{0<\tau<T} \dot{\gamma}^\mu(\tau)\leq C \dot{\gamma}^\mu(0)
 \end{align}
 \begin{proof}
Let $\gamma$ be a future-directed radial timelike geodesic, so that $\dot{\theta} =\dot{\varphi} = 0$. Then, referring to the calculations in Appendix \ref{appendix_riemann_components}, the relevant Christoffel symbols are
        \begin{align}
        \Gamma^u_{uu} &= 2\omega_{,u} \\
        \Gamma^v_{vv} &= 2\omega_{,v} 
    \end{align}
    and $\Gamma^u_{uv} = \Gamma^v_{uv} = \Gamma^{u}_{AB} = \Gamma^v_{AB} = 0$ for $A,B\in \{ \theta, \varphi\}$. We parameterise the geodesic by the proper time so that the geodesic equation reduces to the following  system: 
    \begin{align}
        &\Ddot{\gamma}^u + 2 \partial_u \omega(u,v) (\dot{\gamma}^u)^2 = 0 \label{geod_eqn_u}\\
        &\Ddot{\gamma}^v + 2 \partial_v\omega(u,v)  (\dot{\gamma}^v)^2 = 0 \label{geod_eqn_v}\\
         &   - e^{2\omega} \dot{\gamma}^u \dot{\gamma}^v = -1 \label{normalization}
        \end{align}
        where the last equation is the normalization condition, $g(\dot{\gamma}, \dot{\gamma})=-1$. 
  Since $\gamma$ is future-directed, the components of the tangent to the timelike geodesics satisfy $\dot{\gamma}^u, \dot{\gamma}^v>0$. As $\omega$  extends continuously to $\{ v=0\}$, if $\gamma$ goes to $\mathcal{CH}^+$, the product of the velocities of the timelike curves remains bounded on the weak null singularity on $\mathcal{CH}^+$. Hence,
         \begin{align}
            \frac{\Ddot{\gamma}^u}{\dot{\gamma}^u} &= \frac{d}{d\tau} (\ln(\dot{\gamma}^u)) = - 2\omega_{,u} \dot{\gamma}^u\\
            \frac{\Ddot{\gamma}^v}{\dot{\gamma}^v} &= \frac{d}{d\tau} (\ln(\dot{\gamma}^v)) = - 2\omega_{,v} \dot{\gamma}^v
        \end{align} Integrating along the affinely parameterised geodesic, using that $\dot{\gamma}^u, \dot{\gamma}^v>0$ on $\mathcal{Q_R}$, we have: 
        \begin{align}
            \dot{\gamma}^u(\tau) &= \dot{\gamma}^u(0)\exp\Bigg(-2\int_{ 0}^\tau \omega_{,u}(u(\tau'),v(\tau'))\dot{\gamma}^u(\tau') d\tau'\Bigg) = \dot{\gamma}^u(0) \exp\Bigg(-2\int_{ u(0)}^{u(\tau)} \omega_{,u}(u,v(u))du \Bigg) \leq \nonumber\\
            & \leq \dot{\gamma}^u(0) \exp \Bigg(2\sup_{v\in(v(0),v(\tau))}\int_{u(0)}^{u(\tau)}|\omega_{,u}(u,v)|du \Bigg)\\
            & \leq \dot{\gamma}^u(0) e^{2C_g} \label{ubound}\\
            \dot{\gamma}^v(\tau) &= \dot{\gamma}^v(0) \exp\Bigg(-2\int_{0}^\tau \omega_{,v}(u(\tau'),v(\tau'))\dot{\gamma}^v(\tau') d\tau'\Bigg) =\dot{\gamma}^v(0) \exp\Bigg(-2\int_{v(0)}^{v(\tau)} \omega_{,v}(u(v),v)dv \Bigg) \leq \nonumber\\
            &\leq \dot{\gamma}^v(0) \exp \Bigg(2\sup_{u\in(u(0),u(\tau))}\int_{v(0)}^{v(\tau)}|\omega_{,v}(u,v)|dv \Bigg)\leq \dot{\gamma}^v(0) e^{2C_g}
            \label{vbound}
        \end{align}
        where we used the global geometric assumptions (\ref{geom_bound_as1} - \ref{geom_bound_as3}). We conclude that radial timelike geodesic velocities remain bounded across the Cauchy horizon. 
 \end{proof}
       \end{thm}
    The next step is to show that the geodesic velocities are bounded away from zero. 
   \begin{prop} \label{lower_bds_for_velocities} The double null coordinate components $\dot{\gamma}^u$ and $\dot{\gamma}^v$ of future-directed normalised radial timelike geodesics in $\mathcal{R}$ remain bounded away from zero. More precisely, there exists a $c>0$, depending only on the constant $C_g$ introduced in Assumptions \ref{assumptions_on_L1_norms} such that
   \begin{align}
       \inf_{0<\tau<T} \dot{\gamma}^u(\tau)\geq \frac{c}{\dot{\gamma}^v(0)} \ \ \ \text{   and   } \ \ \  \\ 
       \inf_{0<\tau<T} \dot{\gamma}^v(\tau)\geq \frac{c}{\dot{\gamma}^u(0)}
   \end{align}
   \begin{proof}
               By normalisation, $\dot{\gamma}^u \dot{\gamma}^v = e^{-2\omega}$. Employing equation (\ref{vbound}) and Assumptions (\ref{geom_bound_as1}-\ref{geom_bound_as3}), we obtain
               \begin{align}
                   \dot{\gamma}^u(\tau) = \frac{e^{-2\omega}}{\dot{\gamma}^v(\tau)} \geq \frac{1}{\norm{e^{2\omega}}_{C(\mathcal{R})} \dot{\gamma}^v(0) e^{2C_g}} \geq \frac{1}{C_g e^{2C_g} \dot{\gamma}^v(0)} 
                \end{align}
        Identically, employing equation (\ref{ubound}) and again Assumptions (\ref{geom_bound_as1}-\ref{geom_bound_as3}),
        \begin{align}
            \dot{\gamma}^v(\tau) = \frac{e^{-2\omega}}{\dot{\gamma}^u(\tau)} \geq \frac{1}{\norm{e^{2\omega}}_{C(\mathcal{R})} \dot{\gamma}^u(0) e^{2C_g}} \geq \frac{1}{C_g e^{2C_g} \dot{\gamma}^u(0) } 
        \end{align}
   \end{proof}
   \end{prop}
   The next Corollary combines the result for the lower bound on the geodesic velocities with the $L^1$ bound on the double null derivative $\omega_{,uv}$ from Assumptions \ref{assumptions_on_L1_norms}. It is used in the analysis of the Jacobi equation, Theorem \ref{bounds_for_jacobi_field_thm}.
   \begin{cor}
If $\gamma:(0,T)\to \mathcal{R}$ be a radial future-directed timelike geodesic in $\mathcal{R}$, then in the $(u,v)$ double null coordinate system we have
    \begin{align}
        \norm{\omega_{,uv}}_{L^1(\gamma)}\leq C_g^2 e^{2C_g} (1+|\lim_{\tau\to T}v(\gamma(\tau)) - \lim_{\tau\to 0} v(\gamma(\tau))|) \dot{\gamma}^u(0) \label{geom_bound_4}
    \end{align}
    where in the last inequality $\dot{\gamma}^u(0)$ denotes the $u$-component of the initial velocity of the geodesic $\gamma$.
    \begin{proof}
        Combine (\ref{geom_bound_as4}) from the global geometric assumptions with the result for the lower bound of $\dot{\gamma}^v(\tau)$ from Proposition \ref{lower_bds_for_velocities}. 
    \end{proof}
    \begin{rem}
        We can uniformly bound what is in the parentheses in (\ref{geom_bound_4}) by $1+|v_\kappa|$, hence 
        \begin{align}
            \norm{\omega_{,uv}}_{L^1(\gamma)}\leq C_g^2 e^{2C_g}(1+|v_{\kappa}|)\dot{\gamma}^u(0)
        \end{align}
    \end{rem}
         
   \end{cor}
   \subsubsection{$T(\gamma)$ is finite}
In order to have a well-defined and invertible geodesic variation map $\Gamma$, we want to ensure that any geodesic in our variation reaches the Cauchy horizon in finite proper time.
\begin{lem}\label{bounding_T_lemma}
    Let $p \in \mathcal{R}$ and $\gamma:(0,T)\to \mathcal{R}$ be a future-directed future-inextendible radial timelike geodesic with $\gamma(0)=p$. Then the proper time along $\gamma$ is bounded by
    \begin{align}
        T\leq C_g e^{C_g}|v(p)|\dot{\gamma}^u(0)
    \end{align}
    where $C_g$ is the constant from Assumptions \ref{assumptions_on_L1_norms}.
    
    \begin{proof}
          Assume $\gamma(0)=p$. By Proposition \ref{lower_bds_for_velocities}, the proper time along $\gamma$ is bounded as follows
        \begin{align}
            T = \int_{0}^T d\tau = \int_{v(p)}^{\lim_{\tau\to T }v(\tau)}  \frac{dv}{\dot{\gamma}^v(\tau(v))}  \leq C_g e^{2C_g} |v(p)| \dot{\gamma}^u(0)
        \end{align}
        where we also used that $|v(p) - \lim_{\tau\to T}v(\gamma(\tau))|\leq |v(p)|$ since $\lim_{\tau\to T}v(\gamma(\tau))\leq 0$\footnote{The limit exists because $v$ is strictly increasing along $\gamma$ and bounded above by $0$.}.\\[5pt]
    \end{proof}
\end{lem}
\begin{rem}
    In double null coordinates, 
    Theorem \ref{estimates_for_geodesic_velocities}, Proposition \ref{lower_bds_for_velocities}, Corollary \ref{unif_bounded_omega} and Lemma \ref{bounding_T_lemma} hold identically for the components of the tangent to a future-directed timelike geodesic in $\mathcal{Q_R}$. In what follows, we will apply these results both to radial geodesics in the $3+1$-dimensional $\mathcal{R}$ and to geodesics in the $1+1$-dimensional $\mathcal{Q_R}$. 
\end{rem}
  \begin{subsection}{Existence and properties of $(n_0, N_0, C_K)$--admissible curves} 

\begin{subsubsection}{Showing existence of $\lambda$}

In the next part we construct an example of a spacelike curve segment $\lambda$ which satisfies the conditions of Definition \ref{define_admissible_curve}, hence proving existence of $(n_0, N_0, C_K)$--admissible curves.
 \begin{lem}\label{existence_of_lambda}
       There exists an $(n_0, N_0, C_K)$--admissible curve for 
       \begin{align}
           n_0 = \frac{1}{\sqrt{C_g} e^{2C_g}}; \ \ \ \ \ \ \ \  N_0=\sqrt{C_g} e^{2C_g};\ \ \ \ \ \ \ \   C_K=0
       \end{align}
       where $C_g$ is the constant from the global geometric assumptions (\ref{geom_bound_as1} - \ref{geom_bound_as3}). 
       \begin{proof}
         Let $n_0 = (\sqrt{C_g} e^{2C_g})^{-1}$ and $N_0 = \sqrt{C_g} e^{2C_g}$. Fix $u_1\in(-\infty, u_s)$ and let 
        \begin{equation}
            \mu := \min{\bigg\{v_\kappa, \frac{(u_s-u_1) n_0}{N_0 C_g e^{4 C_g}} \ , \ x_*(n_0, N_0, C_K, C_g) \bigg\}}
        \end{equation}
        where $x_* = x_*(n_0, N_0, C_K, C_g)$ is the unique positive solution of the equation $Ax_* e^{B x_*} = 1$ with $A=A(n_0, N_0, C_K, C_g)$ and $B=B(n_0, N_0, C_K, C_g)$ given by (\ref{A}-\ref{B}). \\[5pt]
        We claim that the affine normalised spacelike geodesic segment $\lambda:(-a, a)\to \mathcal{Q_R}$ with 
        \begin{align}
            &a = \frac{\mu}{{4}N_0}\\
            &\lambda^v(-a) = -{\mu}\\
            &\lambda^u(-a) = u_1\\
            &(\lambda')^v(-a) =  e^{-\omega}|_{(u_1, -\mu)}\\
            &(\lambda')^u(-a) = - e^{-\omega}|_{(u_1, -\mu)}
        \end{align}
        is an $(n_0, N_0, C_K)$--admissible curve.\\[5pt]
        First we have to show that this is a smooth spacelike curve whose image stays in $\mathcal{Q_R}$. Indeed, let $\alpha_0 = \sup_{(-a, a]}\{\alpha:\text{Image}(\lambda|_{(-a,\alpha)})\subset\mathcal{Q_R} \}$. We have to show that $\alpha_0=a$. Suppose $\alpha_0<a$. Given that $v$ is increasing along $\lambda$, we have 
          \begin{align}
            \lambda^v(\alpha_0) \leq \lambda^v(-a) + \int_{-a}^{\alpha_0} (\lambda')^v(\alpha) d\alpha \leq -\mu + \frac{2\mu}{4N_0} N_0 = - \frac{\mu}{2}<0
         \end{align}
         therefore the curve stays strictly away from the Cauchy horizon at $\alpha_0$. Given that $u$ is increasing along $\lambda$, we have that $\lambda^u(\alpha_0)<u_1<u_s$, so the curve stays strictly away from the left future null boundary of $\mathcal{Q_R}$. We also find 
         \begin{align}
           \lambda^u(\alpha_0) \lambda^v(\alpha_0) \leq \Bigg|\frac{\mu}{2}\bigg( \lambda^u(-a) + \int_{-a}^{\alpha_0}(\lambda')^u(\alpha) d\alpha\bigg)\Bigg| \leq \Bigg| \frac{\mu}{2} \bigg( u_1 +
         \frac{2\mu}{4N_0} N_0\bigg)\Bigg| \leq \frac{v_\kappa}{2}(u_1 + u_s - u_1)<v_\kappa u_s
         \end{align}
        Therefore, the curve stays strictly away from the past boundary $\Sigma$, of $\mathcal{Q_R}$. It follows that there must be an $\alpha>\alpha_0$ such that $\lambda|_{(-a,\alpha)}\subset \mathcal{Q_R}$, contradicting the assumption that $\alpha_0$ is a supremum. Therefore, $\alpha_0=a$ and admissibility condition 1 is satisfied. \\[5pt]
        By the geodesic equation (\ref{geod_eqn_u}-\ref{geod_eqn_v}), for $\alpha\in(-a, a)$ we have 
           \begin{align}
               &(\lambda')^u(\alpha) = -e^{-\omega}|_{(u_1, -\mu)} \exp\Bigg(-2\int_{u(\lambda(-a))}^{u(\lambda(\alpha)} \omega_{,u}(u, v(u)) du\Bigg)\\ 
               &(\lambda')^v(\alpha) = e^{-\omega}|_{(u_1, -\mu)}
               \exp\Bigg(-2\int_{v(\lambda(-a))}^{v(\lambda(\alpha)} \omega_{,v}(u(v) u) dv\Bigg) 
           \end{align}
           In view of the assumptions (\ref{geom_bound_as1} - \ref{geom_bound_as3}) and the fact that Image$(\lambda)\subset \mathcal{Q_R}\cap\{u\leq u_1\} $, we then have 
           \begin{align}
             &n_0=\frac{1}{\sqrt{C_g} e^{2C_g}}\leq \frac{1}{\norm{e^{\omega}}_{C(\mathcal{R})}} \exp\bigg({-2\sup_{v\in(v(\lambda(-a)), v(\lambda(a)))}\norm{\omega_{,u}(\cdot,v)}_{L^1[u(\lambda(-a)), u(\lambda(a))]}}\bigg)\leq
             -(\lambda')^u(\alpha); \\ &(\lambda')^u(\alpha) \leq \norm{e^{-\omega}}_{C(\mathcal{R})} \exp\bigg({ 2\sup_{v\in(v(\lambda(-a)), v(\lambda(a)))}\norm{\omega_{,u}(\cdot,v)}_{L^1[u(\lambda(-a)), u(\lambda(a))]}}\bigg)\leq \sqrt{C_g}e^{2C_g} =N_0;\\[8pt]
               &n_0=\frac{1}{\sqrt{C_g} e^{2C_g}}\leq \frac{1}{\norm{e^{\omega}}_{C(\mathcal{R})}} \exp\bigg( {-2\sup_{u\in(u(\lambda(-a)), u(\lambda(a)))}\norm{\omega_{,v}(u,\cdot)}_{L^1[v(\lambda(-a)), 0]}}\bigg) \leq(\lambda')^v(\alpha);  \\ & -(\lambda')^v(\alpha) \leq \norm{e^{-\omega}}_{C(\mathcal{R})} \exp\bigg( {2\sup_{u\in(u(\lambda(-a)), u(\lambda(a)))}\norm{\omega_{,v}(u,\cdot)}_{L^1[v(\lambda(-a)), 0]}}\bigg) \leq \sqrt{C_g}e^{2C_g}=N_0
           \end{align}
            This tells us that the geodesic velocity components are bounded away from null by the constants $n_0 = (\sqrt{C_g} e^{2C_g})^{-1}$ and $N_0 = \sqrt{C_g} e^{2C_g}$. Therefore condition 3 of admissibility is satisfied. \\[5pt]
             In view of the fact that $u$ is decreasing along $\lambda$, we have $u_{\text{max}}=u_1$ and 
        \begin{align}
            u_{\text{min}} = \lim_{\alpha\to a} \lambda^u(\alpha) = u_1 + \int_{-a}^a (\lambda')^u(\alpha)d\alpha \geq u_1 - N_0 \frac{2\mu}{4N_0} = u_1 - \frac{\mu}{2} > -\infty
        \end{align}
            showing that condition 2 of admissibility holds.
           Next, denote by $\nu$ be the future-directed unit normal to $\lambda$ on its image. Since $\lambda$ is a geodesic, the second curvature vanishes:
        \begin{align}
            0 = \big\langle{\mathring{\nabla}}_{\lambda'} \lambda', \nu\big\rangle = - \big\langle {\mathring{{\nabla}}}_{\lambda'} \nu , \lambda' \big\rangle
        \end{align}
        So property 4 of admissibility is trivially satisfied with $C_K=0$. Finally, as $v$ is increasing along $\lambda$, $|v_\text{min}| = |\lambda^v(-a)| = {\mu}$, hence 
        \begin{align}
        &|v_{\text{min}}| \leq \frac{ (u_s-u_{\text{max}})n_0}{N_0 C_g e^{4C_g}} \implies   \text{ condition 5 of admissibility holds, and} \nonumber\\
        & |v_{\text{min}}| \leq x_* \implies  \text{ as $x_*$ solves $Axe^{Bx}=1$, condition 6 of admissibility holds.}
        \end{align}
           \end{proof}
   \end{lem}

\end{subsubsection}

\begin{subsubsection}{Showing that geodesics based on $\lambda$ reach $\mathcal{CH}^+$}

Roughly speaking, the next lemma says that all future-inextendible future-directed radial timelike geodesics emanating from an admissible curve $\lambda$ will reach the $\overline{\mathcal{R}}\cap \mathcal{CH}^+$-portion of the Cauchy horizon rather than escaping through $\{u=u_s\}$. This result is a consequence of Condition 5 of admissibility.
\begin{lem}\label{lambda_sends_to_CH} Let $\gamma:(0,T)\to \mathcal{Q_R}$ be a future-inextendible future-directed timelike geodesic based on an $(n_0, N_0, C_K)$--admissible curve $\lambda$. Then $\lim_{\tau\to T}\gamma(\tau)\in \overline{\mathcal{R}}\cap \mathcal{CH}^+$.
\begin{proof}
    As $\gamma$ is future-directed, $\dot{\gamma}^u, \dot{\gamma}^v>0$ for all $\tau\in (0,T)$. In particular, $0>\gamma^v(\tau)>\gamma^v(0)>v_\kappa$ and $u_s>\gamma^u(\tau)>\gamma^u(0)>-\infty$ for every $\tau$. Since $\gamma$ is inextendible and radial, we have 
       \begin{align}
            &\lim_{\tau \to T} \gamma^u(\tau) = u_s\\
            & \text{or } \lim_{\tau\to T} \gamma^v(\tau) = 0
       \end{align}
        or both. We will show that the first does not happen, so the second must hold. Note that for any $0<\tau<T$ 
        \begin{align}
            \gamma^u(\tau) - \gamma^u(0)&=  \int_0^\tau \dot{\gamma}^u(s)ds \leq T \dot{\gamma}^u(0) \exp{\bigg({2\sup_{v\in(v(0), v(\tau))}\norm{\omega_{,u}(\cdot,v)}_{L^1[u(0), u(\tau)]}}\bigg)} \leq  T \dot{\gamma}^u(0) e^{2 C_g} \leq \nonumber\\
            &  \leq  e^{2C_g} |v_{\text{min}}| \sqrt{\frac{C_g N_0}{n_0}} \times \sqrt{\frac{C_g N_0}{n_0}} \times e^{2 C_g}  =\frac{C_g N_0}{n_0} e^{4C_g}|v_{\text{min}}|<u_s - u_\text{max}
        \label{sends_to_CH_1}
        \end{align}

        Where we applied Theorem \ref{estimates_for_geodesic_velocities} to bound the integrand by $\dot{\gamma}^u(0)e^{2C_g}$ and Lemma \ref{bounding_T_lemma} to bound $T$. We also the fact that the $u$-component of the normal to $\lambda$ is given by $\dot{\gamma}^u(0) = \nu^u = e^{-\omega} \sqrt{(\lambda')^u/(\lambda')^v}\leq \sqrt{C_g N_0/n_0}$ where the inequality is a consequence of admissibility condition 3 and equation (\ref{geom_bound_as1}) from Assumptions \ref{assumptions_on_L1_norms}. In the last line, we applied admissibility condition 5.
         Taking the first and last expression in (\ref{sends_to_CH_1}), we immediately get
        \begin{align}
            \gamma^u(\tau)<u_s -u_{\text{max}} + \gamma^u(0) \leq u_s
        \end{align}
        Because $\gamma^u(0) \leq  u_{\text{max}}$.
\end{proof}
\end{lem}
\end{subsubsection}
\end{subsection}
    
 \subsection{The Variational Field of Timelike Geodesic Variations \label{regularity_of_timelike_geodesic_variations}}
  In this section we prove an estimate for the geodesic deviation vector field along a future-directed normalised timelike geodesic $\gamma$ in $\mathcal{Q_R}$, for geodesics based on an admissible curve. It is a standard result in differential geometry that the variation vector $J$ field along a variation of geodesics satisfies the Jacobi equation,
  \begin{align}
      \mathring{\nabla}_{\dot{\gamma}}\mathring{\nabla}_{\dot{\gamma}} J = \mathring{R}(\dot{\gamma}, J) \dot{\gamma} \label{jacobi_eqn_general}
  \end{align}
  where $\mathring{R}$ denotes the Riemann tensor of $g_{\mathcal{Q_R}}$
\subsubsection{Jacobi equation on $\mathcal{Q}_{\mathcal{R}}$ in orthonormal basis}
 To simplify the analysis in 1+1 dimensions, we rewrite the Jacobi equation along $\gamma$ in a parallelly transported orthonormal basis (ONB). Let $\tau\in(0, T)$ and define the following  basis of $T_{\gamma(0)}\mathcal{Q_R}$.
\begin{equation}
(e_0, e_1) = (\dot{\gamma}^u \partial_u + \dot{\gamma}^v \partial_v, \dot{\gamma}^u \partial_u - \dot{\gamma}^v \partial_v)  \label{ON_BASIS}
\end{equation}
In this basis $e_0 = \dot{\gamma}(0)$ and $g_{\mathcal{Q_R}}\to $ diag$(-1,1)$, the components of the Minkowski metric. Stipulating that $\nabla_{\dot{\gamma}}e_A = \nabla_{e_0} e_A = 0$ for $A=0,1$ everywhere along $\gamma$ determines orthonormal bases for the tangent spaces at every point on the image of $\gamma$, such that the resulting basis vectors are parallelly transported along $\gamma$. Expanding the Jacobi equation in the parallelly transported orthonormal basis for $T_{\gamma(\tau)}\mathcal{Q_R}$, assuming that $\dot{\gamma}^\varphi = \dot{\gamma}^\theta = 0$ we get {for $A \in \{0,1\}$:}
\begin{align}
    \mathring{\nabla}_{e_0} \mathring{\nabla}_{e_0} J^A(\tau) &= \frac{d^2 J^A}{d\tau^2}(\tau) = \big[\mathring{R}(e_0, J)e_0\big]^A = \mathring{R}^A_{\ BCD} e_0^B e_0^C J^D = \mathring{R}^A_{\ 00D} J^D \label{JAD}
\end{align}
 Let $(\Lambda^A_{\mu})$ denote the $2\times2$ basis transformation matrix from the $(\partial_\mu)$ basis to the $(e_A)$ basis at some arbitrary point on the curve, and let $(\big(\Lambda^{-1})^{\nu}_{D}\big)$ denote its inverse. Clearly
\begin{align}    (\Lambda^{A}_{\mu}) = 
    \begin{pmatrix}
    \frac{1}{2\dot{\gamma}^u} & \frac{1}{2\dot{\gamma}^u}  \\
    \frac{1}{2\dot{\gamma}^v}  & -\frac{1}{2\dot{\gamma}^v}  \\
    \end{pmatrix}; \ \ 
    \big((\Lambda^{-1})^{\nu}_D\big) = \begin{pmatrix}
    \dot{\gamma}^u & \dot{\gamma}^v\\
    \dot{\gamma}^u & -\dot{\gamma}^v\\
    \end{pmatrix}; \hspace{10pt}
\label{ON_basis_relations}
\end{align}
Hence, after a straightforward calculation transforming the Riemann tensor components given in Appendix \ref{appendix_riemann_components} into the ONB $(e_A)$, we obtain the following equations for the radial components of the Jacobi field:
\begin{align}  \Ddot{J}^0(\tau)=2\omega_{,uv}e^{-2\omega} J^0(\tau) \label{JAD_0} \\
    \Ddot{J}^1(\tau) = 2 \omega_{,uv}e^{-2\omega} J^1(\tau) \label{JAD_1}
\end{align}
So the system becomes uncoupled and very simple.

\subsubsection{Bounds for the Jacobi field}
 We are ready to state and prove the main result of this subsection which concerns the variational field of a geodesic variation based on an admissible curve in $\mathcal{Q_R}$. 
\begin{thm}\label{bounds_for_jacobi_field_thm}
 Let $\lambda$ be a $(n_0, N_0, C_K)$--admissible curve, $p\in$Image$(\lambda)$ and let  $\lambda'$ denote the tangent to $\lambda$ at $p$. Let $n_0, N_0$ be the bounds on the components of $\lambda'$ as defined in (\ref{condd5}) and let $C_K$ be the bound on the second fundamental form of $\lambda$ as defined in (\ref{bound_on_curvature_of_lambda}). Consider the initial value problem for the Jacobi equation (\ref{jacobi_eqn_general}) along the normalised future-directed timelike geodesic $\gamma: (0, T) \to \mathcal{Q_R}$ such that $\gamma(0)=p$ with $\gamma^v(\tau)\to 0$ as $\tau\to T$. Let $(e_A)$ be a parallelly transported orthonormal basis along $\gamma$, such that $\dot{\gamma} = e_0$ and $e_1 = \langle\lambda',\lambda'\rangle^{-1/2}\lambda'$ at $p$. Assume that the following initial data is supplied for the Jacobi equation:
\begin{align}
    &J|_{\tau=0}= J^1(0) e_1 = \lambda' \neq 0 \label{ICs_for_jacobi1}\\
    &\mathring{\nabla}_{\dot{\gamma}} J|_{\tau=0} \equiv K|_{\tau=0} = \langle e_1, \mathring{\nabla}_{\lambda'} \nu\rangle e_1 = \frac{\langle \mathring{\nabla}_{\lambda'}\nu, \lambda'\rangle}{\sqrt{\langle \lambda' , \lambda' \rangle}}   e_1 \label{ICs_for_jacobi2}
\end{align} 
 Here $J^1(\tau), K^1(\tau)$ are the transverse components of $J$ and $K={\mathring{\nabla}}_{\dot{\gamma}} J$ with respect to the parallelly transported ON basis. Then the solution to the initial value problem (\ref{jacobi_eqn_general}), (\ref{ICs_for_jacobi1} - \ref{ICs_for_jacobi2}) satisfies
\begin{align}
    n_\perp \leq |J^1(\tau)|\leq N_\perp \ \forall \ 0 < \tau < T \label{J_bounds_statement}
\end{align}
 for some constants $0<n_\perp\leq N_\perp$, depending on $n_0, N_0, C_K$, on $C_g$. 

    \begin{rems}
    \begin{enumerate}
        \item Notice that the initial value for the Jacobi field (\ref{ICs_for_jacobi1}) is given by the tangent to the admissible curve $\lambda$, while the initial value for the normal derivative to the Jacobi field (\ref{ICs_for_jacobi2}) is given by the (only nonvanishing component of) the second fundamental form of $\lambda$ as a 1-dimensional submanifold of $\mathcal{Q}_\mathcal{R}$.  This choice is consistent with the fact that the restriction of the geodesic velocity field to $\{\tau=0\}$ coincides with the normal to $\lambda$. In the framework of the full geodesic variation based on $\lambda$, $\gamma$ is a geodesic $\Gamma(\alpha, \cdot)$ for some $\alpha\in(-a,a)$, and $J$ is the variation field along this geodesic. This means that the commutator of $J$ and $\dot{\gamma}$ vanishes, so $\mathring{\nabla}_{\dot{\gamma}} J = \mathring{\nabla}_{J}\dot{\gamma}$. Thinking about the normal $\nu$ as a vector field defined on the image of $\lambda$, at proper time zero this looks like $\mathring{\nabla}_\nu \lambda' = \mathring{\nabla}_{\lambda'}\nu$.
        Therefore, the choice of the admissible curve $\lambda$ fully determines the initial data for the Jacobi equation.
        \item In the $3+1$-dimensional Lorentzian manifold $(\mathcal{R}, g)$, this result is equivalent to showing that the transverse component of the Jacobi field along a radial future-directed timelike geodesic is bounded and bounded away from zero, provided we start with nonzero initial transverse component $J^1$, and sufficiently close to the Cauchy horizon as per admissibility condition 6 for $\lambda$. This is because along radial geodesics the transverse (radial, normal to $\dot{\gamma}$) and parallel (radial, colinear with $\dot{\gamma}$) components of the Jacobi equation decouple when expressed in basis $(e_0, e_1, e_2, e_3)$ where $(e_0, e_1)$ are given by (\ref{ON_BASIS}) and $(e_2,e_3)$ is a basis for $T_{\gamma(0)}S^2$, parallelly transported along $\gamma$. Furthermore, along radial geodesics the angular components $J^2$ and $J^3$ of the Jacobi equation decouple from the radial components $J^0$ and $J^1$ when the equation is expressed in the basis $(e_0, e_1, e_2, e_3)$. Therefore the time evolution of $J^2$ and $J^3$ (equivalently, $J^\theta$ and $J^\varphi$) does not affect the time evolution of $J^1$.  
        \item  The parallel component is irrelevant when it comes to answering the questions whether conjugate points develop between points on a geodesic variation. Changing it only changes the parameterisation of a given family of geodesics.
         \end{enumerate}
         \begin{figure}[h]
    \centering
    \includegraphics[width=0.4\linewidth]{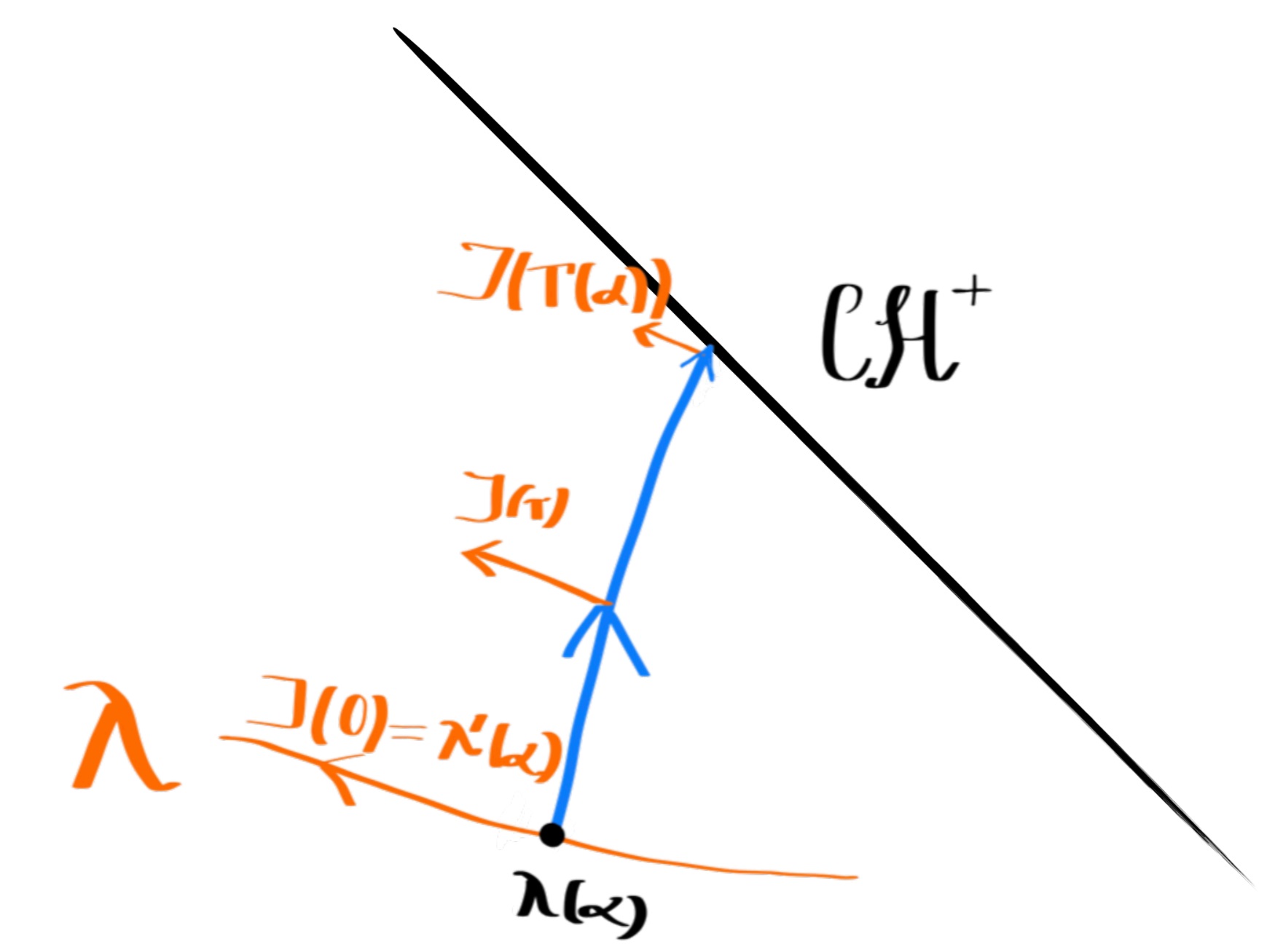}
    \caption{The Jacobi field along a future-directed inextendible radial timelike geodesic based on an admissible curve $\lambda$.}
    \end{figure}

    \end{rems}
    \begin{proof} Since the system is decoupled and $J^{0}(0) = {K}^{0}(0) = 0$, by uniqueness of the solution of the Jacobi equation, no parallel component develops i.e. $J^0 \equiv 0$ along $\gamma$.
        In this case we're left with a scalar linear ODE
        \begin{equation}
            \Ddot{J}^1 = 2 e^{-2\omega} \omega_{,uv} J^1 \label{2nd_order_lner_ode}
        \end{equation}
        We turn this into a first order system by setting $K^1 = \dot{J^1}$, so
        \begin{align}
            \frac{d}{d\tau}\begin{pmatrix}
                J^1 \\ K^1
            \end{pmatrix} = \begin{pmatrix}
                0 & 1 \\
                2\omega_{,uv} e^{-2\omega} & 0
            \end{pmatrix} \begin{pmatrix}
                J^1 \\ K^1
            \end{pmatrix} \label{Radial_Jacobi_System}
        \end{align}
        If $\norm{\cdot}_2$ denotes the Euclidean norm on $\mathbb{R}^2$ then
        \begin{align}
            &\frac{d}{d\tau}\begin{vmatrix}\begin{vmatrix}
                J^1 \\ K^1
            \end{vmatrix}\end{vmatrix}_2 = \frac{d}{d\tau} \sqrt{(J^1)^2(\tau) + (K^1)^2(\tau)} = \frac{J^1(\tau)\dot{J}^1(\tau) + K^1(\tau) \dot{K}^1(\tau)}{\sqrt{(J^1(\tau))^2 + (K^1(\tau))^2}} = \nonumber \\[5pt]
            & = \frac{J^1(\tau) K^1(\tau) + 2\omega_{,uv}e^{-2\omega}J^1(\tau) K^1(\tau)}{\begin{vmatrix}
                \begin{vmatrix}
                    J^1 \\ K^1
                \end{vmatrix}
            \end{vmatrix}_2}
        \end{align}
        where we used equation (\ref{Radial_Jacobi_System}) in the last equality. Rearranging and applying Young's inequality yields
        \begin{align}
            &\begin{vmatrix}
                \begin{vmatrix}
                    J^1 \\ K^1
                \end{vmatrix}
            \end{vmatrix}_2 \frac{d}{d\tau} \Bigg( \begin{vmatrix}
                \begin{vmatrix}
                    J^1 \\ K^1
                \end{vmatrix}
            \end{vmatrix}_2\Bigg) = \frac{1}{2}\frac{d}{d\tau} \Bigg( \begin{vmatrix}
                \begin{vmatrix}
                    J^1 \\ K^1
                \end{vmatrix}
            \end{vmatrix}^2_2\Bigg) \nonumber\\
            &= J^1 K^1 (1+2\omega_{,uv}e^{-2\omega}) \leq \frac{1}{2}\big((J^1)^2 + (K^1)^2\big)(1+ 2|\omega_{,uv}| e^{-2\omega}) \implies \nonumber \\
            &\frac{d}{d\tau}\big((J^1)^2 + (K^1)^2\big) \leq \big((J^1)^2 + (K^1)^2\big)(1+2|\omega_{,uv}|e^{-2\omega}) \label{some_est}
        \end{align}
        By uniqueness of the solution to the Jacobi equation, assuming nonzero initial data ensures that $\big( (J^1)^2 + (K^1)^2\big)$ is nonzero everywhere on $\gamma$. We can apply Gr\"{o}nwall's inequality after equation (\ref{some_est}), giving
        \begin{align}
            &(J^1(\tau))^2 + (K^1(\tau))^2 \leq \big( (J^1(0))^2 + (K^1(0))^2\big)\exp{\Bigg(\int_{0, \gamma}^ \tau[1+ 2|\omega_{,uv}| e^{-2\omega}]ds\Bigg)} \leq \nonumber \\
            \leq &\big((J^1(0))^2 + (K^1(0))^2\big)\exp{\big(T + 2 \norm{e^{-2\omega}}_{C(\gamma)}\norm{\omega_{,uv}}_{L^1(\gamma)}\big)} \leq \nonumber \\
            &\leq J^1(0)^2 \Bigg(1+\bigg(\frac{K^1(0)}{J^1(0)}\bigg)^2\Bigg) \exp{\big(T + 2 C_g^{3}e^{2C_g}(1+|v_\text{min}|)\dot{\gamma}^u(0)\big)} =\nonumber\\
            &=\langle \lambda',\lambda' \rangle \Bigg(1+ \bigg(\frac{\langle \lambda',\nabla_{\lambda'} \nu \rangle}{\langle\lambda', \lambda'\rangle}\bigg)^2\Bigg) \exp{\Bigg[T + 2 C_g^{3} e^{2C_g}(1+|v_\text{min}|) e^{-\omega}\sqrt{\frac{(\lambda')^u}{(\lambda')^v}}\Bigg]} 
            \label{upper_j_bound_1}
        \end{align}
        In the second inequality we used that $\norm{e^{2\omega}}_{C(\mathcal{R})}\leq C_g$ and $\norm{\omega_{,uv}}_{L^1(\gamma)}\leq C_g^2 e^{2C_g}(1+|v_\text{min}|)\dot{\gamma}^u(0)$. Next, we used that the initial $u$-velocity of the geodesic $\gamma$ is the $u$-component of the unit normal $\nu$ to $\lambda$ and that the $u$- component of $\nu$ can be expressed in terms of the components of $\lambda'$ as
        \begin{align}
            \nu^u = \dot{\gamma}^u(0)=e^{-\omega} \sqrt{\frac{(\lambda')^u}{(\lambda')^v}}
        \end{align}
        We also used $K^1(0) = {\langle \lambda' , \nabla_{\lambda'} \nu \rangle}/\sqrt{\langle\lambda',\lambda'\rangle}$\footnote{Recall that $K^1(0)$ geometrically corresponds to the only nonvanishing component of the second fundamental form of the 1-dimensional submanifold $\lambda$.}. We also used that the initial Jacobi field equals the tangent to $\lambda$, hence $g(\lambda', \lambda') = J^1(0)^2 =e^{2\omega}(\lambda')^u (\lambda')^v$.
        By admissibility condition 3 of $\lambda$, $(\lambda')^u<N_0$ and $(\lambda')^v< N_0$. Also, by admissibility condition 4 of $\lambda$, ${\langle \lambda',\nabla_{\lambda'} \nu \rangle}/{\langle\lambda', \lambda'\rangle}\leq C_K$. When we apply the abovementioned consequences of the admissibility of $\lambda$, (\ref{upper_j_bound_1}) becomes
        \begin{align}
            &(J^1(\tau))^2 + (K^1(\tau))^2 \leq C_g N_0^2 (1+C_K^2) \exp{\Bigg[T + 2 C_g^{7/2} e^{2C_g} (1+|v_{\text{min}}|)\sqrt{\frac{N_0}{n_0}}\Bigg]} \label{upper_j_bound_2}
        \end{align}
        It remains to bound the proper time $T$ along $\gamma$. By Lemma \ref{bounding_T_lemma}, the proper time along the geodesic $\gamma$ emanating from the admissible curve $\lambda$ with $v[\lambda(\alpha)]>v_{\text{min}} \ \forall \alpha$ satisfies
        \begin{align}
            T\leq C_g |v_{\text{min}}| \sqrt{\frac{C_g N_0}{n_0}} \exp{(2C_g)}
        \end{align}
       Inserting this into \ref{upper_j_bound_2} yields 
        \begin{align}
            &(J^1(\tau))^2 + (K^1(\tau))^2 \leq C_g N_0^2 (1+C_K^2) \exp{\Bigg[ C_g |v_{\text{min}}| \sqrt{\frac{C_g N_0}{n_0}} \exp{(2C_g)} + 2 C_g^{7/2} e^{2C_g} (1+|v_{\text{min}}|)\sqrt{\frac{N_0}{n_0}}}\Bigg] \label{upper_j_bound_3}
        \end{align}
        This establishes the upper bound on the Jacobi field along the geodesic $\gamma$:
        \begin{align}
            N_\perp := N_0\sqrt{C_g (1+C_K^2)} \exp{\Bigg[ \frac{C_g}{2} |v_{\text{min}}| \sqrt{\frac{C_g N_0}{n_0}} \exp{(2C_g)} +  C_g^{7/2} e^{2C_g} (1+|v_{\text{min}}|)\sqrt{\frac{N_0}{n_0}}\Bigg]} \label{Nperp_expression}
        \end{align}
        Note that the normal derivative component $K^1(\tau)$ is bounded by the same constant. \\[5pt]
        To prove the lower bound, we note that for each $0\leq\tau\leq T$ we have
        
        \begin{align}
            &|J^1(\tau) - J^1(0)| = \Bigg|\int_0^\tau K^1(s)ds\Bigg| \implies\nonumber\\
            &|J^1(\tau)| \geq |J^1(0)| - \int_0^\tau |K^1(s)|ds \geq |J^1(0)|\Bigg[1 - T \times \sqrt{1+\bigg(\frac{K^1(0)}{J^1(0)}\bigg)^2}\exp{\bigg(\frac{T}{2} +  C_g^3e^{2C_g}(1+|v_\text{min}|)\dot{\gamma}^u(0)\bigg)}\Bigg]\label{lower_j_bound_1}
        \end{align}
        where we inserted the square root of the third line of (\ref{upper_j_bound_1}). We already have (upper) bounds on $T$, $\dot{\gamma}^u(0)$ and $|K^1(0)/J^1(0)|$ from the admissibility conditions. Again due to admissibility condition 3, $|J^1(0)| = e^{\omega}\sqrt{(\lambda')^u(\lambda')^v}\geq n_0/\norm{e^{-\omega}}_{C(\mathcal{R})}\geq n_0/\sqrt{C_g}$. Combining everything, we can establish the lower bound of the Jacobi field,
        \begin{align}
            |J^1(\tau)|&\geq \frac{n_0}{\sqrt{C_g}}\Bigg[1 - |v_{\text{min}}| \sqrt{\frac{C_g^3 N_0(1+C_K^2)}{n_0}}\exp{\bigg(\frac{1}{2} C_g |v_{\text{min}}| \sqrt{\frac{C_g N_0}{n_0}} e^{2C_g} + C_g^{3} e^{2C_g} (1+|v_{\text{min}}|)\sqrt{\frac{C_g N_0}{n_0}}}\Bigg)\Bigg] = \nonumber\\
            &= \frac{n_0}{\sqrt{C_g}}\Bigg[1 - |v_{\text{min}}| \sqrt{\frac{C_g^3 N_0(1+C_K^2)}{n_0}}\exp{\bigg(\sqrt{\frac{C_g^7 N_0}{n_0}}e^{2C_g}  + \sqrt{\frac{C_g^3 N_0}{n_0}} e^{2C_g}\bigg(\frac{1}{2} + C_g^2 \bigg)|v_{\text{min}}| }\Bigg)\Bigg]=\nonumber\\
            &=  \frac{n_0}{\sqrt{C_g}} \bigg[1-{A|v_{\text{min}}|e^{B|v_\text{min}|}} \bigg] =: n_\perp
        \end{align}
         This is strictly positive because of admissibility condition 6. More precisely, let $x_*>0$ be the unique solution of $Axe^{Bx}=1$. Then $|v_{\text{min}}|<x_*$ implies $A|v_{\text{min}}|e^{B|v_\text{min}|}<1$. \label{n_perp_exprn}
    \end{proof}
\end{thm}
\begin{cor}\label{orientation_preserving_corollary}
    Under the hypotheses of Theorem \ref{bounds_for_jacobi_field_thm}, the components $J^u$ and $J^v$ of the Jacobi field along $\gamma$ are uniformly bounded and bounded away from zero. More precisely, $\exists\  0<n_{null}<N_{null} $ such that
    \begin{align}
    n_{null}\leq|J^u(\tau)|, |J^v(\tau)|\leq N_{null}
    \end{align}
    for every $\tau\in (0, T]$.
    \begin{rem}
        This implies that the Jacobi field preserves the original orientation of the admissible curve $\lambda$, i.e. if $u$ is decreasing and $v$ increasing along the admissible curve $\lambda$, then $J^u<0$ and $J^v>0$ holds at every $\tau$ and for all geodesics based on $\lambda$ (and vice versa if $u$ is increasing and $v$ decreasing). 
    \end{rem}
    \begin{proof}
        Using (\ref{ON_basis_relations}), we compute
        \begin{align}
            J^u  = (\Lambda^{-1})^u_{1}J^1 = \dot{\gamma}^v J^1 \implies |J^u| \in \Bigg[\sqrt{\frac{n_0}{C_g^3N_0}}  e^{-2C_g}n_\perp, \sqrt{ \frac{C_gN_0}{n_0}} e^{2C_g}N_\perp\Bigg] =: [n_{null}, N_{null}]\\
            J^v = (\Lambda^{-1})^v_{1}J^1 = -\dot{\gamma}^v J^1\implies |J^v| \in \Bigg[\sqrt{\frac{n_0}{C_g^3N_0}}  e^{-2C_g}n_\perp, \sqrt{ \frac{C_gN_0}{n_0}} e^{2C_g}N_\perp\Bigg] = [n_{null}, N_{null}] \label{nullbounds_on_J}
        \end{align}
        where we applied Theorem \ref{estimates_for_geodesic_velocities} and Proposition \ref{lower_bds_for_velocities}. We also expressed the initial future-directed unit normal $\nu =\dot{\gamma}(0)$ in terms of the double null components of $\lambda$ and applied condition 3 of admissibility. 
    \end{proof}
    \begin{cor}\label{unidorm_bounds_on_geod_velocities} Under the hypotheses of Theorem \ref{bounds_for_jacobi_field_thm}, the components $\dot{\gamma}^u$ and $\dot{\gamma}^v$ of the geodesic velocity of $\gamma$ are uniformly bounded and bounded away from zero. More precisely,
    \begin{align}
        e^{-2C_g} \sqrt{\frac{n_0}{C_g^3 N_0}} \leq \dot{\gamma}^u(\tau), \dot{\gamma}^v(\tau) \leq e^{2C_g} \sqrt{\frac{C_g N_0}{n_0}}
    \end{align}
    for every $\tau\in(0,T]$.
        \begin{proof}
            The components of unit normal to $\lambda$ are given by $\dot{\gamma}^u(0)=\nu^u=e^{-\omega}\sqrt{-(\lambda')^u/(\lambda')^v}$ and $\dot{\gamma}^v(0)=\nu^v=e^{-\omega}\sqrt{-(\lambda')^v/(\lambda')^u}$. In view of admissibility condition 3 and Assumptions \ref{assumptions_on_L1_norms}, these are $\leq \sqrt{C_g N_0/n_0}$. Applying Theorem \ref{estimates_for_geodesic_velocities} and Proposition \ref{lower_bds_for_velocities} yields the result. 
        \end{proof}
    \end{cor}
\end{cor}

\subsection{$\Gamma$ is a diffeomorphism}
\label{subsection_gamma_regularity}
\begin{thm} \label{interior_regularity_of_gamma}
 Let $\lambda:(-a,a)\to \mathcal{Q}_\mathcal{R}$ be a $(n_0, N_0, C_K)$--admissible curve with future-directed unit normal $\nu$. Let $\overline{\Gamma}: {D}\to \overline{\mathcal{Q}}_\mathcal{R}$ be the variation of geodesics based on $\lambda$, and let $\Gamma:=\overline{\Gamma}|_{D^0}$, constructed as in Theorem \ref{master_theorem} . Then items 2, 3 and 4 from the statement of Theorem \ref{master_theorem} are satisfied, namely
    \begin{itemize}
        \item {Image}$(\Gamma) \subset \mathcal{Q}_\mathcal{R}$, 
        \item The function $T(\cdot)$ is strictly monotonic and continuous.
        \item $\Gamma$ is a diffeomorphism onto its image. 
    \end{itemize}
    \end{thm}
The proof of claim 1 is established in Lemma \ref{image_of_Gamma_is_in_Q}. The proof of claim 2 is split between Proposition \ref{T_is_monotonic}, Corollary \ref{T_is_monotonic_reverse_dirn_corollary} and Proposition \ref{contn_dep_of_T_on_alpha}. Claim 3 is split between Theorem \ref{Gamma_local_diffeomorphism} and Theorem \ref{injectivity_of_Gamma}. 
\begin{rems}
\begin{enumerate}
    \item  To reiterate key points from the Main Theorem \ref{master_theorem}: ${D} = \{ (\alpha,\tau) \in \mathbb{R}^2: -a < \alpha < a; 0 \leq \tau\leq T(\alpha)\}$
    and the map $\overline{\Gamma}: {D} \to \overline{\mathcal{Q}}_\mathcal{R}$ is a variation of affine, normalized, future-directed future-inextendible timelike geodesics on the quotient manifold-with-boundary $\overline{\mathcal{Q}}_\mathcal{R}$, such that
    $ \overline{\Gamma}(\alpha, 0)=\lambda(\alpha) $ (i.e. all geodesics start on the 'initial curve' $\lambda$) and $ v(\overline{\Gamma}(\alpha, T(\alpha))) = 0 $ for all $ \alpha\in (-a, a)$ (i.e. all geodesics end on $\mathcal{CH}^+$)  
    \item Furthermore, at each point $\lambda(\alpha)$ on the curve $\lambda$ we have 
    $\partial_\alpha\overline{\Gamma}(\alpha, 0) = \lambda'$ and $\langle\partial_\tau \overline{\Gamma}(\alpha, 0),\lambda' \rangle = 0$.  That is, the initial variational field of $\overline{\Gamma}$ equals that tangent to the initial curve $\lambda$ and is normal to the value of the geodesic velocity at $\lambda(\alpha)$, respectively.
     \item The function $T(\alpha)$ tells us the value of the proper time at which the geodesic $\overline{\Gamma}(\alpha, \cdot)$ reaches the Cauchy horizon $\{v=0\}$. 
    \item As $\lambda$ is a $(n_0, N_0, C_K)$--admissible curve, it is compactly contained in a coordinate rectangle close to $\mathcal{CH}^+$.
    \item Condition admissibility condition 5, equation (\ref{condd5})  of admissible curves implies that $\lambda$ cannot deform into a null curve at the endpoints.  Combining  admissibility condition 5 with Remark 2 above, we see that the $(u,v)$ coordinate components of the geodesic velocity are uniformly bounded away from zero on the initial curve $\lambda$ i.e. $0<c<|\nu^u|, |\nu^v|$ for some $c>0$.
    \item Through pullback using projection map $\overline{\mathcal{R}}\to\overline{\mathcal{Q}}_\mathcal{R}$, $\lambda$ and $\overline{\Gamma}$ are in one-to-one correspondence with a spacelike hypersurface Image$(\lambda)\times S^2$ and a variation through radial ($\overline{\Gamma}^\varphi, \overline{\Gamma}^\theta =$const), affine, timelike, normalized geodesics in $\overline{\mathcal{R}}$, all of which terminate on the Cauchy horizon.
    \item Showing that
    the image of the restriction of $\overline{\Gamma}$ to $D^0$ is contained in the quotient spacetime $\mathcal{Q_R}$ means showing that it does not intersect $\{v=0\}$.  
    \end{enumerate}
\end{rems}

\begin{lem}\label{no_closed_timelike_curves}
    There are no closed timelike curves in $(\mathcal{Q_R}, g_{\mathcal{Q_R}})$\footnote{It can be shown that the spacetime $(\mathcal{R}, g)$, and by extension the quotient spacetime $(\mathcal{Q_R}, g_{\mathcal{Q_R}})$, are globally hyperbolic. Nevertheless, we find the monotonicity argument above to be a more streamlined way to show that there are no closed timelike curves in $(\mathcal{Q_R}, g)$.}.
    \begin{proof}
        Indeed, $(\mathcal{Q_R}, g_{\mathcal{Q_R}})$ is equipped with a global double null coordinate chart $(u,v)$ and the $u,v$-components of timelike curves are either both strictly positive (if the curve is future-directed) or both strictly negative (if the curve is past-directed). Therefore along a timelike curve $u$ and $v$ are either both strictly increasing or strictly decreasing. It follows that a timelike curve in $\mathcal{Q_R}$ cannot be closed.
    \end{proof}
\end{lem}
\begin{lem}\label{image_of_Gamma_is_in_Q}
    Under the hypotheses of Theorem \ref{interior_regularity_of_gamma}, Image$(\Gamma)\subset\mathcal{Q_R}$. 
    \begin{proof}
      Let $\Gamma(0, \tau) = \gamma(\tau)$ and consider this geodesic. Using that Image$(\overline{\Gamma})\subset \overline{\mathcal{Q}}_\mathcal{R}$, showing that $\Gamma(D^0) \subset \mathcal{Q}_\mathcal{R}$ is easy. We just have to show that the restricted map $\Gamma$ does not intersect the Cauchy horizon. Since $\gamma$ is future-directed and timelike, it follows that $v$ is strictly increasing along $\gamma$, so $v(\gamma(\tau))<v(\gamma(T(0))) = 0$ for every $0\leq\tau<T(0)$. Therefore Image$(\Gamma)\cap \{v=0 \} = \emptyset$.
    \end{proof}
\end{lem}
\begin{rems}
\begin{enumerate}
     \item As a consequence of Lemma \ref{image_of_Gamma_is_in_Q}, $D^0=\{(\alpha,\tau)\in D: \tau>0, v(\overline{\Gamma}(\alpha,\tau))<0\} = \overline{\Gamma}^{-1}(\mathcal{Q_R})\cap I^+(\lambda)$. 
    \item The interior of $D$ equals $(D)=D\backslash\partial D$, where the boundary $\partial D$ includes the horizontal segment $(-a,a)\times\{0\}$ and the graph of the function $T(\cdot)$. Therefore, int$(D)\cap(-a,a)\times\{0\}=\emptyset$ and int$(D)\cap \{(\alpha,T(\alpha)):\alpha\in(-a,a)\}=\emptyset$. It follows that int$(D)\subseteq D^0 \subset D$.
\end{enumerate}
\end{rems}

    \subsubsection{Monotonicity of $T(\cdot)$}
    \begin{prop}
   \label{T_is_monotonic}
        Under the hypotheses of Theorem \ref{regularity_of_timelike_geodesic_variations}, if $(\lambda')^v(\alpha)>0$ then $T(\cdot)$ is strictly decreasing in $\alpha$. 
        \begin{proof}
             We prove this through a bootstrap argument. Wlog we work with the central geodesic $\Gamma(0, \cdot)$. Let $I$ be the interval defined by:
            \begin{align}
                I = \{\tau \in [0, T(0)] : \exists \delta>0  \text{ with }[-\delta,0]\times[0,\tau]\subset D \text{ and } -\delta \leq\alpha_1<\alpha_2\leq 0\nonumber\\
                \implies \Gamma^v(\alpha_1,\tau')< \overline{\Gamma}^v(\alpha_2,\tau') \ \forall \tau'\in[0, \tau] \}
            \end{align}
             \begin{rems}    
             \begin{enumerate}
            \item The interval $I$ is the set on which our bootstrap assumptions hold. The Proposition will be established by showing that the interval $I$ is nonempty, open and closed as a subset of $[0,T(0)]$. As a consequence, there is a vertical strip $[-\delta, 0)\times[0, T(0)]$ which is in $D^0$. All of the geodesics in this vertical strip reach the WNS after the geodesic $\Gamma(0, \cdot)$, so they all have proper time larger than the 'central' geodesic. Applying this argument to every $\alpha\in(-a,a)$, one finds that if $\alpha'<\alpha$ then $T(\alpha')>T(\alpha)$. Illustration of the setup is provided in Figure \ref{Tmonotonic}.
            \item Technically, we do not know a priori if $D^0$ is open, so we can only use the differentiability of $\Gamma$ in int$(D)$. When we compare the coordinates of the endpoints of a given contour in $D^0$, we are essentially differentiating under the integral sign. Therefore we must ensure that the contour minus its endpoints lies inside int$(D)$, or argue using limits of sequences which are inside int$(D)$.
            \end{enumerate}
            \end{rems}
            \underline{$I\neq \emptyset$.} Since Image$(\lambda) \subset D$ and since $(\lambda')^v>0$, $0\in I$ so $I \neq \emptyset$.\\[5pt]
            \underline{$I$ is open} (as a subset of $[0, T(0)]$). Indeed, suppose for contradiction that $I=[0, T]$ for some $T<T(0)$ (if $T=T(0)$ then $I=[0,T(0)]$ which is open). Then $[-\delta, 0]\times \{T\} \subset D^0$. Let 
            $\alpha\in[-\delta, 0]$ and define $T'=\inf_{\alpha'\in[0,\delta]}T(\alpha')$. By the bootstrap assumptions and Corollary \ref{unidorm_bounds_on_geod_velocities} we have
            \begin{align}
                0-\Gamma^v(0,T)\leq 0 - \Gamma^v(\alpha, T) = \int_{T}^{T(\alpha)} \partial_\tau \Gamma^v(\alpha,\tau)d\tau \leq e^{2C_g}\sqrt{\frac{C_g N_0}{n_0}} (T(\alpha)-T) \implies\nonumber\\
                T(\alpha)-T\geq -\Gamma^v(0,T)e^{-2C_g}\sqrt{\frac{n_0}{N_0 C_g}}>0 \implies\nonumber\\
                T' - T \geq -\Gamma^v(0,T)e^{-2C_g}\sqrt{\frac{n_0}{N_0 C_g}}>0 \label{0T_argument_1}
            \end{align}
            Then $(T,T')\neq \emptyset$ and $[0, \delta]\times [0,T')\subset D^0\subset D$. Since $(-\delta, 0)\times(0, T')$ is an open subset of $D^0$ and since int$(D)\subseteq D^0$, it follows that $(-\delta, 0)\times(0, T')$ is contained in the interior of $D$, where $\Gamma$ is known to be smooth. Therefore, Corollary \ref{orientation_preserving_corollary} can be applied to each of the horizontal segments $[-\delta, 0]\times \{\tau\}$ for $\tau \in [0, T')$, showing that $v$ is strictly increasing in $\alpha$ along each of those segments. This implies that every $\tau \in (T, T')$ is in $I$, contradicting the assumption that $I=[0,T]$. We conclude that $I$ must be open. \\[5pt]
            \underline{$I$ is closed} (as a subset of $[0, T(0)]$). Suppose for contradiction that $I=[0, T)$ for some $T>0$. We consider two cases.\\[5pt]
            Case 1: $T=T(0)$ then $I=[0, T(0))$ so there is a $\delta>0$ such that $[-\delta, 0]\times [0, T(0))\subset D$ and for $\tau\in[0,T(0))$, $\alpha\in[-\delta, 0]$ we have $\Gamma^v(\alpha,\tau)<\Gamma^v(0,\tau)<\Gamma^v(0, T(0))=0$. The second inequality follows because $v$ is increasing along the geodesic $\Gamma(0,\cdot)$. Therefore $[-\delta, 0]\times[0, T(0))$ is in $D^0$. But as $\Gamma^v( \alpha,\tau)<0$ for $(\alpha,\tau)\in[-\delta,0]\times[0,T(0))$, it follows that for $\alpha\in[-\delta,0]$ one has $\lim_{\tau \to T(0)} \Gamma^v(\alpha,\tau)\leq 0$. This implies that $[-\delta, 0]\times [0,T(0)] \subset D$. Since $(-\delta,0)\times(0,T(0))$ is an open subset of $D^0$, it must lie inside int$(D)$. Fix $\alpha\in[-\delta,0)$. Applying Corollary \ref{orientation_preserving_corollary} between the points $(0,\tau)$ and $(\alpha,\tau)$ for $\tau<T(0)$ and applying Corollary \ref{unidorm_bounds_on_geod_velocities} firstly between the points $(0,\tau)$ and $(0, T(0))$ and secondly between the points $(\alpha,\tau)$ and $(\alpha, T(0))$, we find
            \begin{align}
                &0-\Gamma^v(0, \tau)\geq e^{-2C_g}\sqrt{\frac{n_0}{C_g^3 N_0}}(T(0)-\tau);\nonumber\\
                &\Gamma^v(0,\tau) - \Gamma^v(\alpha,\tau) \geq n_{null} (0-\alpha); \nonumber \\
                & \Gamma^v(\alpha, T(0)) - \Gamma^v(\alpha, \tau) \leq e^{2C_g}\sqrt{\frac{C_g N_0}{n_0}}(T(0)-\tau)
            \end{align}
            It follows that 
            \begin{align}
                \Gamma^v(\alpha, T(0))\leq \bigg(e^{2C_g}\sqrt{\frac{C_g N_0}{n_0}} - e^{-2C_g}\sqrt{\frac{n_0}{C_g^3 N_0}}\bigg)(T(0) - \tau) - n_{null} \alpha
            \end{align}
            For large enough $\tau<T(0)$, the RHS is less than $-n_{null} \alpha/2$, so $(\alpha,T(0)) \in D^0$. Since $\alpha$ was arbitrary, it follows that $[-\delta, 0)\times \{T(0)\}\subset D^0$. We would like to apply Corollary \ref{orientation_preserving_corollary} to show that  $-\delta\leq \alpha_2\leq \alpha_1<0 \implies\Gamma^v(\alpha_1,T(0))>\Gamma^v(\alpha_2,T(0))$, but we have to be careful as we cannot guarantee that $(-\delta,0)\times\{T(0)\}\subset $ int$(D)$. Let $\{\tau_j\}_{j=0}^\infty \subset(0,T(0))$ be an increasing sequence converging to $T(0)$. Since the geodesics $\Gamma^v(\alpha_1,\cdot)$ and $\Gamma^v(\alpha_2, \cdot)$ are continuous,
            \begin{align}
                \lim_{j\to \infty} [\Gamma^v(\alpha_1,\tau_j)-\Gamma^v(\alpha_2,\tau_j)]= \Gamma^v(\alpha_1,T(0))-\Gamma^v(\alpha_2,T(0)). 
             \end{align}
            By Corollary \ref{orientation_preserving_corollary}, for each $j\geq 0$ we have 
            \begin{align}
                \Gamma^v(\alpha_1,\tau_j)-\Gamma^v(\alpha_2, \tau_j)\geq  n_{null}(\alpha_1-\alpha_2)
            \end{align}
            Taking the limit as $j\to \infty$, it follows that $\Gamma^v(\alpha_2,T(0))<\Gamma^v(\alpha_1, T(0))$ i.e. $v$ is increasing in $\alpha$ along the segment $[-\delta, 0)\times\{T(0)\}$. Since $\overline{\Gamma}^v(0, T(0))=0$, $v$ increases with $\alpha$ on the closed segment $[-\delta, 0]\times\{T(0)\}$. Therefore $T(0)\in I$, contradicting the assumption that $I$ is open. 
            \\[5pt]
            Case 2: $I=[0, T)$ where $T<T(0)$. Then for every $\tau<T$ there exists a $\delta>0$ such that $[-\delta,0]\times[0,\tau]\subset D$ and $-\delta\leq \alpha<0 \implies \Gamma^v(\alpha, \tau)< \Gamma^v(0,\tau)$. As $v$ is increasing along the geodesic $\Gamma(0, \cdot)$, we also have $\Gamma^v(0,\tau)<\Gamma^v(0,T)<0$. So for each $\alpha\in[-\delta, 0]$, $\Gamma^v(\alpha, \tau)$ is strictly bounded away from zero. By the continuity of the $v$-coordinate along each geodesic between $\alpha=-\delta$ and $\alpha=0$, it follows that $\overline{\Gamma}^v(\alpha, T)<0$ for $-\delta\leq \alpha\leq 0$, so the segment $[-\delta, 0]\times \{T\}$ must lie in $D^0$. As $(-\delta,0)\times (0,T)$ is open, it lies in the interior of $D$. Furthermore, bootstrap assumptions hold for $0<\tau<T$. Let ${\tau_j}\subset(0,T)$ be an increasing sequence converging to $T$. By Corollary \ref{orientation_preserving_corollary}, for $-\delta\leq\alpha\leq 0$ we have
            \begin{align}
                \Gamma^v(0, \tau_j)-\Gamma^v(\alpha, \tau_j)\geq n_{null}(0-\alpha)>0
            \end{align}
            By continuity along each geodesic,
            \begin{align}
                \lim_{j\to \infty} \Gamma^v(0, \tau_j)-\Gamma^v(\alpha, \tau_j) =\Gamma^v(0,T)-\Gamma^v(\alpha,T)
            \end{align}
            we conclude that the bootstrap assumptions are also satisfied at $T$. This implies $T\in I$, so $I=[0,T]$ which is closed. \\[5pt]
            We conclude that $I$ is nonempty, open and closed, so $I=[0,T(0)]$. This implies that there exists a $\delta>0$ such that $[-\delta, 0]\times[0, T(0)]\subset D$ and $[-\delta, 0]\times(0, T(0)]/\{(0,T(0))\}\subset D^0$. Therefore, at time $T(0)$, the geodesics $\Gamma(\alpha, \cdot)$ for $\alpha\in[-\delta, 0)$ have not reached the singularity, so $T(\alpha)>T(0)$ for  $\alpha\in[-\delta, 0)$. This argument was applied wlog to the 'central' geodesic $\Gamma(0, \cdot)$. Applying it to every other $\alpha\in(-a,a)$ yields that if $\alpha'<\alpha$ then $T(\alpha')>T(\alpha)$, hence concluding the proof.
        \end{proof} 
    \end{prop}
    By reversing the direction of the tangent to the admissible curve $\lambda$ and repeating the proof of Proposition \ref{T_is_monotonic}, we also establish:
\begin{cor}\label{T_is_monotonic_reverse_dirn_corollary}
         Under the hypotheses of Theorem \ref{regularity_of_timelike_geodesic_variations}, if $(\lambda')^v(\alpha)<0$ then $T(\cdot)$ is strictly increasing in $\alpha$. 
    \end{cor}
    \begin{figure}
            \centering
        \includegraphics[width=0.3\linewidth]{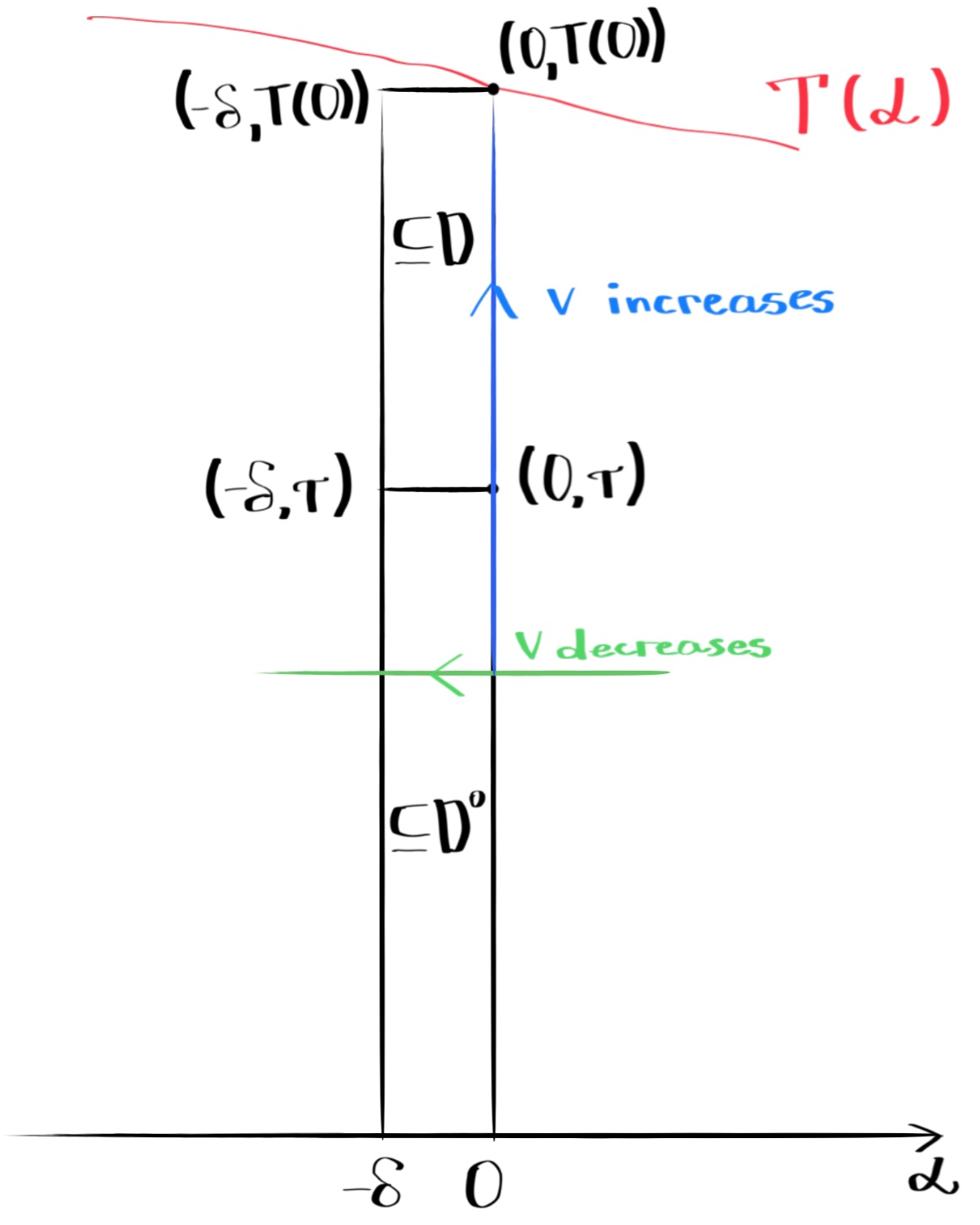}
            \caption{Illustration of the setup for the bootstrap argument used to prove Proposition \ref{T_is_monotonic}.}
            \label{Tmonotonic}
        \end{figure}
\subsubsection{Continuous dependence of $T(\cdot)$ on $\alpha$}
In the following Proposition we prove that the proper time along the geodesics depends continuously on the parameter $\alpha$, indexing the geodesics.
\begin{prop} \label{contn_dep_of_T_on_alpha}
    Under the hypotheses of Theorem \ref{regularity_of_timelike_geodesic_variations}, the final proper time function $T=T(\alpha)$ is a continuous function of the parameter $\alpha$ along the initial admissible curve $\lambda$.  
   \begin{proof}
       Wlog assume $(\lambda')^{v}(\alpha)>0 \ \forall \alpha\in(-a, a)$ and let $\alpha_*\in(-a,a)$. Let
       \begin{align}
            T_+:=\lim_{\alpha\to \alpha_*^-} T(\alpha)\geq \lim_{\alpha\to \alpha_*^+} T(\alpha) =: T_-
       \end{align}
       The left and right limits exist because $T(\cdot)$ is monotonically decreasing in $\alpha$ by Proposition \ref{T_is_monotonic}. By the same Lemma $T_-\leq T(\alpha_*)\leq T_+$.\\[5pt]
       Let $\delta_1>0$ be sufficiently small so that the straight line segment 
         \begin{align}
           \mathcal{C}_{L,\delta}: [0,1)\to \mathbb{R}^2,s\mapsto \begin{pmatrix}
               (\alpha_*-\delta)(1-s)  + \alpha_* s\\
               (T_- - \delta)(1-s) + T_+ s
           \end{pmatrix}
       \end{align}
       lies in $D^0$ for every $\delta\in (0, \delta_1]$. Such $\delta_1$ can be chosen in view of the proof of Proposition \ref{T_is_monotonic}. Then since $C_{L,\delta}\backslash\{(\alpha_*,T)\} \subset(\alpha_*-\delta_1, \alpha_*)\times(0,T_+)$, the latter being an open subset of $D$,  $C_{L,\delta}\backslash\{(\alpha_*,T)\}\subset$ int$(D)$ for every $\delta\in(0,\delta_1]$.
        Let $0<\mu\leq 1$ be sufficiently small so that the horizontal straight line segment 
       \begin{align}
           \mathcal{C}_{H,\delta}:[\alpha_*-\delta,\alpha_*+\mu\delta]\to \mathbb{R}^2, \alpha\to \begin{pmatrix}
               \alpha\\
               T_- - \delta
           \end{pmatrix}
       \end{align}
       is inside $D^0$ for every $\delta\in(0,\delta_1]$. Since $\mathcal{C}_{H,\delta}\backslash\{(\alpha_*+\mu\delta, T_--\delta)\} \subset (\alpha_*-\delta, \alpha_*+\mu\delta)\times (0,T_--\delta/2)$, the latter being an open subset of $D$, it follows that $\mathcal{C}_{H,\delta}\backslash\{(\alpha_*+\mu\delta, T_--\delta)\} \subset$ int$(D)$ for every $\delta\in(0,\delta_1]$. We assume that $\delta\leq \delta_1$ and that $\mu$ satisfies the above condition so that $\mathcal{C}_{H,\delta}\cup \mathcal{C}_{L,\delta}\subset D^0$ and $\mathcal{C}_{H,\delta}\cup \mathcal{C}_{L,\delta}\backslash\{(\alpha_*+\mu\delta, T_--\delta)\}\subset$ int$(D)$.\\[5pt]
       Firstly, we claim that $\lim_{s\to 1^-} \Gamma^v(\mathcal{C}_{L,\delta}(s)) = 0$.
        To prove this claim, let $\varepsilon>0$ be arbitrary and note that 
        \begin{align}
            \exists s_1 \in[0,1) \text{ s.t. } s_1<s<1 \implies 0 < T\big((\alpha_*-\delta)(1-s)+\alpha_* s \big) - T_+ <\varepsilon
        \end{align}
        where we used the monotonicity of $T(\cdot)$, established in Proposition \ref{T_is_monotonic}. So let $s\in(s_1,1)$. Integrating $\partial_\tau \Gamma^v$ along the geodesic $\Gamma((\alpha_*-\delta)(1-s)+\alpha_* s, \cdot)$ between the time $\tau=(T_- - \delta)(1-s) + T_+ s$ and the proper time $\tau=T((\alpha_*-\delta)(1-s)+\alpha_* s)$ when this geodesic reaches the WNS, 
        \begin{align}
            0-\Gamma^v(\mathcal{C}_{L,\delta}(s)) &= - \Gamma^v\big((\alpha_*-\delta)(1-s)+\alpha_* s, (T_- - \delta)(1-s) + T_+ s\big) \nonumber\\ 
            &\leq \sqrt{\frac{C_g N_0}{n_0}}e^{2C_g} \big[ T((\alpha_*-\delta)(1-s)+\alpha_* s) - (T_- - \delta)(1-s) - T_+ s \big] \nonumber \\
            &= \sqrt{\frac{C_g N_0}{n_0}}e^{2C_g} \big[ \underbrace{T((\alpha_*-\delta)(1-s)+\alpha_* s) - T_+}_{<\varepsilon} + (1-s)(T_+ -T_- + \delta) \big] \nonumber\\
            &<\sqrt{\frac{C_g N_0}{n_0}}e^{2C_g} \big[ \varepsilon + (1-s)(T_+ -T_- + \delta) \big]
        \end{align}
        where we used the uniform bound on $\partial_\tau \Gamma$ from Corollary \ref{unidorm_bounds_on_geod_velocities} and in the third line we added and subtracted $T_+$. We conclude that $\Gamma^v(\mathcal{C}_{L,\delta}(s))\to 0$ as $s\to 1^-$. 
       \\[5pt]
       Having established this claim, we apply Corollary \ref{orientation_preserving_corollary} to compare the $v$-coordinates of the endpoints $\Gamma^v(\alpha_* + \delta, T_- - \delta)$ and $\Gamma^v(\alpha_* - \delta, T_- - \delta)$ of the horizontal segment $\mathcal{C}_{H,\delta}$, 
   \begin{align}
       \Gamma^v(\alpha_*+\delta, T_- - \delta) - \Gamma^v(\alpha_*-\delta, T_- - \delta) = \int_{\alpha_* - \delta}^{\alpha_* + \mu\delta} \partial_\alpha\Gamma^v(\alpha, T_- -\delta) d\alpha \leq (1+\mu)N_{null} \delta\leq 2 N_{null} \delta\label{v_distance_along_horizontal}
   \end{align}
 By Proposition \ref{T_is_monotonic}, $T_-> T(\alpha_*+\mu\delta)$. Integrating along the geodesic $\Gamma(\alpha_*+\mu\delta, \cdot)$ between $T_--\delta$ and its proper time $T(\alpha_*+\mu\delta)$ allows us to bound the $v$-coordinate of the right endpoint $(\alpha_*+\mu\delta, T_--\delta)$ of the horizontal contour $\mathcal{C}_{H,\delta}$.
   \begin{align}
       -\Gamma^v(\mathcal{C}_{H,\delta}(\alpha_*+\mu\delta))=-\Gamma^v(\alpha_*+\mu\delta, T_- - \delta) \leq e^{2C_g}\sqrt{\frac{C_g N_0}{n_0}} \big(T(\alpha_*+\mu\delta)-T_- + \delta \big)<e^{2C_g}\sqrt{\frac{C_g N_0}{n_0}}\delta
       \label{v_distance_right_contour}
   \end{align}
   where we applied Corollary \ref{unidorm_bounds_on_geod_velocities} to get the uniform upper bound on $\partial_\tau\Gamma$.\\[5pt]
   Comparing the $v$-coordinates of the endpoints of the contour $\mathcal{C}_{L,\delta}$, we find 
   \begin{align}
       \underbrace{\lim_{s\to 1^-} \Gamma^v\big(\mathcal{C}_{L,\delta}(s)\big)}_{=0} - \Gamma^v\big(\mathcal{C}_{L,\delta}(0)\big) &= \int_{0}^{1} \partial_s \Gamma^v\big((\alpha_* - \delta)(1-s) + \alpha_* s, (T_--\delta)(1-s) +T_+ s\big)ds  \nonumber\\
       &=\int_0^1 \delta\partial_\alpha \Gamma^v\big(\mathcal{C}_{L,\delta}(s)\big) + (\delta+T_+-T_-)\partial_\tau \Gamma^v\big(\mathcal{C}_{L,\delta}\big(s)) ds \nonumber\\
       &\geq \delta\bigg[n_{null} + \bigg(1+\frac{T_+-T_-}{\delta}\bigg)\sqrt{\frac{n_0}{C_g^3 N_0}}e^{-2C_g}\bigg]\label{v_distance_left_contour}
   \end{align}
   where we applied Corollary \ref{orientation_preserving_corollary} and Corollary \ref{unidorm_bounds_on_geod_velocities} to get the uniform lower bounds on $\partial_\alpha\Gamma$ and on $\partial_\tau\Gamma$.\\[5pt]
   Combining the estimate (\ref{v_distance_left_contour}) with (\ref{v_distance_along_horizontal}) and (\ref{v_distance_right_contour}),
   \begin{align}
       &n_{null} + \bigg(1+\frac{T_+-T_-}{\delta}\bigg)\sqrt{\frac{n_0}{C_g^3 N_0}}e^{-2C_g} \leq \frac{1}{\delta}\big[\Gamma^v(\mathcal{C}_{H,\delta}(\alpha_*+\mu\delta)) - \Gamma^v(\mathcal{C}_{L,\delta}(0)) - \Gamma^v(\mathcal{C}_{H,\delta}(\alpha_*+\mu\delta)) \big] \nonumber\\
       &\leq\frac{1}{\delta} \bigg[2N_{null} \delta + \delta \sqrt{\frac{C_g N_0}{n_0}}e^{2C_g} \bigg] \implies\nonumber\\
       &T_+-T_- \leq \delta\bigg(\sqrt{\frac{C_g^3 N_0}{n_0}} e^{2C_g}(2N_{null}-n_{null}) + \frac{C_g^2 N_0}{n_0}e^{4C_g} -1\bigg) \to 0 \ \ \text{ as }\ \  \delta\to 0
   \end{align}
   We conclude that $T_+=T_-$. By the squeeze theorem, they both equal $T(\alpha_*)$. 
       
   \end{proof}
   \end{prop}
   \begin{figure}[h]
       \centering
       \includegraphics[width=0.4\linewidth]{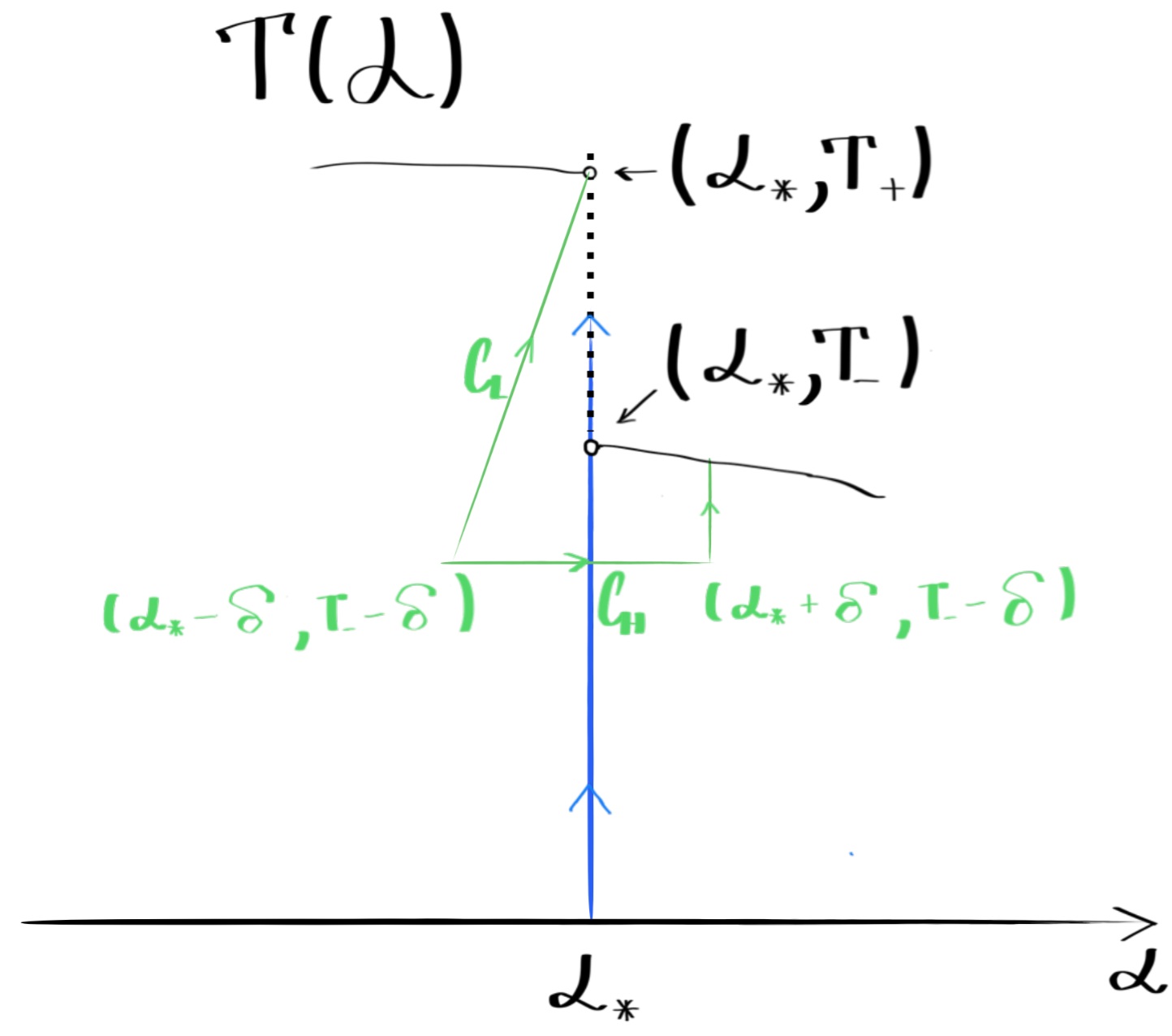}
       \caption{Illustration of the proof of Proposition \ref{contn_dep_of_T_on_alpha}. The $v$-component of the right endpoint of the horizontal contour $\mathcal{C}_{H,\delta}$ is $>-C\delta$ because the point is within $\delta$ vertical distance away from the graph of $T(\cdot)$. The difference between the $v$-coordinates of the two endpoints of $\mathcal{C}_{H,\delta}$ is $<C\delta$ because they are within $2\delta$ horizontal distance of each other. Finally, by comparing the $v$-coordinates of the endpoints of $\mathcal{C}_{L,\delta}$ we can bound the difference $T_+-T_-$.}
       \label{Cont_dep_T}
   \end{figure}
    \begin{rem}
 Since $D^0\subset D$, $D^0\backslash$int$(D)\subseteq\partial D$. Since $T(\cdot)$ is a continuous function by Proposition \ref{contn_dep_of_T_on_alpha}, 
         \begin{align}
            \partial D =& [-a, a]\times\{\tau=0\} \cup \{\alpha=-a\}\times[0, \lim_{\alpha\to -a^+}T(\alpha)]\cup \{\alpha=a\}\times [0, \lim_{\alpha\to a^-}T(\alpha)]\cup\nonumber\\
            &\{(\alpha,T(\alpha)):\alpha\in(-a,a)\}
        \end{align}
        is a continuous closed loop in $\mathbb{R}^2$ and by construction $D^0\cap \partial D=\emptyset$. 
        So $D^0=$ int$(D)$.
   \end{rem} 
    \subsubsection{$\Gamma$ is a \textit{local} diffeomorphism in the interior}
      \begin{thm}\label{Gamma_local_diffeomorphism}
          Under the hypotheses of Theorem \ref{interior_regularity_of_gamma}, $\Gamma$ is a \textit{local diffeomorphism} onto its image.
      \end{thm}
       \begin{proof}
    Without loss of generality, consider the 'central' geodesic $\gamma = \Gamma(0, \cdot)$. The variational field of $\Gamma$ along $\gamma$ has components in the double null basis given by:
      \begin{equation}
          J^\mu (0, \tau)= \frac{\partial}{\partial \alpha}\Gamma^\mu(0,\tau)
      \end{equation}
      for $\mu\in\{u, v\}$ where $\Gamma^\mu = x^\mu \circ\Gamma$ are the coordinates of the geodesic variation in the double null chart. Hence at $\tau=0$:
      \begin{equation}
          J^\mu(0,0) = \frac{d\lambda}{d\alpha}(0)\label{determine_J_initial}
      \end{equation}
      The tangent to $\gamma$ at $\tau=0$ has components in this coordinate basis
      \begin{equation}
     \dot{\gamma}^\mu (\tau) = \frac{\partial}{\partial \tau} \Gamma^\mu (0,\tau=0) = \nu^\mu(\alpha=0)
      \end{equation} 
      The components of the normal derivative to the variational field $J$, $\nabla_{\dot{\gamma}} J$ at $(\alpha=0, \tau=0)$ are given by:
      \begin{align}
          \nabla_{\dot{\gamma}} J^\mu (0,0) = \partial_\tau \partial_\alpha {\Gamma}^\mu(0,0) +\mathring{\Gamma}^\mu_{\rho\sigma} \dot{\gamma}^\rho \partial_\alpha \Gamma^\sigma (0,0) = \nabla_{J(0,0)} \partial_\tau{\Gamma}^\mu (0,0) =  \nabla_{\lambda'(0)}\nu^\mu (0) \label{determine_normal_derivative}
      \end{align}
       Where the Christoffel symbols $\mathring{\Gamma}^\mu_{\nu\rho}$ of $g_\mathcal{Q_R}$ are listed in Appendix \ref{Christoffel_Symbols}. As described in Theorem \ref{bounds_for_jacobi_field_thm}, $J(0, \cdot)$ satisfies the Jacobi equation with initial data given by (\ref{determine_J_initial}) and (\ref{determine_normal_derivative}). Spacetime is smooth away from the Cauchy horizon, so this IVP for the Jacobi equation admits a unique smooth solution\footnote{This is because the geodesic velocity components and the Riemann tensor components are smooth away from the Cauchy horizon, and the Jacobi equation is a second-order linear ODE for the components of $J$.} for $0<\tau<T(\alpha=0)$. At $\tau=0$ the Jacobi field is normal to the geodesic $\gamma$. Since  $0=\nabla_{\lambda'} \langle\nu,\nu\rangle=2\langle\nabla_{\lambda'} \nu, \nu\rangle$, at $\tau=0$ the normal derivative of the Jacobi field is also normal to the geodesic $\gamma$. Hence, and as we work on the quotient manifold\footnote{Results obtained for the geodesic variation on $\mathcal{Q_R}$ are then lifted to the 3+1 manifold $\mathcal{R}$ where the fluid lives.} $\mathcal{Q}_\mathcal{R}$, by Equation (\ref{JAD_0}) both the Jacobi field $J(0,\tau)$ and its derivative $\nabla_{\dot{\gamma}} J(0, \tau)$ remain radial and normal to the geodesic $\gamma$ for all $0<\tau<T(0)$. Referring again to the parallelly transported orthonormal basis $(e_0:=\dot{\gamma}, e_1)$ from (\ref{ON_BASIS}), we can write the initial data for the Jacobi equation along $\gamma$ as $J(0,0)=J^1(0,0) e_1 = \lambda'(0)$ and  $\nabla_{\dot{\gamma}}J(0,0) = \big\langle\nabla_{\lambda'(0)} \nu(0), e_1\big\rangle e_1$ in the orthonormal basis at $\gamma(0)$. As $\lambda$ is a $(n_0, N_0, C_K)$--admissible curve, the hypotheses of Theorem \ref{bounds_for_jacobi_field_thm} are satisfied. Hence by Theorem \ref{bounds_for_jacobi_field_thm} there exist $0<n_{\perp,0}<N_{\perp,0}<\infty$ such that
      \begin{equation}
          n_{\perp,0}\leq |J^1(0,\tau)|\leq N_{\perp,0} \ \forall\  0<\tau<T(0)
      \end{equation}
      Using this property, we can employ the inverse function theorem to show that for arbitrary $\tau \in (0, T(0))$, the point $(0,\tau) \in D^0$ has an open neighbourhood $V\subset \mathbb{R}^2$ such that the restriction $\Gamma|_V: V \to \Gamma(V)\subset \mathcal{Q_R}$ is a diffeomorphism. By the inverse function theorem, this holds if and only if the linearisation $D\Gamma(0, \tau): T_{(0,\tau)}D^0\to T_{\Gamma(0,\tau)}\mathcal{Q}_\mathcal{R}$, of $\Gamma$, is invertible at $(0,\tau)$. Since coordinate charts are diffeomorphisms, this occurs if and only if the Jacobian matrix, representing the linearisation of $\Gamma$ in the double null chart $\{u,v\}$ is invertible at $(0,\tau)$. Indeed,
      \begin{align}
          [D\Gamma(0, \tau)]_i^\mu =
          \begin{pmatrix}
          \frac{\partial \Gamma^u}{\partial \alpha} & \frac{\partial \Gamma^v}{\partial \alpha} \\
          \frac{\partial \Gamma^u}{\partial \tau} & \frac{\partial \Gamma^v}{\partial \tau} 
          \end{pmatrix}\Bigg|_{\alpha=0} = 
          \begin{pmatrix}
              J^u(0,\tau) & J^v(0,\tau)\\
              \dot{\gamma}^u(\tau) & \dot{\gamma}^v(\tau)
          \end{pmatrix}
      \end{align}
      where $i$ stands for either $\tau$ or $\alpha$. We can transform this matrix into the parallelly transported orthonormal basis $(e_A)$ along $\gamma$ as in Theorem \ref{bounds_for_jacobi_field_thm},
      \begin{align}
           \begin{pmatrix}
              J^u(0,\tau) & J^v(0,\tau)\\
              \dot{\gamma}^u(\tau) & \dot{\gamma}^v(\tau)
          \end{pmatrix} = 
          \begin{pmatrix}
               (\Lambda^{-1})^u_A J^A(0,\tau) & (\Lambda^{-1})^v_A J^A(0,\tau)\\
              (\Lambda^{-1})^u_A \dot{\gamma}^A(\tau) & (\Lambda^{-1})^v_A \dot{\gamma}^A(\tau)
          \end{pmatrix} = [(\Lambda^{-1})^\mu_A][D\Gamma(0,\tau)]^A_i
      \end{align}
      where $\Lambda$ is the basis transformation matrix in (\ref{ON_basis_relations}). Note that 
      \begin{align}
              [D\Gamma(0,\tau)]^A_i = 
              \begin{pmatrix}
                  \dot{\gamma}^0(\tau) & J^0(\tau) \\
                  \dot{\gamma}^1(\tau) & J^1(\tau)
              \end{pmatrix} = 
              \begin{pmatrix}
                  1 & 0\\
                  0 & J^1(\tau)
              \end{pmatrix}
      \end{align}
      This is invertible, because $J^1(0, \cdot)$ is bounded away from zero on the image of $\gamma$. Since basis transformation matrices are invertible, it follows that $[D\Gamma]^\mu_i$ is invertible at $(0, \tau)$. Therefore by the inverse function theorem, each point in $\{0\}\times (0,T(0))$ has an open neighbourhood $V\subset D^0$ so that $\Gamma|_V: V \to \Gamma(V) \subset \mathcal{Q}_\mathcal{R}$ is a diffeomorphism. \\[5pt]
       Note that condition 3 of admissibility (Definition \ref{define_admissible_curve}) guarantees that the infimum $\inf_{\alpha} n_{\perp,\alpha}=n_\perp$ is positive and that the supremum $\sup_{\alpha} N_{\perp,\alpha} = N_\perp$ is finite. So the bounds on the components of the variational field can be taken to be uniform across the geodesic variation. This means that one can apply the above argument for each geodesic $\Gamma(\alpha, 0)$ with $\alpha\neq 0$ in the variation, proving that $\Gamma$ is locally a diffeomorphism onto its image.
    \end{proof}
   \subsubsection{$\Gamma$ is {injective in the interior}}
       In order to show that $\Gamma$ is \textit{globally} a diffeomorphism onto its image, we need to show that it is injective. The proof of this requires using the monotonicity of the coordinates $u$ and $v$ along timelike and spacelike curves (specific choices of these curves are defined below). For this argument to be executed, we need to ensure that the curves in question lie in the domain $D^0$ of $\Gamma$. Indeed, this follows from the fact that $D^0$ is a convex subset of $\mathbb{R}^2$, which is itself a consequence of the monotonicity of $T(\cdot)$, Proposition \ref{T_is_monotonic}.    
    \begin{thm}\label{injectivity_of_Gamma}
        Under the hypotheses of Theorem \ref{interior_regularity_of_gamma}, $\Gamma$ is injective.
    \end{thm}
    \begin{proof}
    Having shown monotonicity of $T(\cdot)$, we ensure that every straight line with endpoints in $D^0$ lies entirely in $D^0$. 
      To show that $\Gamma$ is injective, let $\Gamma(\alpha_p,\tau_p) = p, \Gamma(\alpha_q, \tau_q) = q$ with $(\alpha_p, \tau_p)\neq (\alpha_q,\tau_q) \in D^0$. We consider the two cases: (a) $\alpha_p = \alpha_q$; and (b) $\alpha_p \neq \alpha_q$\\[5pt]
      Case (a) is near obvious: since $\Gamma(\alpha, \cdot)$ is future-directed timelike for each $\alpha$, if (wlog) $\tau_p<\tau_q$ and $\alpha_p = \alpha_q$ then $p=\Gamma(\alpha_p,\tau_p) \ll \Gamma(\alpha_q,\tau_q)=q$\footnote{Notation: $p\ll q \iff q \in I^+(p)$.}. By Lemma \ref{no_closed_timelike_curves} there are no closed timelike curves in $\mathcal{Q_R}$, so $p\neq q$.\\[5pt]
      Case (b): $\alpha_p \neq \alpha_q$. Suppose $\alpha_p<\alpha_q$ and $\tau_p \leq \tau_q$. The cases $\alpha_p>\alpha_q$ and/or $\tau_p\geq\tau_q$ are treated in an identical way. In the $(\alpha,\tau)$-plane, let $\mathcal{C}_\perp$ denote the horizontal straight line segment which connects $(\alpha_p, \tau_p)$ to $(\alpha_q, \tau_p)$, and let $\mathcal{C}_{||}$ be the (possibly constant\footnote{If $\tau_p=\tau_q$ then $\mathcal{C}_{||} = \{(\alpha_q, \tau_q) \}$ is a point.}) vertical straight line segment connecting $(\alpha_q, \tau_p)$ to $(\alpha_q, \tau_q)$. See Figure \ref{Gamma_Map_fig}.
     \begin{figure}[h]
         \centering
         \includegraphics[width=0.9\linewidth]{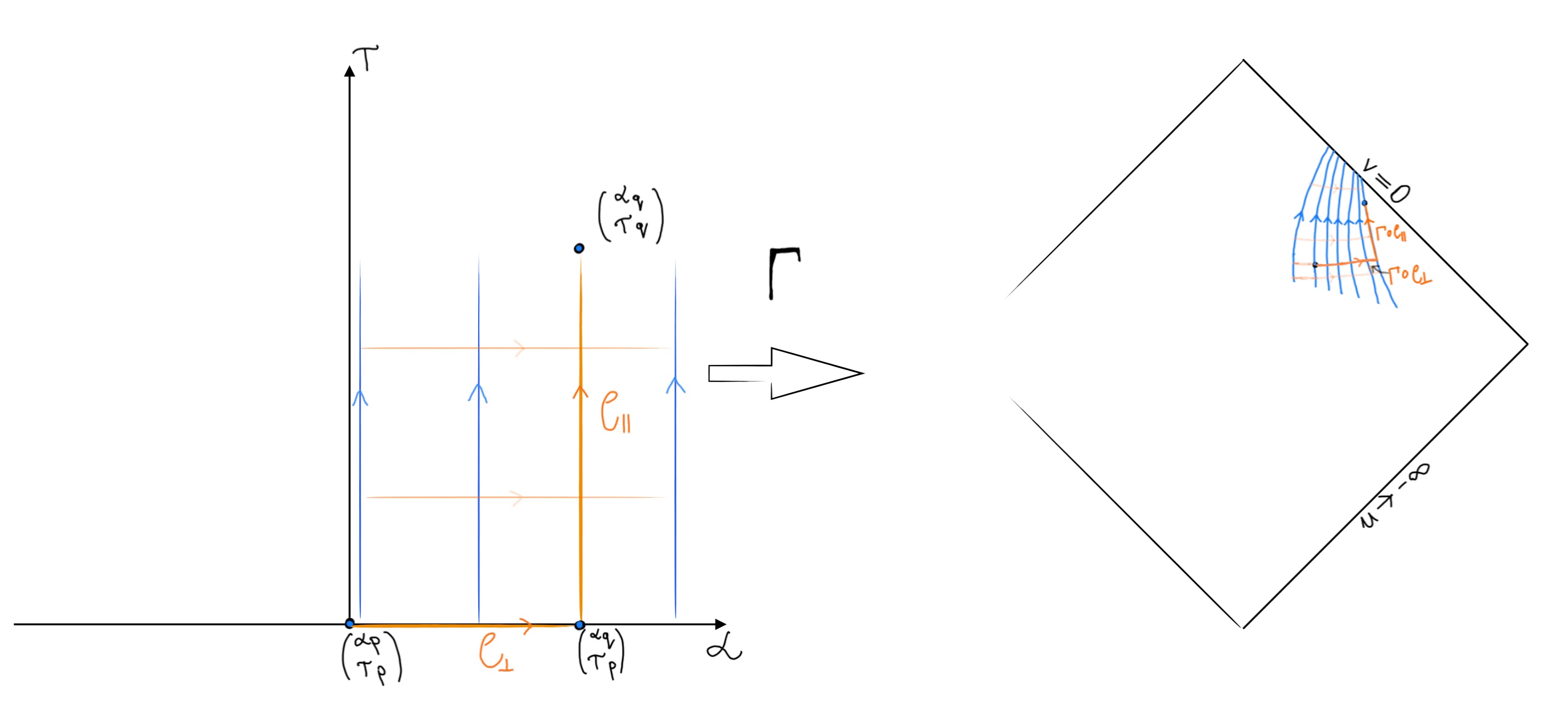}
         \caption{Images of the contours $\mathcal{C}_\perp$ and $\mathcal{C}_{||}$ under the map $\Gamma$. The image of $\mathcal{C}_{||}$ is a segment of a geodesic of $\Gamma$, and the image of $\mathcal{C}_{\perp}$ is a segment of an integral curve of the Jacobi field.}
         \label{Gamma_Map_fig}
     \end{figure}
     Then the contour $\mathcal{C}_\perp$ can be parameterized by 
     \begin{equation}\label{par1}
         [0, 1] \to \mathbb{R}^2, t \mapsto \begin{pmatrix}
             (1-t)\alpha_p + t \alpha_q \\ \tau_p
         \end{pmatrix}
     \end{equation}
     and the contour $\mathcal{C}_{||}$ can be parameterized by 
     \begin{equation}\label{par2}
         [0, 1] \to \mathbb{R}^2, t \mapsto \begin{pmatrix}
             \alpha_q \\ (1-t)\tau_p + t \tau_q
         \end{pmatrix}
     \end{equation}
     As a consequence of Lemma  \ref{T_is_monotonic}, and since the boundary of $D^0$ consists of the straight lines $\{\alpha = -a\}, \{\alpha=a\}$, $\{\tau=0\}$ and the graph of $T(\cdot)$, it follows that $D^0$ is a convex subset of $\mathbb{R}^2$. Consequently, the contour $\mathcal{C}_\perp$ lies in the domain of $\overline{\Gamma}$. With the parameterization given by (\ref{par1}-\ref{par2}),
     \begin{enumerate}
         \item \label{contour1} $\Gamma \circ \mathcal{C}_\perp, t \mapsto \Gamma((1-t)\alpha_p + t \alpha_q, \tau_p)$ is a continuous segment of the integral curve $\Gamma(\cdot, \tau_q)$ of the geodesic deviation field $J$ which passes through $p\equiv \Gamma(\alpha_p, \tau_p)$. 
         \item  \label{contour2} $\Gamma \circ \mathcal{C}_{||}, t \mapsto \Gamma(\alpha_q, (1-t)\tau_p + t \tau_q)$ is either a continuous segment of a geodesic $\Gamma(\alpha_q, \cdot)$ of $\Gamma$ which passes through $q \equiv \Gamma(\alpha_q, \tau_q)$, or a single point. 
     \end{enumerate}
     We claim that $\Gamma(\alpha_p, \tau_p) \neq  \Gamma(\alpha_q, \tau_p)$. Indeed, as the integral curves of $J$ are spacelike (so $g(J,J)= -e^{2\omega} J^u J^v > 0 \implies J^u J^v<0$), along the segment $\Gamma \circ \mathcal{C}_\perp$ the following holds for the double null coordinates $u$ and $v$: either $u$ is strictly decreasing and $v$ is stricty increasing or $u$ is strictly increasing and $v$ is strictly decreasing.  We can wlog assume that $v$ is increasing (becoming less negative) and $u$ is decreasing (becoming more negative). The opposite case is obtained by flipping the sign of $J$, or by assuming $\alpha_q<\alpha_p$. We find that 
     \begin{equation}
         \Gamma^v(\alpha_q, \tau_p)> \Gamma^v(\alpha_p, \tau_p) \text{ and } \Gamma^u (\alpha_q, \tau_p)<\Gamma^u(\alpha_p, \tau_p) \label{comparison1}
     \end{equation}
     Next, \ref{contour2}. implies that either $\tau_p = \tau_q$, or the curve $\Gamma \circ \mathcal{C}_{||}$ is future-directed timelike -  hence injective. In each case, along $\gamma$ both $u$ and $v$ are non-decreasing by future-directedness. Therefore 
     \begin{equation}
         \Gamma^v(\alpha_q, \tau_p)\geq \Gamma^v(\alpha_q,\tau_q) \text{  and  } \Gamma^u(\alpha_q, \tau_p)\leq\Gamma^u(\alpha_q, \tau_q) \label{comparison2}
     \end{equation}
     By (\ref{comparison1}) and (\ref{comparison2}) it follows that 
     \begin{equation}
         v(p) = v(\Gamma(\alpha_p, \tau_p))<v(\Gamma(\alpha_q, \tau_q)) = v(q) \implies p \neq q
     \end{equation}
     Concluding that the map $\Gamma$ is injective, hence a (global) diffeomorphism onto its 2-dimensional image. 
       \end{proof}

\subsubsection{$\overline{\Gamma}$ is {injective on the boundary}}
\begin{prop} \label{boundary_regularity_of_gamma}
    Let $\overline{\Gamma}: {D}\to \overline{\mathcal{Q}}_\mathcal{R}$ be as in Theorem \ref{interior_regularity_of_gamma}. We claim that the boundary map $\Gamma_{\mathbf{b}}: \alpha \to \Gamma(\alpha, T(\alpha))\in \mathcal{CH}^+$ is injective. 
    \begin{proof}
        Wlog we assume that $(\lambda')^v>0$ and $(\lambda')^u<0$. Then the component $J^u=\partial_\alpha\Gamma^u$ of the Jacobi field is negative everywhere on Image$(\Gamma)\subseteq \mathcal{Q}_\mathcal{R}$. Let $\alpha, \beta$ satisfy $-a<\beta<\alpha<a$. By Proposition \ref{T_is_monotonic}, $T(\beta)> T(\alpha)$. Let $\tau_\delta = T(\alpha)- \delta$ for $\delta>0$ small as we wish.
         By definition of $\Gamma_{\mathbf{b}}$, $\Gamma_\mathbf{b}^v(\alpha) = \Gamma_\mathbf{b}^v(\beta) = 0$, so we have to show that $\Gamma_\mathbf{b}^u(\alpha) \neq \Gamma_\mathbf{b}^u(\beta)$. The difference between the $u$-coordinates of the boundary values of the two geodesics can be written as:
        \begin{align}
        &\Delta\Gamma_{\mathbf{b}}^u[\alpha;\beta]:=\Gamma(\beta, T(\beta)) - \Gamma(\alpha, T(\alpha)) = \nonumber \\[3pt]
            & =\underbrace{\int_{\tau_\delta}^{T(\alpha)}\big[ \partial_\tau{\Gamma}^u(\beta, \tau) - \partial_\tau{\Gamma}^u(\alpha,\tau) \big]d\tau}_{A} + \underbrace{\big[ \Gamma^u(\beta, \tau_\delta) - \Gamma^u(\alpha, \tau_\delta)\big]}_{B} + \underbrace{\int_{T(\alpha)}^{T(\beta)} \partial_\tau{\Gamma}^u(\beta,\tau)d\tau}_{C} \label{delta_gamma_b}
        \end{align}
        \vspace{-2em}
        \begin{figure}[H]
        \centering
\includegraphics[width=0.45\linewidth]{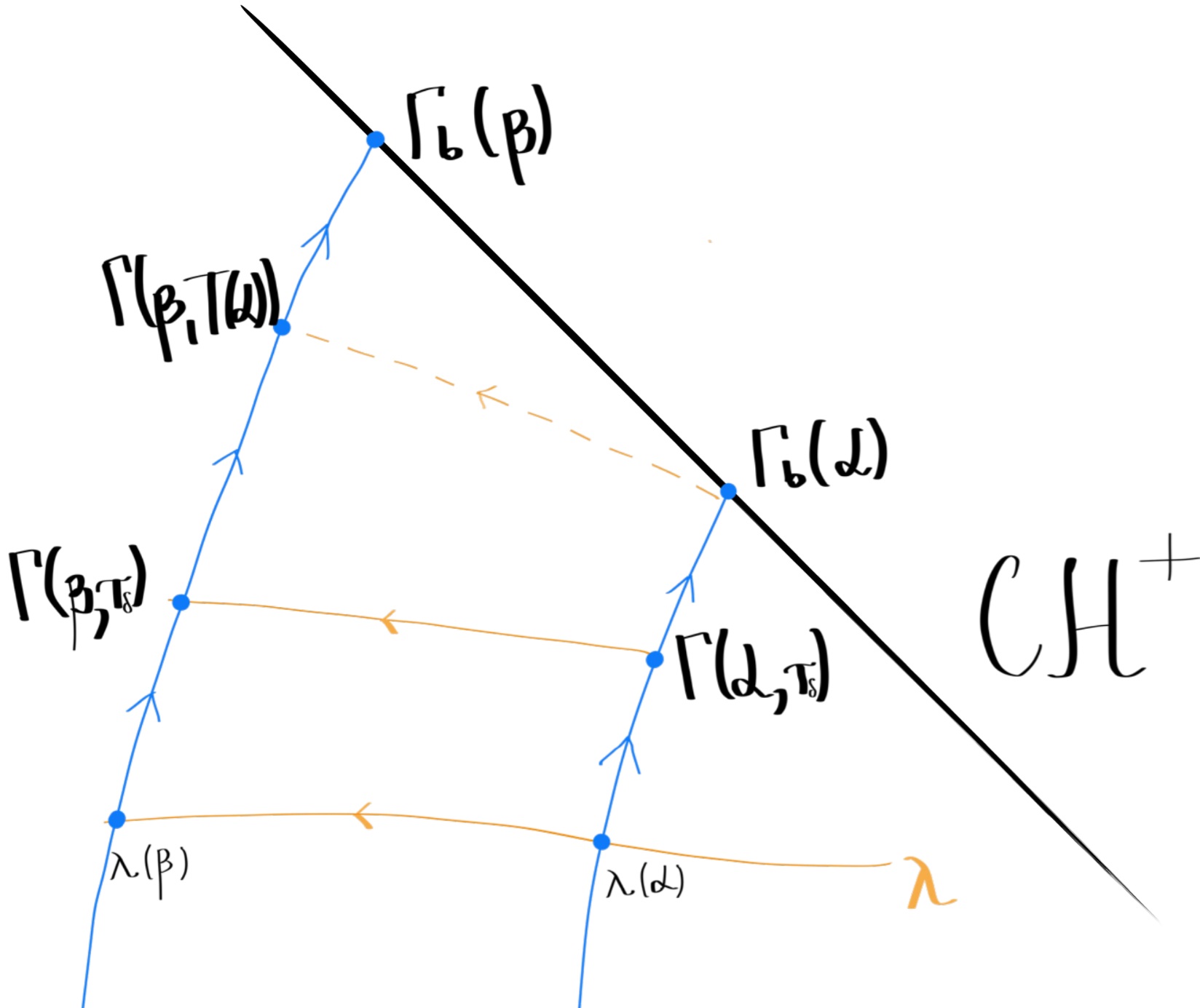}
        \caption{Illustration of the quantities involved in equation (\ref{delta_gamma_b}). We can identify the terms from the equation in the diagram: $A=[\Gamma^u(\beta,T(\alpha)) -\Gamma^u(\beta,\tau_\delta)] + [\Gamma^u(\alpha,T(\alpha)) - \Gamma^u(\alpha,\tau_\delta)]$ is the difference of the two 'parallel vertical segments', $B=[\Gamma^u(\alpha, \tau_\delta) - \Gamma^u(\beta, \tau_\delta)]$ is the 'horizontal segment' and $C = \Gamma^u(\beta,T(\beta)) - \Gamma^u(\beta, T(\alpha))$ is the 'remaining vertical segment'.}
    \end{figure}
    
       It should be clear that $A$ and $B$ are not the constants in (\ref{A}-\ref{B}) but rather terms that we estimate in the current Proposition. By Corollary \ref{unidorm_bounds_on_geod_velocities},
       \begin{align}
            A \leq \delta \big(e^{2C_g}\sqrt{C_g N_0/n_0}-e^{-2C_g} \sqrt{n_0/C_g^3N_0}\big)\to 0 
       \end{align}
       as $\delta \searrow 0$.
       Next, we look at 
        \begin{equation}
            B={\Gamma}^u(\alpha, \tau_\delta) - {\Gamma}^u(\beta,\tau_\delta) = \int^\alpha_\beta -J^u(\alpha', \tau_\delta) d\alpha' \geq n_{null}(\alpha-\beta)>0\label{term_B}
        \end{equation}
        by Corollary \ref{orientation_preserving_corollary}. It follows that $ \lim_{\delta\to 0^+} B \geq n_{null}(\alpha - \beta) > 0$. Finally, as $T(\alpha)< T(\beta)$ and $\partial_\tau{\Gamma}^u>0$ on all of $D^0$, it follows that $C>0$. We conclude that 
        \begin{equation}
        \Delta\Gamma_\mathbf{b}^u[\alpha;\beta] = \lim_{\delta\to 0^+} [A+B+C] >   n_{null} (\beta - \alpha) > 0
        \end{equation}
    \end{proof}
\end{prop}

 \subsection{The Dust Energy Density}
Consider a spherically symmetric dust field, infalling towards $\mathcal{CH}^+$. The ultimate aim of this section is to prove that the fluid energy density does not blow up before, or at, the singularity, given reasonable initial data for the energy density $\rho(u,v)$ and four-velocity $U(u,v)$. Reasonable in our case means that the initial data will be supplied on an $(n_0, N_0, C_K)$--admissible curve $\lambda$. Admissible curves are contained in a compact subset of $\mathcal{Q}_\mathcal{R}$ close enough to the Cauchy horizon, so that a variation of future-directed timelike geodesics based on an admissible curve experiences no shell-crossing. The ultimate aim is achieved in the following result:
 \begin{thm}\label{dust_energy_density_theorem}  Let  $\overline{\Gamma}: D \to \overline{\mathcal{Q}_\mathcal{R}}$ be the geodesic variation based on an  $(n_0, N_0, C_K)$--admissible curve $\lambda:(-a, a)\to \mathcal{Q}_\mathcal{R}$ with future-directed unit normal $\nu$ and with $\inf_\alpha v[\lambda(\alpha)]=v_{\text{min}}$. Let $\rho=\rho(u,v)$ be the energy density of a spherically symmetric perfect pressureless fluid in $\mathcal{R}$ such that the pushforward of the 4-velocity field $U=U(u,v)$ into $\mathcal{Q_R}$ coincides with the tangent field to the geodesic variation $\Gamma$ as defined in Theorem \ref{master_theorem}\footnote{This includes the assumption, as posed in Theorem \ref{master_theorem}, that $\rho$ is bounded and positive on $\lambda$.}. Then $\rho$ remains bounded as $v\to 0^-$. More precisely, $\exists C_\rho>0$, depending on $n_0, N_0, C_K, v_{\text{min}}$ and $C_g$ such that $\rho\leq C_\rho \norm{\rho}_{C^0(\lambda)}$ on $\supp(U)\equiv\overline{\Gamma}(D)\times S^2$.

\begin{proof} With $U$ being spherically symmetric, $U^\theta= U^\varphi=0\ $ and $U^u = \partial_\tau \Gamma^u, U^v = \partial_\tau \Gamma^v\  \forall p \in \Gamma(D^0)$, so we can think of the image of integral curves of $U$ in $\mathcal{R}$ as identical to the image of geodesics of the variation $\Gamma$ in $\mathcal{Q_R}$. Each future-directed radial timelike geodesic approaches a unique point on $\mathcal{CH}^+$ at finite affine parameter time, so the extension of the geodesic variation (i.e. the flow of $U$) to $\overline{\mathcal{Q}_\mathcal{R}}$ is well-defined. Recall, as in Theorem \ref{master_theorem}, we denote such a variation of geodesics by 
\begin{equation}
    \overline{\Gamma}: D \to \overline{\mathcal{Q_R}} \ \ \text{      for      } \ \ 
    D = \{ (\alpha,\tau):0 \leq \tau \leq T(\alpha), -a<\alpha<a\} \subset \mathbb{R}^2
\end{equation}
where at affine parameter time $0$ one has $\overline{\Gamma}(\alpha, 0) = \lambda(\alpha)$ and at affine parameter time  $T(\alpha)$ one has $\overline{\Gamma}(\alpha, T(\alpha)) \in \mathcal{CH}^+$ for every $\alpha\in(-a,a)$. \\[5pt]
With this choice of initial data, $\overline{\Gamma}|_{D^0}=: \Gamma$ satisfies the hypotheses of Theorem \ref{interior_regularity_of_gamma} and so is a diffeomorphism onto its image (which coincides with $\supp (U)/S^2$). Moreover $\overline{\Gamma}|_{\{(\alpha, s):s=T(\alpha)\}}=:\Gamma_\mathbf{b}$ is injective in $\alpha$ by Proposition \ref{boundary_regularity_of_gamma}.
By the former property, we can use $(\alpha,\tau, \theta, \varphi)$ as a coordinate chart on $\supp(U)\subset \mathcal{R}$. By the transport equation (\ref{dust_transport_equation}) for $\rho$ along the integral curves of $U$,
\begin{align}
\ln\bigg[\frac{\rho\circ\Gamma(\alpha, \tau)}{\rho\circ\Gamma(\alpha,0)}\bigg] = - \int_0^\tau \nabla \cdot U (u,v) (\Gamma(\alpha,s))ds \label{consequence_of_transport_eqn}
\end{align}
As $\rho>0$ on Image$(\lambda)\times S^2$, the LHS is well-defined. It follows that for any $0<\tau\leq T(\alpha)$ and all $\alpha \in (-a, a)$, $\rho$ is finite at the point $\Gamma(\alpha, \tau)$ iff the RHS in (\ref{consequence_of_transport_eqn}) is finite, that is $\norm{\nabla\cdot U}_{L^1(\Gamma(\alpha,\cdot))}<\infty$. To deal with this expression, we fix a local parallelly transorted orthonormal basis (ONB) $(e_A)$:
\begin{align}
    \begin{pmatrix}
        e_0 \\
        e_1 \\
        e_2 \\
        e_3 
    \end{pmatrix} = \begin{pmatrix}
        U^u \partial_u + U^v \partial_v \\
        U^u \partial_u - U^v \partial_v \\
        r^{-1} \partial_\theta \\
        r^{-1}(\sin \theta)^{-1} \partial_\varphi
    \end{pmatrix}
\end{align}
Since the ONB is parallelly transported, $\nabla_U e_A \equiv \nabla_{e_0} e_A = 0$ for $A=0,1,2,3$. Furthermore, since $U$ is normalised, 
\begin{align}
    g(U,U) = g_{00} = -1; &\  g(e_1, e_1) = g_{11}= -2g_{uv}U^u U^v = 1
\end{align}
and it is easy to verify that
$ g(e_2,e_2) = g(e_3,e_3) = 1$.
Expressed in the ONB, the divergence of $U$ is
\begin{align}
    &\nabla\cdot U  = \nabla_A U^A = \sum_{A=1}^3 g(e_A, \nabla_{e_A} U)
\end{align}
where we used that $\nabla_{e_0} U = 0$. Let's examine the terms $A=1,2,3$ in the above sum. At each point on the  geodesic $\tau\mapsto (\Gamma(\alpha, \tau), \theta=\text{const}, \varphi=\text{const})$, $e_1$ is radial and normal to $U$. Extend the $1+1$-dimensional variational field $J$ and its normal derivative $K$ to vector fields on Image$(\Gamma)\times S^2\subset\mathcal{R}$ by setting $J^2, J^3, K^2,K^3=0$ everywhere, and denote them again by $J$ and $K$, respectively. Thus $e_1$ is related to the $3+1$-dimensional variational field $J$ of the geodesic variation by
\begin{equation}
    J = |J^1(\alpha,\tau)| e_1(\tau)
\end{equation}
at each $0\leq \tau < T(\alpha)$. Hence
\begin{align}
&\big\langle e_1, \nabla_{e_1} U \big\rangle = \bigg\langle \frac{J(\alpha,\tau)}{|J^1(\alpha,\tau)|} , \frac{\nabla_{J}U(u,v)}{|J^1(\alpha,\tau)|}\bigg\rangle = \frac{1}{|J^1(\alpha,\tau)|^2} \big\langle J(\alpha,\tau), \nabla_U J(\alpha,\tau) \big\rangle =\nonumber \\[5pt]
&=\frac{1}{|J^1(\alpha,\tau)|^2}\big\langle J(\alpha,\tau), K(\alpha,\tau) \big\rangle \leq \frac{\norm{e^{2\omega}}_{C(\Gamma(\alpha,\cdot))} N_\perp^2}{n_\perp^2} \leq \frac{C_g N_\perp^2}{n_\perp^2} \label{ONB_approach1}
\end{align}
where we used $[U,J]=0 = \nabla_U J - \nabla_J U$ and that $K=\nabla_U J$ is the derivative of the Jacobi field along the geodesic $\tau\mapsto(\Gamma(\alpha,\tau),\theta,\varphi)$. In the final estimate we used the uniform upper and lower bounds on the Jacobi field on all of Image$(\overline{\Gamma})$\footnote{Recall that $N_\perp = \sup_{\alpha} N_{\perp, \alpha}$ and $n_\perp = \inf_{\alpha} n_{\perp, \alpha}$ obtained by taking supremum/infimum over $\alpha$ of the upper/lower bounds of the Jacobi field provided by Theorem \ref{bounds_for_jacobi_field_thm}.}, the uniform upper bound on its derivative, and the uniform upper bound on the lapse function in $\mathcal{R}$ from Assumptions \ref{assumptions_on_L1_norms}. Recall that $n_\perp, N_\perp$ in Theorem \ref{bounds_for_jacobi_field_thm} depend on $C_g,v_{\text{min}}, N_0, n_0, C_K$.  We conclude that $\langle e_1, \nabla_{e_1} U \rangle$ is in $L^1(\Gamma(\alpha, \cdot))$. \\[6pt]
Next, we look at 
\begin{align}
    &\big\langle e_2, \nabla_{e_2} U \big\rangle = \frac{1}{r^2}\big\langle \partial_\theta, \nabla_{\theta} U \big\rangle = \frac{g_{\theta\theta}}{r^2} (\Gamma^\theta_{u\theta} U^u + \Gamma^\theta_{v\theta}U^v) = \frac{1}{r}(r_{,u} U^u + r_{,v} U^v)\\[3pt]
    &\big\langle e_3, \nabla_{e_3} U \big\rangle = \frac{1}{r^2\sin^2{\theta}}\big\langle \partial_\varphi, \nabla_{\varphi} U \big\rangle = \frac{g_{\varphi\varphi}}{r^2\sin^2 \theta} (\Gamma^\varphi_{u\varphi} U^u + \Gamma^\varphi_{v\varphi}U^v) = \frac{1}{r}(r_{,u} U^u + r_{,v} U^v)
\end{align}
  We can bound the integral of $(\ln r)_{,\mu}$ along the geodesic $\Gamma(\alpha, \cdot)$ as follows:
 \begin{align}
     \norm{(\ln r)_{,u}}_{L^1(\Gamma(\alpha,\cdot))} &\leq \frac{\norm{r_{,u}}_{L^1(\Gamma(\alpha,\cdot))}}{\inf_{\Gamma(\alpha,\cdot)}r(u,v)} = \frac{1}{r(\Gamma^u_{\mathbf{b}}(\alpha), 0)} \int_{u(\lambda(\alpha))}^{u(\Gamma_{\mathbf{b}}(\alpha))} \frac{|r_{,u}(u,v(u))|}{\partial_\tau \Gamma^u(\alpha, \tau(u))}du \leq \nonumber\\
      &\leq \frac{1}{r(\Gamma^u_{\mathbf{b}}(\alpha), 0)} \times \frac{1}{\inf_{\tau\in(0, T(\alpha))}{\partial_\tau \Gamma^u(\alpha,\tau)}} \times \sup_{v\in(v(\lambda(\alpha)), 0)}\norm{r_{,u}(\cdot,v)}_{L^1[u(\lambda(\alpha)), u(\Gamma_{\mathbf{b}}(\alpha))]} \leq \nonumber\\
     &\leq \frac{1}{r(\Gamma^u_{\mathbf{b}}(\alpha), 0)} { C_g e^{2C_g} \nu^v(\alpha)} \times C_g \leq  \frac{C_g^{5/2} e^{2C_g}}{r(\Gamma^u_{\mathbf{b}}(\alpha), 0)} \sqrt{\frac{N_0}{n_0}} \label{du_lnr_final_estimate}
 \end{align}
 The estimates applied in (\ref{du_lnr_final_estimate}) are in the following order:
 \begin{enumerate}
     \item Assumptions \ref{instability_assumptions} are used: since for each $u_0$ there is a $v_0$ such that $r_{,u}<0$ for all $v>v_0$, and since $r_{,v} \to - \infty$ as $v\to 0$, it follows that the infimum of the area radius on the image of the geodesic $\Gamma(\alpha,0)$ is achieved at the point on the Cauchy horizon $(u(\Gamma_{\mathbf{b}}(\alpha), 0), 0)$, where the geodesic terminates.
     \item Proposition \ref{lower_bds_for_velocities} is applied to obtain a bound on the $u$-component of the geodesic velocity in terms of $C_g$ and the $u$-component of the initial velocity $\partial_\tau \Gamma^u(\alpha, 0) = \nu^u(\alpha)$, where $\nu$ denotes the future-directed unit normal to $\lambda$. 
     \item From Assumptions \ref{assumptions_on_L1_norms}, equation (\ref{geom_bound_as2}) is applied to bound the $L^1$ norm of $r_{,u}$ along an ingoing null interval.
     \item By expressing $\nu^u(\alpha) = e^{-\omega}|_{\lambda(\alpha)} \sqrt{(\lambda')^u(\alpha)/(\lambda')^v(\alpha)}$ and applying equation (\ref{geom_bound_as1}) from Assumptions \ref{assumptions_on_L1_norms} as well as admissibility condition 3, we bound $\nu^u(\alpha)$ and hence obtain the final expression.
 \end{enumerate}
  We estimate $\norm{(\ln r)_{,v}}_{L^1(\Gamma(\alpha, \cdot))}$ by following identical steps and obtain the same bound on the RHS. Nevertheless, the main steps are shown below.
 \begin{align}
       \norm{(\ln r)_{,v}}_{L^1(\Gamma(\alpha,\cdot))} &\leq\frac{1}{r(\Gamma^u_{\mathbf{b}}(\alpha), 0)} \int_{v(\lambda(\alpha))}^0 \frac{|r_{,v}(u(v),v)|}{\partial_\tau \Gamma^v(\alpha, \tau(v))}dv \leq \nonumber\\
      &\leq \frac{\sup_{u\in(u(\lambda(\alpha)), u(\Gamma_{\mathbf{b}}(\alpha)))}\norm{r_{,v}(u,\cdot)}_{L^1[v(\lambda(\alpha)), 0]}}{r(\Gamma^u_{\mathbf{b}}(\alpha), 0) \inf_{\tau\in(0, T(\alpha))} {\partial_\tau \Gamma^v(\alpha,\tau)}}  \leq  \frac{C_g^{5/2} e^{2C_g}}{r(\Gamma^u_{\mathbf{b}}(\alpha), 0)} \sqrt{\frac{N_0}{n_0}} \label{dv_r_final_estimate}
 \end{align}
Next, we would like to bound $\norm{U^\mu}_{C(\Gamma(\alpha, \cdot))}$ for $\mu\in\{u,v\}$. This is easier and the estimate consists of two steps. First, applying Theorem \ref{estimates_for_geodesic_velocities} to bound $U^\mu(u(\Gamma(\alpha,\tau)),v(\Gamma(\alpha,\tau))) = \partial_\tau \Gamma^\mu(\alpha,\tau)$ in terms of $U^\mu(u(\Gamma(\alpha,0)),v(\Gamma(\alpha,0)) ) = \nu^\mu(\alpha)$ and the universal constant $C_g$. Second, expressing $\nu^\mu$ in terms of the components of $\lambda'$ and applying admissibility condition 3 in a way that is identical to point 4 in the estimate of $\norm{(\ln r)_{,u}}_{L^1(\Gamma(\alpha, \cdot))}$. In detail,
\begin{align}
    \norm{U^\mu}_{C(\Gamma(\alpha, \cdot))} = \sup_{(u,v)\in \text{Image}(\Gamma(\alpha,\cdot))} U^\mu(u,v) =\sup_{\tau\in(0,T(\alpha))}{\partial_\tau \Gamma^\mu(\alpha,\cdot)} \leq \nu^u(\alpha) e^{2C_g} \leq C_g^{1/2} e^{2C_g} \sqrt{\frac{N_0}{n_0}} \label{U_final_estimate}
\end{align}
 Combining (\ref{du_lnr_final_estimate} - \ref{dv_r_final_estimate}) and (\ref{U_final_estimate}), we conclude that $(\ln r)_{,u} U^u + (\ln r)_{,v} U^v$ is in $L^1(\Gamma(\alpha,\cdot))$. Therefore, $\langle e_2, \nabla_{e_2} U \rangle + \langle e_3, \nabla_{e_3} U \rangle \in L^1(\Gamma(\alpha, \cdot))$. Choosing $\tau\in(0, T(\alpha)]$, we conclude that $\nabla\cdot U$ is integrable along the integral curves of $U$ up to the Cauchy horizon, that is,
 \begin{align}
     &\Bigg|\ln\bigg[\frac{\rho\circ\Gamma(\alpha, \tau)}{\rho\circ\Gamma(\alpha,0)}\bigg]\Bigg| =\Bigg|\int_0^\tau \nabla \cdot U (u,v) (\Gamma(\alpha,s))ds \Bigg| \leq \sum_{A=1}^3 \int_{0}^\tau |\langle e_A , \nabla_{e_A} U  \rangle | ds \leq \nonumber\\
     &\leq C_g \frac{N_\perp^2}{n_\perp^2} + 4C_g^{1/2} e^{2C_g} \sqrt{\frac{N_0}{n_0}} \times \frac{1}{r(\Gamma_{\mathbf{b}}(\alpha, 0))} C_g^{5/2} e^{2C_g} \sqrt{\frac{N_0}{n_0}}\leq \nonumber\\
      &\leq \frac{N_0^2 C_g^3 (1+C_K^2) \exp{\Bigg[ C_g^{3/2} e^{2C_g}|v_{\text{min}}| \sqrt{\frac{N_0}{n_0}}  +  2C_g^{7/2} e^{2C_g} (1+|v_{\text{min}}|)\sqrt{\frac{N_0}{n_0}}\Bigg]} }{n_0^2 \bigg[1-{A|v_{\text{min}}|e^{B|v_\text{min}|}} \bigg]^2} + 4C_g^4 e^{4C_g} \frac{N_0}{n_0} \label{long_dust_estimate}
 \end{align}
 The first inequality is simply the integral triangle inequality. In the second inequality, we substitute the bounds obtained for $\norm{\langle e_1, \nabla_{e_1} U \rangle}_{L^1(\Gamma(\alpha,\cdot))}$ and for $2\norm{(\ln r)_{,u} U^u + (\ln r)_{,v} U^v}_{L^1(\Gamma(\alpha, \cdot))} = \norm{\langle e_2, \nabla_{e_2} U\rangle}_{L^1(\Gamma(\alpha, \cdot))} + \norm{\langle e_3, \nabla_{e_3} U\rangle}_{L^1(\Gamma(\alpha, \cdot))}$. In the third inequality we substituted the expressions (\ref{Nperp_expression}) and (\ref{lower_j_bound_1}) for $N_\perp$ and $n_\perp$ and used the uniform upper bound on $1/r$ provided in Assumptions \ref{assumptions_on_L1_norms}, equation (\ref{geom_bound_as1}).  In the interest of readability, we have not substituted in the expressions (\ref{A}-\ref{B}) for $A$ and $B$, but we recall that they depend on $C_g, N_0, n_0, C_K$ and $|v_{\text{min}}|$. Hence the RHS of (\ref{long_dust_estimate}) is a constant $C$, depending on $C_g, N_0, n_0, C_K$ and $|v_{\text{min}}|$. In conclusion,
 \begin{align}
     &|\ln (\rho(\Gamma(\alpha,\tau)) - \ln (\rho(\Gamma(\alpha,\tau)))| \leq C(C_g, n_0, N_0, C_K,|v_{\text{min}}|) \implies \nonumber\\
     &\rho(\Gamma(\alpha,\tau)) \leq \exp(C(C_g,n_0,N_0,C_K,|v_{\text{min}}|))\rho(\Gamma(\alpha,0)) \leq C_\rho \norm{\rho}_{C^0(\lambda)} \label{final_rho_estimate}
 \end{align}
 In the last inequality we defined the constant $C_\rho:= e^C$. Since the point $\Gamma(\alpha, \tau)\in \Gamma(D^0)$ was chosen arbitrarily, the dust energy density remains bounded on $\overline{\Gamma}(D)$.
\end{proof}
\begin{rem} From (\ref{final_rho_estimate}) one can also obtain a lower bound on $\rho(\Gamma(\alpha, \tau))$ in terms of $\rho(\Gamma(\alpha, 0))$ as follows: the first line of (\ref{final_rho_estimate}) is equivalent to 
\begin{align}
    &\pm \ln(\rho(\Gamma(\alpha,\tau))) \leq C \pm \ln(\rho (\Gamma(\alpha, 0))) \implies \nonumber\\
    &\exp(\pm \ln \rho(\Gamma(\alpha,\tau))) \leq C_\rho \exp(\pm \ln \rho(\Gamma(\alpha,0)))
\end{align}
Taking the $+$ sign, we obtain the statement of Theorem \ref{dust_energy_density_theorem} i.e. the second line of (\ref{final_rho_estimate}). Taking the $-$ sign, we obtain
\begin{align}
    \rho(\Gamma(\alpha,\tau)) \geq \frac{1}{C_\rho} \rho(\Gamma(\alpha,0))
\end{align}    
\end{rem}
\end{thm}

\subsection{Summary of Results: Proof of Theorem \ref{master_theorem}}
Let $\lambda, \Gamma, U, \rho$ satisfy the hypotheses of Theorem \ref{master_theorem}. By Corollary \ref{unidorm_bounds_on_geod_velocities}, $U^\mu$ are uniformly bounded and bounded away from zero on Image$(\overline{\Gamma})\times S^2$ as they are the components of timelike geodesics based on an admissible curve $\lambda$. This establishes claim 1 of Theorem \ref{master_theorem}. By Proposition \ref{T_is_monotonic},  the proper time function $T(\cdot)$ is monotonic and by Proposition \ref{contn_dep_of_T_on_alpha}, it is continuous in the parameter $\alpha$ which parameterises the initial curve $\lambda$ and labels the geodesics. By Theorem \ref{interior_regularity_of_gamma}, the geodesic variation map $\Gamma$ is a diffeomorphism onto its image $\Gamma(D^0)\subset \mathcal{Q}_\mathcal{R}$, and by Proposition \ref{boundary_regularity_of_gamma}, its extension  $\overline{\Gamma}: D \mapsto\overline{\Gamma}(D)\subset \overline{\mathcal{Q_R}}$ is a homeomorphism onto its image. Thus claims 2, 3, 4 and 5 of Theorem \ref{master_theorem} are established. By Theorem \ref{dust_energy_density_theorem}, the energy density $\rho= \rho(u,v)$ remains bounded throughout the region where the dust is supported, including at the weak null singularity on $\mathcal{CH}^+$, which proves the last claim of Theorem \ref{master_theorem}. 
\section{Spherical Stiff Fluid and The Wave Equation near the WNS}\label{wave_section}
    This Section is dedicated to proving the second main result of this paper, namely Theorem \ref{masterTheorem_2}. 
     The stiff fluid analysis in Section \ref{preliminaries} and in particular equation (\ref{lner_wave_eqn_psi}) tells us that, provided the gradient of the scalar field remains timelike and future-directed, the characteristic initial value problem for $(\rho, U)$ in Theorem \ref{masterTheorem_2} is equivalent to a characteristic initial value problem for the homogeneous linear wave equation for a smooth scalar field $\psi$ such that $\nabla_a \psi = V_a = \sqrt{\rho} U_a$ (so $\nabla\cdot V = 0 \iff \Box_g \psi = 0$). In view of the future-directedness of $V$, $-2e^{-2\omega}\partial_v\psi =V^u>0$ and $-2e^{-2\omega}\partial_u\psi =V^v>0$. We therefore insist that the derivatives $\partial_u \psi$ and $\partial_v \psi$ must have a negative sign (on the initial characteristic hypersurface). We remark that in a generic spacetime the duality between the relativistic stiff fluid problem (an inherently nonlinear problem) and the linear wave equation can break down even inside the domain of dependence of the hypersurface where the initial data is supplied. Namely, for solutions of the linear wave equation with either compactly supported data, or tails which vanish at spatial infinity, $\nabla^a \psi$ can become spacelike at some point, violating the physically motivated condition that the fluid velocity is a  future-directed timelike vector. An example where breakdown of duality happens for a compactly supported stiff fluid in Minkowski background is provided in \cite{Gavassino}. We will show that in the spherically symmetric weak null singularity background $(\mathcal{R}, g)$, imposing characteristic initial data so that $\partial_u \psi, \partial_v\psi <0$ on a bifurcate null hypersurface close to the Cauchy horizon leads to a solution of the wave equation whose gradient remains timelike and future-directed in the domain of dependence of the data. This crucially relies on restricting to a region where the gradient of $r$ is future-directed and timelike i.e. $\partial_u r, \partial_v r<0$ in double null coordinates. The existence of such a region is a key property of the singular nature of the spacetime, a consequence of Assumptions \ref{instability_assumptions}. \\[5pt]
     To prove the statement of Theorem \ref{masterTheorem_2} we must show that any solution of the characteristic IVP for the wave equation indeed blows up in $C^1$ at the singularity. This implies blow-up of the energy density of the stiff fluid. On the other hand, showing that the ratio $\partial_u\psi/\partial_v\psi$ goes to zero at the singularity will be needed to establish the statement regarding the behaviour of the fluid velocity components.\\[5pt]
     We begin by formulating the IVP in terms of the wave equation and proceed by proving the relevant estimates in the language of the wave equation, leading to the conclusion of Theorem \ref{masterTheorem_2}.
\subsection{Formulating the characteristic IVP}
A crucial part of the monotonicity argument in this section is choosing a domain where $\nabla^ar$ is a future-directed timelike vector i.e. where $r_{,u}$ and $r_{,v}$ are negative. Let $(u_0,u_1) \subset(-\infty, u_s)$ and $v_0\in(v_\kappa, 0)$ be such that  $\mathcal{C}:= \mathcal{C}_+ \cup \overline{\mathcal{C}_-}$, where
\begin{align}
    &\mathcal{C}_+ = \{u_0\}\times [v_0, 0)\times S^2; \\
    &\overline{\mathcal{C}_-} = [u_0, u_1] \times \{v_0\}\times S^2
\end{align}
is a $u_0$--admissible bifurcate null hypersurface in $\mathcal{R}$ in the sense of Definition \ref{define_admissible_hypersurface}.\\[5pt]

On $D^+(\mathcal{C})=[u_0,u_1]\times[v_0,0)\times S^2$, let $\rho$ and $U$ be defined as in Theorem \ref{masterTheorem_2}. and let $\psi$ be a scalar function in $D^+(\mathcal{C})$ defined by $\nabla^a\psi =\sqrt{\rho} U^a$. Such a scalar function exists by  Section \ref{stiff_fluid_to_wave_equation}. Again by Section \ref{stiff_fluid_to_wave_equation}, $\psi$ satisfies the spherically symmetric homogeneous linear wave equation in $D^+(\mathcal{C})$ with characteristic initial data $\mathring{\psi}$ supplied on $\mathcal{C}$
\begin{align}
    &\Box_g{\psi}(u,v) \ \ \in D^+(\mathcal{C})  \label{box_wave_eqn}\\ 
    &\psi|_{\mathcal{C}} = \mathring{\psi} \label{box_wave_char_data}
\end{align}
where $\partial_v\mathring{\psi}|_{\mathcal{C}_+} =-\frac{1}{2}e^{2\omega}\sqrt{\mathring{\rho}} \mathring{U}^u|_{\mathcal{C}_+}$ and $\partial_u\mathring{\psi}|_{\overline{\mathcal{C}_-}} =-\frac{1}{2}e^{2\omega}\sqrt{\mathring{\rho}} \mathring{U}^v|_{\overline{\mathcal{C}_-}}$ with $(\mathring{\rho}, \mathring{U})$ defined as in Theorem \ref{masterTheorem_2}.
Since $\mathring{\rho}$ and $\mathring{U}^\mu$ are assumed to be smooth functions on the restrictions $\mathcal{C}_+$ and $\overline{\mathcal{C}_-}$ of $\mathcal{C}$, it follows that $\mathring{\psi}|_{\mathcal{C}_+}\in C^\infty(\mathcal{C}_+)$ and $\mathring{\psi}|_{\overline{\mathcal{C}_-}}\in C^\infty(\overline{\mathcal{C}_-})$. \\[5pt]
Section 2 of \cite{Rendall} establishes a  method to reduce the characteristic IVP for a generic \textit{quasilinear} wave equation with smooth coefficients and smooth characteristic initial data to a Cauchy problem for an equation of the same form and with smooth Cauchy data. Furthermore, the future domain of dependence of the Cauchy data contains the domain of dependence of the characteristic data and the solution is identically zero on the complement. By Chapter 9 of \cite{Ringstrom}, the Cauchy problem for the quasilinear wave equation admits a unique smooth solution in a future neighbourhood of the initial spacelike hypersurface. On the other hand, Chapter 8 of \cite{Ringstrom} shows that local solutions to Cauchy problems for \textit{linear} wave equations can be uniquely extended to global solutions on the domain of dependence of the initial spacelike hypersurface. Combining the abovementioned results, we verify that there is a unique smooth solution of (\ref{box_wave_eqn}-\ref{box_wave_char_data}) in the future domain of dependence $D^+(\mathcal{C})$ of the hypersurface $\mathcal{C}$, illustrated in Figure \ref{wave_domain} below.

 \begin{figure}[H]
        \centering
\includegraphics[width=0.5\linewidth]{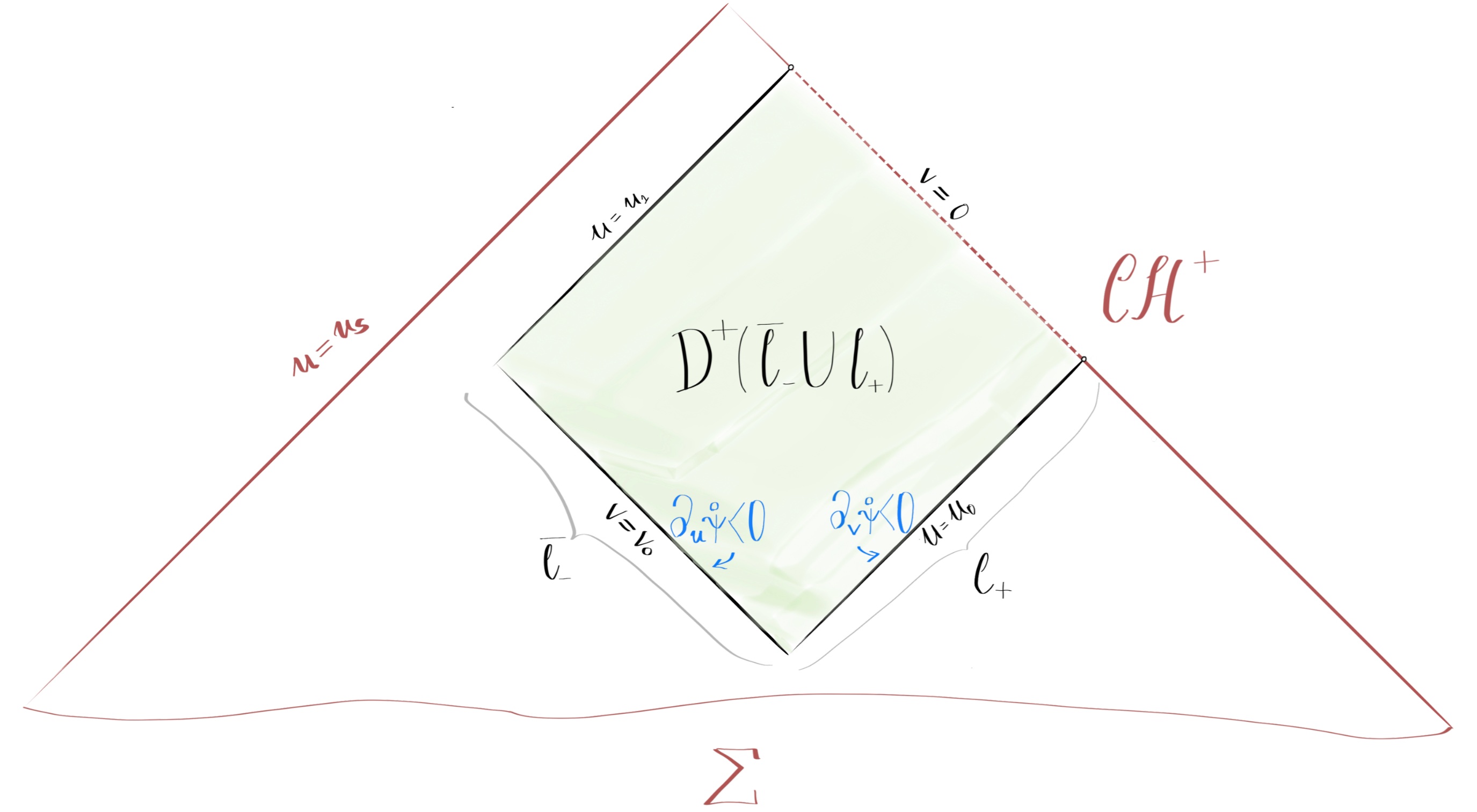}
        \caption{Illustration of the domain where we have a unique solution to the characteristic initial value problem for the spherically symmetric wave equation. Characteristic initial data $\mathring{\psi}$ is supplied on the bifurcate null hypersurface $\overline{\mathcal{C}_-} \cup \mathcal{C}_+$.}
        \label{wave_domain}
    \end{figure}
    
In double null coordinates on the spherically symmetric manifold $(\mathcal{R}, g)$ this becomes
\begin{align}
    &\partial_u \partial_v \psi (u,v) = - (\ln r)_{,u} \psi_{,v} - (\ln r)_{,v} \psi_{,u}   \in D^+(\mathcal{C})  \label{wave_eqn_main} \\
    &\psi|_{\mathcal{C}} = \mathring{\psi} \label{wave_char_data_main}
\end{align}
where (\ref{wave_eqn_main}) is also equivalent to both of the following:
\begin{align}
    &\partial_u(r \psi_{,v}) = - r_{,v} \psi_{,u} \label{wave_eqn_du}\\ 
    &\partial_v(r \psi_{,u}) = - r_{,u} \psi_{,v} \label{wave_eqn_dv}
\end{align}
in $\in D^+(\mathcal{C})$.

\subsection{Blow-up results on $\mathcal{CH}^+$}
Showing that the solution of the wave equation fails to maintain a bounded $C^1$ norm at the singularity leads to blow-up of the stiff fluid energy density. Theorem \ref{wave_blow_up_with_monotonicity} below establishes that the outgoing derivative of $\psi$ blows up as a consequence of Assumptions \ref{instability_assumptions}, which are related to $\mathcal{CH}^+$ being a weak null singularity. The proof of Theorem \ref{wave_blow_up_with_monotonicity} relies on a monotonicity argument established in Lemma \ref{monotonicity_is_propagated_on_C} and Proposition \ref{monotonocity_is_propagated_in_D}. Similar monotonicity arguments are used for the IVP for the Einstein--Maxwell--scalar field system from Dafermos' thesis, \cite{Dafermos03} and the more recent paper  by Sbierski and Fournodavlos studying the wave equation on Schwarzschild background, \cite{Sbierski_Fournodavlos}.
\subsubsection{Blow-up of $\partial_v \psi$}
\begin{thm}\label{wave_blow_up_with_monotonicity} Let $\psi$ be a smooth solution of the characteristic initial value problem for the wave equation (\ref{wave_eqn_main} - \ref{wave_char_data_main}). Assume that the null derivatives of the initial data satisfy $\partial_u\mathring{\psi}|_{\overline{\mathcal{C}_-}} <0$ and $\partial_v \mathring{\psi}|_{\mathcal{C}_+}<0$. 
Then 
\begin{align}
    |\partial_v \psi (u,v)| \geq   \frac{\inf_{\overline{\mathcal{C}_-}}|{\partial_u\mathring{\psi}}|}{r(u_0,v_0)}  \int_{u_0}^u (-r_{,v})(u',v)du'
\end{align}
\begin{cor}\label{integral_of_dvr_blows_up}
     Under the hypotheses of Theorem \ref{wave_blow_up_with_monotonicity}, $|\partial_v \psi| \to \infty$ as $v\to 0$ for every $u\in (u_0, u_1]$. 
    \begin{proof} 
    We apply Fatou's lemma to prove this. As $-\partial_v r$ is continuous in $[u_0, u_1]\times[v_0,0)$, for every fixed $v$ the function $(-\partial_v r)(\cdot, v)$ is (Lebesque) measurable on $[u_0, u_1]$. Let $\{v^{(n)}\}_{n=0}^\infty$ be an increasing sequence in $[v_0, 0)$ with $v^{(n)}\to 0$ as $n\to \infty$. By Assumptions \ref{instability_assumptions}, for every $u\in[u_0, u_1]$ we have $\partial_v r (u,v) \to -\infty$ as $v\to0$. Then for every $u\in[u_0, u_1]$ we have $-\partial_v r(u, v^{(n)})\to \infty$ as $n\to \infty$, so $\liminf_{n\to \infty}-\partial_v r(u, v^{(n)}) = \infty$. By Fatou's lemma,
    \begin{align}
         \liminf_{n}\int_{u_0}^{u}(-r_{,v})(u',v^{(n)}) du' \geq \int_{u_0}^{u}\liminf_n{(-r_{,v})(u',v^{(n)})du'} 
    \end{align}
    but the RHS is infinite unless $u=u_0$. Therefore, so is the LHS (unless $u=u_0$). As the sequence $v^{(n)}$ was chosen arbitrarily, we can conclude that for $u\in(u_0, u_1]$, 
    \begin{align}
        \lim_{v\to0} \int_{u_0}^{u}(-r_{,v})(u',v) du' = \infty = \int_{u_0}^{u}\lim_{v\to 0}(-r_{,v})(u',v)du'
    \end{align}
    By the statement of Theorem \ref{wave_blow_up_with_monotonicity}, for every $u\in (u_0, u_1]$, $|\partial_v \psi(u,v)|\to \infty$ as $v\to 0$. 
    \end{proof}
\end{cor}
    
\end{thm}
Before we prove Theorem \ref{wave_blow_up_with_monotonicity}, we need to show that the monotonicity assumption on $\psi$ in the ingoing coordinate is propagated inside the domain of dependence of $\mathcal{C}$. First, in Lemma \ref{monotonicity_is_propagated_on_C} we will prove that the monotonicity assumptions are propagated along the initial bifurcate null hypersurface $\mathcal{C} = \overline{\mathcal{C}_-} \cup \mathcal{C}_+$. Following that, in \ref{monotonocity_is_propagated_in_D}, we will use a bootstrap argument to prove that the subset of $D^+(\mathcal{C})$ on which $|\partial_u \psi|\geq c$ and $|\partial_v \psi| \geq c$ coincides with $D^+(\mathcal{C})$.
\begin{lem}\label{monotonicity_is_propagated_on_C}
 Under the hypotheses of Theorem \ref{wave_blow_up_with_monotonicity} we have $\partial_u{\psi}|_{\mathcal{C}_+}< \partial_u \mathring{\psi}(u_0, v_0)$ and $\partial_v{\psi}|_{\overline{\mathcal{C}}_-} < \partial_v\mathring{\psi}(u_0, v_0)$. 
\begin{proof}
  By the wave equation (\ref{wave_eqn_dv}), we have at any $(u,v) \in D^+(\mathcal{C})$,
  \begin{align}
      &\partial_u \psi (u,v) = \frac{1}{r(u,v)} \Bigg[ r(u,v_0) \partial_u \mathring{\psi}(u, v_0)  - \int_{v_0}^v \partial_u r(u, v') \partial_v \psi(u,v') dv' \Bigg] \label{dupsi_exprn}\\
      &\partial_v \psi (u,v) = \frac{1}{r(u,v)} \Bigg[ r(u_0,v) \partial_v \mathring{\psi}(u_0, v)  - \int_{u_0}^u \partial_v r(u', v) \partial_u \psi(u',v) du' \Bigg] \label{dvpsi_exprn}
  \end{align}
  so on $\mathcal{C}_+ = \{u_0\} \times [v_0, 0)$ we have in particular
  \begin{align}
      \partial_u {\psi} (u_0,v) = \frac{1}{r(u_0,v)} \Bigg[ r(u_0,v_0) \partial_u \mathring{\psi}(u_0, v_0)  - \int_{v_0}^v \partial_u r(u_0, v') \partial_v \psi(u_0,v') dv' \Bigg] < \partial_u \mathring{\psi}(u_0, v_0)
  \end{align}
  where we used that the integrand is strictly positive by the assumption that $\partial_u \mathring{\psi}(u_0,v_0)<0$, and the property that the area radius function is decreasing in $u$ and in $v$. The latter also ensures that $r(u_0,v_0)/r(u_0,v)\geq 1$. Using the negativity of $\partial_v \mathring{\psi}(u_0,v_0)$, we get an identical estimate for $\partial_v {\psi}$ on $\overline{\mathcal{C}}_{-} = [u_0,u_1]\times\{v_0\}$:
  \begin{align}
      \partial_v {\psi} (u,v_0) = \frac{1}{r(u,v_0)} \Bigg[ r(u_0,v_0) \partial_v \mathring{\psi}(u_0, v_0)  - \int_{u_0}^u \partial_v r(u', v_0) \partial_u \psi(u',v_0) du' \Bigg] < \partial_v \mathring{\psi}(u_0,v_0) 
  \end{align}
As before, the monotonicity of $r(u,v)$ ensures that $r(u_0,v_0)/r(u,v_0)\geq 1$
  
\end{proof}
\end{lem}

\begin{prop}\label{monotonocity_is_propagated_in_D} Let $v_1 \in (v_0/2, 0)$ be arbitrary and let 
\begin{align}
    c:=\min\{\inf_{\overline{\mathcal{C}_-}}|\partial_u\mathring{\psi}|, \inf_{\mathcal{C}_+\cap[v_0,v_1]}|\partial_v\mathring{\psi}|\} >0
\end{align}
Under the hypotheses of Theorem \ref{wave_blow_up_with_monotonicity}, $|\partial_u \psi(u,v)| \geq c$ and $|\partial_v \psi (u,v)| \geq c$ at every $(u,v) \in {[u_0, u_1]\times [v_0, v_1]}$\footnote{Note that since the null hypersurfaces $\overline{\mathcal{C}_-}$ and $ \mathcal{C}_+\cap [v_0,v_1]$ are compact, $\partial_u \mathring{\psi}$ attains a negative maximum on ${\overline{\mathcal{C}_-}}$ and $\partial_v \mathring{\psi}$ attains a negative maximum on ${\mathcal{C}_+}\cap[v_0,v_1]$.}.

\begin{figure}
    \centering
    \includegraphics[width=0.5\linewidth]{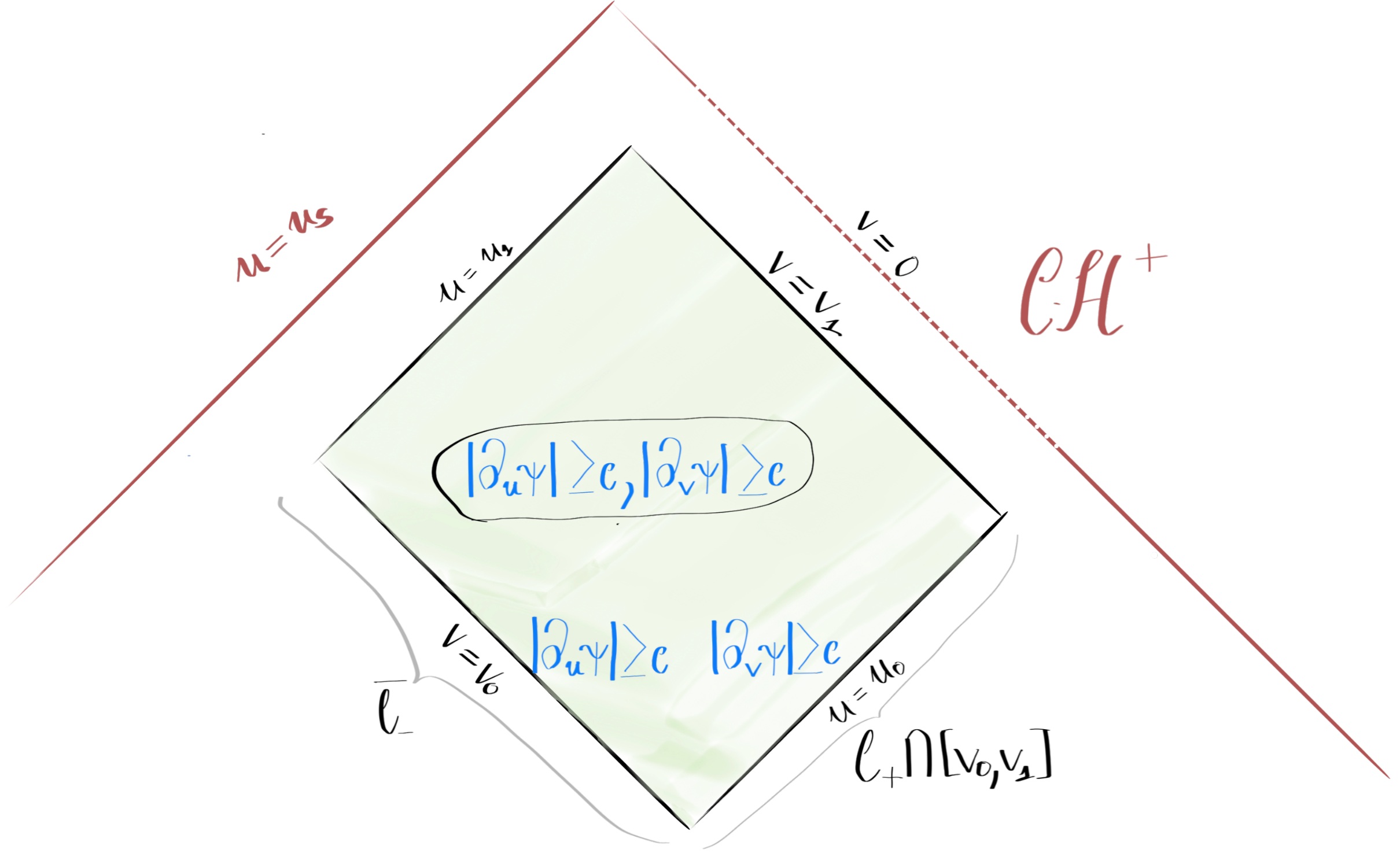}
    \caption{Illustration of Proposition \ref{monotonocity_is_propagated_in_D}. The minimal value of the null derivatives of $\psi$ in the domain of dependence of a bifurcate null hypersurface is no less than their minimal value on the bifurcate null hypersurface itself.}
    \end{figure}

\begin{proof}
  Given $t\in [0,1]$, define the following subsets of $D^+(\mathcal{C})$
  \begin{align}
      &\mathcal{A}_L(t) := 
     [u_0, u_1] \times [v_0, v_0 + t(v_1 - v_0)]\\
      &\mathcal{A}_R(t) := [u_0, u_0 + t(u_1 - u_0)] \times [v_0, v_1]
  \end{align}
  and let $\mathcal{A}(t):= \mathcal{A}_L(t) \cup \mathcal{A}_R(t)$. These sets are left and right shells around the initial ingoing vs outgoing null hypersurfaces $\overline{\mathcal{C}}_-$ and $\mathcal{C}_+\cap [v_0,v_1]$. These are illustrated in Figure \ref{bootstrap_shells}. Define the bootstrap set
  \begin{equation}
      \mathcal{B}:= \big\{t \in [0,1] : \partial_u \psi  \leq -c \text{ and } \partial_v \psi \leq -c \text{ in } \mathcal{A}(t)  \big\}
  \end{equation}
 
 \begin{figure}[H]
        \centering
\includegraphics[width=0.4\linewidth]{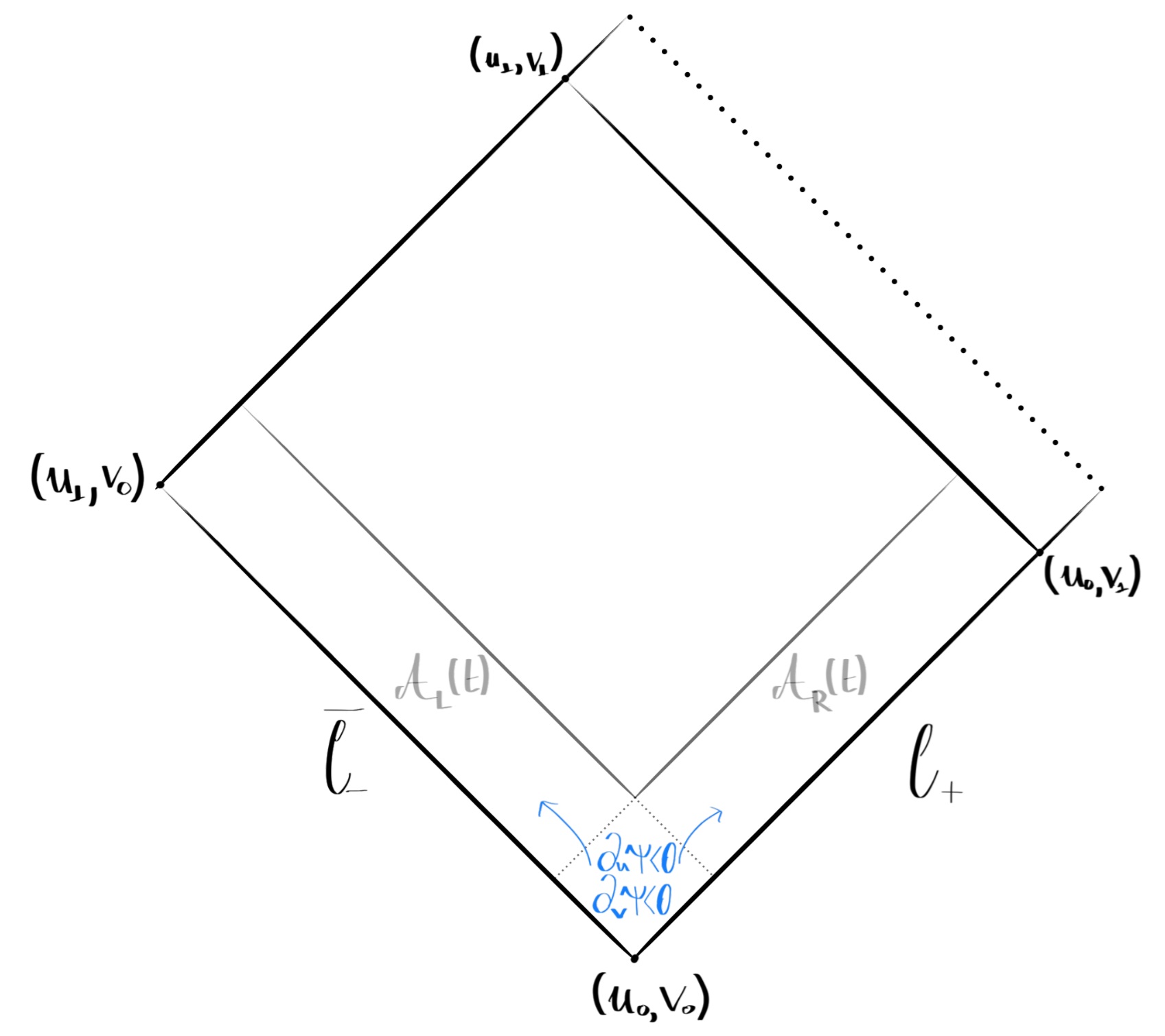}
        \caption{The left and right shells $\mathcal{A}_L(t)$ and $\mathcal{A}_R(t)$. We note that $t\in\mathcal{B}$ if and only if both the ingoing and outgoing derivatives of $\psi$ are negative on the bifurcate shell $\mathcal{A}_L(t) \cup \mathcal{A}_R(t)$.}
    \label{bootstrap_shells}
    \end{figure}
    
 \underline{Step 1: $\mathcal{B}$ is nonempty}\\[5pt]
  By Lemma \ref{monotonicity_is_propagated_on_C}, $\partial_u{\psi}|_\mathcal{C}\leq -c$ and $\partial_v{\psi}|_\mathcal{C} \leq -c$. Since $\mathcal{A}(0) = \{u_0\}\times [v_0, v_1] \cup [u_0, u_1]\times\{v_0\} = \mathcal{C}$, $0 \in \mathcal{B}$, so $\mathcal{B}\neq \emptyset$. \\[5pt]
  
  \underline{Step 2: $\mathcal{B} $ is closed in $[0,1]$.}\\[5pt]
   Let $\{t_i \}_{i=1}^\infty\subset\mathcal{B}$ be a sequence in $\mathcal{B}$ converging to $\overline{t}\in [0,1]$. We need to show that $\overline{t} \in \mathcal{B}$. Since $\mathcal{B}$ is obviously a connected interval containing $0$, we can assume wlog that $\{t_i\}$ is an increasing  sequence, $t_i \nearrow \overline{t}$. Indeed, by the continuity of the null derivatives $\partial_\mu \psi$ for $\mu \in \{u,v\}$ we have 
  \begin{align}
      &\lim_{i\to \infty}\partial_\mu \psi (u_0 + (u_1-u_0)t_i, v) = \partial_\mu \psi(u_0 + (u_1-u_0)\overline{t}, v) \text{ for $\overline{t}\leq 1$,  and  }\\
      & \lim_{i\to \infty} \partial_\mu \psi (u, v_0 + (v_1-v_0)t_i) = 
      \partial_\mu \psi (u, v_0 + (v_1-v_0)\overline{t}) \text{ for } \overline{t}\leq 1 
  \end{align}
  Hence
  \begin{align}
      &|\partial_\mu \psi (u_0 + (u_1-u_0)\overline{t}, v)| \geq \liminf_{i\to \infty} |\partial_\mu \psi (u_0 + (u_1-u_0)t_i, v)|  \geq |\partial_\mu \mathring{\psi} (u_0, v_0)|\geq c \text{  since   } \{t_i\} \subset \mathcal{B} \\
      &|\partial_\mu \psi (u, v_0 + (v_1-v_0)\overline{t})| \geq \liminf_{i\to \infty} |\partial_\mu \psi (u, v_0 + (v_1-v_0)t_i)| \geq |\partial_\mu \mathring{\psi} (u_0, v_0)|\geq c \  \text{  since   } \{t_i\} \subset \mathcal{B} \text{  for  } \overline{t}\leq 1
  \end{align}
  Since $\partial_\mu \psi|_{\mathcal{A}(t_i)} \leq -c$ for all $i$, this also holds on $\mathcal{A}(1)$. We conclude that $\mathcal{B}$ is closed in $[0,1]$.
  \\[5pt]

  \underline{Step 3: $\mathcal{B} $ is open in $[0,1]$.}\\[5pt]
  We note the nontrivial step of the bootstrap argument is to show that $\mathcal{B}$ is open. We will do this by showing that 
  \begin{itemize}
  \item There is a $t_0>0$ such that $t_0\in \mathcal{B}$. That is, there is at least a two-dimensional shell around $\overline{\mathcal{C}_-} \cup \mathcal{C_+} \cap [v_0, v_1]$ on which the bootstrap assumptions hold.
\item If $t_0 \in \mathcal{B}$ then there exists a $\delta>0$ such that $\partial_u \psi<-c$ and $\partial_v \psi<-c$ on $\mathcal{A}(t_0+\delta)$ (improving the bootstrap).  
  \end{itemize}
  Because $\partial_u \psi$ and $\partial_v \psi$ are bounded above by $-c$ on $\overline{\mathcal{C}}_-$ and continuous in $\mathcal{A}(1)$, there is a $\delta_L>0$ such that  
  \begin{align}
      \partial_u \psi \leq -c/2\  \text{  and  }\  \partial_v \psi \leq -c/2 \ \text{ in }B_{\delta_L}(u,v_0)
  \end{align}
  \footnote{Notation: $B_r(u_0,v_0)$ denotes the ball of radius $r$ centered at $(u_0,v_0)\in \mathcal{Q_R}$ with respect to the Euclidean metric.}for every $u \in [u_0, u_1]$.
  Identically, as $\partial_u \psi$ and $\partial_v \psi$ are 
  bounded below by $-c$ on $\mathcal{C}_+ \cap [v_0, v_1]$ and continuous in $\mathcal{A}(1)$, there is a $\delta_R>0$ such that 
  \begin{align}
      \partial_u \psi \leq -c/2\  \text{  and  }\  \partial_v \psi \leq -c/2 \ \text{ in }B_{\delta_R}(u_0,v)
  \end{align}
  for every $v\in[v_0,v_1]$. Let $t_0 = \min\{ \delta_L, \delta_R \}$ 
  \begin{align}
      \implies \partial_u \psi \leq -c/2\  \text{  and  }\  \partial_v \psi \leq -c/2 \ \text{ in } \mathcal{A}(t_0)
  \end{align}
  By (\ref{dupsi_exprn} - \ref{dvpsi_exprn}), for $(u,v) \in \mathcal{A}(t_0)$ we have 
  \begin{align}
      &\partial_u \psi (u,v) \leq -c - \frac{c}{2r(u_0, v_0)}(v-v_0)\norm{r_{,u}(u, \cdot)}_{L^1[v_0,v]}<-c \label{improve_b_u} \\
      &\partial_v \psi (u,v)  \leq -c - \frac{c}{2r(u_0,v_0)}(u-u_0)\norm{r_{,v}(\cdot, v)}_{L^1[u_0,u]}<-c  \label{improve_b_v}
  \end{align}
  As the bootstrap assumptions hold on $\mathcal{A}(t_0)$, $t_0 \in \mathcal{B}$. So the first point is proved. \\[5pt]
  Now suppose the bootstrap assumptions hold on $\mathcal{A}(t_0)$ for $t_0>0$. We can repeat the argument of the first point, swapping $\overline{\mathcal{C}_-}$ for $[u_0, u_1]\times \{v_0 + t_0 v_1\}$ and $\mathcal{C}_+$ for $\{u_0 + t_0 (u_1-u_0)\} \times [v_0, v_1]$. By calculations identical to (\ref{improve_b_u}-\ref{improve_b_v}) we get that $\exists \ \delta>0$ such that $\partial_u \psi <- c$ and $\partial_v \psi <-c$ on $\mathcal{A}(t_0+\delta)$.  It follows that $\mathcal{B}$ is open. We conclude that $\mathcal{B} = [0,1]$.
\end{proof}
\end{prop}

\begin{cor}\label{monotonicity_is_propagated_in_D+}
    Under the hypotheses of Theorem \ref{wave_blow_up_with_monotonicity}, the monotonicity properties $\partial_u \psi<0$ and $\partial_v\psi<0$ hold on $D^+(\mathcal{C})$.
    \begin{proof} In Proposition \ref{monotonocity_is_propagated_in_D} we showed that for every $v_1\in(v_0/2, 0)$ there exists a $c=c(v_1)>0$ such that $\partial_u \psi\leq -c(v_1)$ and $\partial_v \psi \leq- c(v_1)$ in $[u_0,u_1]\times[v_0, v_1]$. Therefore on $D^+(\mathcal{C}) = [u_0,u_1]\times \bigcup_{v_1<0} [v_0, v_1]$ we have $\partial_u \psi<0$ and $\partial_v\psi <0$.
    \end{proof}
\end{cor}

Having shown that monotonicity of the derivatives is propagated throughout the domain of $\psi$ all the way to the singularity, we are ready to prove Theorem \ref{wave_blow_up_with_monotonicity}.\\[5pt]
\textit{Proof of Theorem \ref{wave_blow_up_with_monotonicity}.}
From the wave equation (\ref{wave_eqn_du}), applying the monotonicty assumptions on the initial data $\partial_v \mathring{\psi}(u_0, v)< 0$ and $\partial_u \mathring{\psi}(u,v_0)<0$, we find
\begin{align}
    \partial_v \psi(u,v) &= \frac{1}{r(u,v)} \bigg[ r(u_0, v) \partial_v \mathring{\psi}(u_0, v) - \int_{u_0}^u r_{,v}(u',v) \partial_u \psi(u', v) du'  \bigg] < \nonumber\\
    & <-\frac{1}{r(u_0, v_0)} \bigg|\int_{u_0}^u r_{,v}(u',v) \partial_u \psi(u', v) du'\bigg| = \frac{1}{r(u_0, v_0)} \int_{u_0}^u r_{,v}(u',v) |\partial_u \psi(u', v)| du' \leq \nonumber\\
    &\leq  \frac{\inf_{\overline{\mathcal{C}_-}}{|\partial_u\mathring{\psi}}|}{r(u_0, v_0)} \int_{u_0}^u r_{,v}(u', v) du' 
\end{align}
    where in the equality we used Corollary \ref{monotonicity_is_propagated_in_D+} to remove the absolute value around the integral, as well as the property that $r_{,v}<0$ in $D^+(\mathcal{C})$.
$\qed$
\begin{figure}[h]
    \centering
    \includegraphics[width=0.5\linewidth]{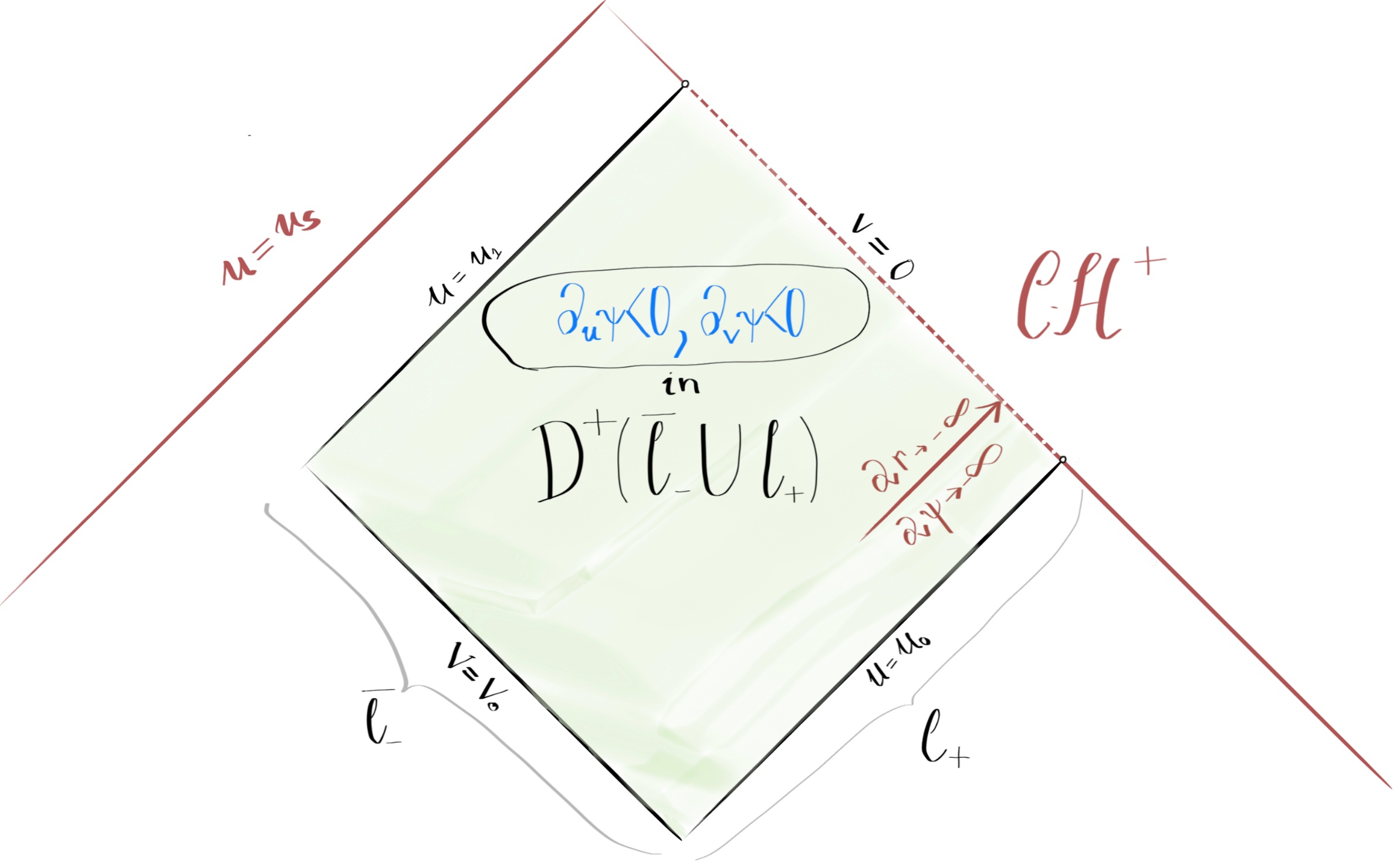}
    \caption{Illustrating the blow-up of the solution to $\Box_g\psi=0$ arising from smooth initial data on the bifurcate null hypersurface $\mathcal{C}$.}
    \end{figure}

\begin{cor}\label{dupsi_upper_bound} The ingoing null derivative of $\psi$ is bounded away from zero in $D^+(\mathcal{C})$. More precisely, for every 
\begin{align}
|\partial_u \psi (u,v)| \geq \frac{r(u_1,v_0)}{r(u_0,v_0)} \inf_{\overline{\mathcal{C}_-}}|\partial_u \mathring{\psi}|  \ \ \forall(u,v)\in D^+(\mathcal{C})
\end{align}
\begin{proof}
    This follows directly from expression (\ref{dupsi_exprn}) by applying the monotonicity of $r$ inside the domain $D^+(\mathcal{C})$.
\end{proof}
    
\end{cor}

\subsubsection{$\psi$ stays bounded}
In the following proposition we show that $\psi$ stays bounded as $v\to0$ for every $u\in [u_0,u_1]$, regardless of whether $\partial_u \psi$ blows up as $v\to 0$ or not. The result is crucial to showing that the ratio $\partial_u \psi/\partial_v \psi \to 0$ as we approach the weak null singularity. We will later demonstrate how this implies that the velocity of the stiff fluid approaches an ingoing null vector as $v\to 0$. \\[5pt]
Firstly we show that the initial data $\mathring{\psi}$ for the characteristic IVP for the wave equation, as defined in terms of the initial data $(\mathring{\rho},\mathring{U})$ for the Euler equations for the stiff fluid, does not blow up as $v\to 0$.
\begin{lem}\label{psi_initially_bounded} On the initial hypersurface $\mathcal{C}$, $\mathring{\psi}$ is bounded. More precisely,
\begin{align}     &\sup_{u\in[u_0,u_1]}|\mathring{\psi}(u,v_0)|<\infty \label{psi_bounded_on_Cminus}\\
    &|\mathring{\psi}(u_0,v)|\leq |\mathring{\psi}(u_0,v_0)| +\frac{C_g}{8\pi}\norm{\sqrt{\mathring{\rho}}\mathring{U}^u(u_0,\cdot)}_{L^1(\mathcal{C}_+)}  \ \ \text{for every }v\in[v_0,0) \label{psi_bounded_on_Cplus}
\end{align}
\begin{proof}
    The first estimate (\ref{psi_bounded_on_Cminus}) holds because $\mathring{\psi}(\cdot,v_0)$ is a continuous function on a closed interval, namely $[u_0,u_1]$.
    For the second estimate (\ref{psi_bounded_on_Cplus}), we apply the definition of the gradient of $\psi|_{\mathcal{C}_+}$ in terms of the initial values $(\mathring{\rho},\mathring{U})$, of the energy density and the stiff fluid velocity $(\rho, U)$ from Theorem \ref{masterTheorem_2}. In coordinates we have $\sqrt{\mathring{\rho}} \mathring{U}^u= -2e^{-2\omega}\partial_v \mathring{\psi}$ on $\mathcal{C}_+$. It follows that on $\mathcal{C}_+$
    \begin{align}
        |\mathring{\psi}(u_0, v)| = \Bigg|\mathring{\psi}(u_0,v_0) -\frac{1}{2}\int_{v_0}^v e^{2\omega} \sqrt{\mathring{\rho}}\mathring{U}^u(u_0,v') dv' \Bigg| \leq |\mathring{\psi}(u_0,v_0)| + \frac{C_g}{8\pi} \norm{\sqrt{\mathring{\rho}}\mathring{U}^u(u_0,\cdot)}_{L^1(\mathcal{C}_+)} 
    \end{align}
    Where we applied Assumptions \ref{assumptions_on_L1_norms} and the estimate (\ref{new_assumed_L1_bound_masterThm2}) from the hypotheses of Theorem \ref{masterTheorem_2}.
\end{proof} 
\end{lem}

\begin{prop}\label{psi_stays_bounded}
    There exist constants $0<c<C$, depending on $r_0:=r(u_0, v_0)$, $r_1:= \lim_{v\to0} r(u_1, v)$, the initial data $\mathring{\psi}|_{\mathcal{C}}$ such that
    \begin{align}
        c\leq \psi (u,v) \leq C \text{ in } D^+(\mathcal{C})
    \end{align}
    \begin{rem}
        See the proof for explicit expressions of the constants $c$ and $C$. \vspace{-10pt}
    \end{rem}
    \begin{proof}
        For $n \in \mathbb{N}$ let
        \begin{center}
        $u^{(n)} = \begin{cases}
            u_0 + \frac{2n}{3}\frac{r_1}{\norm{\partial_u r}_{C(\mathcal{R})}}, \text{ if this is }<u_1\\
            u_1, \text{  otherwise}
        \end{cases}$
        \end{center}
        We claim that $\psi$ is bounded as $v\to 0$ for $u\in[u_0, u^{(n)}]$ for all $n\in \mathbb{N}$. We prove the claim by induction on $n$. Base case is $n=1$. Hence $u^{(1)} = u_0 + \frac{2}{3}\frac{r_1}{\norm{\partial_u r}_{C(\mathcal{R})}}$. So let $u\in[u_0, u^{(1)}]$.
        \begin{align}
            &\mathring{\psi}(u_0, v) - \psi(u,v) = \int_{u_0}^u -\partial_u \psi(u',v) du' = \nonumber\\
            &=\int_{u_0}^u -\frac{1}{r(u',v)} \bigg[r(u', v_0) \partial_u \mathring{\psi}(u', v_0) - \int_{v_0}^v \partial_u r(u',v) \partial_v \psi (u',v') dv'\bigg] du'\label{basecase_inductiveproof}
        \end{align}
        By an application of the fundamental theorem of calculus, the first term in the integral is bounded by 
        \begin{align}
            \frac{r_0}{r_1} \big( \mathring{\psi}(u_0, v_0) - \mathring{\psi} (u,v_0)\big) 
        \end{align}
        By Assumptions \ref{instability_assumptions}, $\norm{\partial_u r}_{C(\mathcal{R})}<\infty$. The second term is therefore bounded by 
        \begin{align}
            \int_{u_0}^u \frac{\norm{\partial_u r}_{C(\mathcal{R})}}{r_1} \big[\mathring{\psi}(u', v_0) - \psi (u',v)\big] du' \label{secondterm_basecase_induction}
        \end{align}
        As $\psi$ is decreasing in $u$ and $v$, it follows that for every $u'\in [u_0, u]$ we have $\mathring{\psi}(u_0, v_0) \geq \mathring{\psi} (u', v_0) \geq \psi (u', v) \geq \psi (u,v)$. In particular, $\mathring{\psi}(u', v_0) - \psi (u',v) \leq \mathring{\psi}(u_0,v_0) - \psi(u,v)$. Hence the expression in (\ref{secondterm_basecase_induction}) is bounded by 
        \begin{align}
            \frac{\norm{\partial_u r}_{C(\mathcal{R})}}{r_1} \big[\mathring{\psi}(u_0, v_0) - \psi (u,v)\big](u-u_0) 
        \end{align}
        It follows that 
        \begin{align}
            \mathring{\psi}(u_0, v) - {\psi}(u,v)  \leq \frac{r_0}{r_1} \big( \mathring{\psi}(u_0, v_0) - \mathring{\psi} (u,v_0)\big) +  \frac{\norm{\partial_u r}_{C(\mathcal{R})}}{r_1} \big[\mathring{\psi}(u_0, v_0) - \psi (u,v)\big](u-u_0) 
        \end{align}
        Rearranging and using that $\frac{\norm{\partial_u r}_{C(\mathcal{R})}}{r_1}(u-u_0) \leq \frac{2}{3} $, we find  
        \begin{align}
            \bigg(1-\frac{\norm{\partial_u r}_{C(\mathcal{R})}}{r_1}(u-u_0) \bigg)^{-1}\bigg( \mathring{\psi}(u_0, v) - \frac{r_0}{r_1} \big( \mathring{\psi}(u_0,v_0) - \mathring{\psi}(u,v_0)\big) - \frac{\norm{\partial_u r}_{C(\mathcal{R})}}{r_1}(u-u_0) \mathring{\psi}(u_0,v_0)\bigg) \leq \psi(u,v) \label{basecase3}
        \end{align}
         Extend the initial data $\mathring{\psi}(u_0, v)$ to the singularity by defining $\lim_{v\to 0} \mathring{\psi}(u_0, v)=:\mathring{\psi}(u_0,0)$. This limit exists by the monotone convergence theorem, which applies because of the assumption that the initial data $\mathring{\psi}|_{\mathcal{C}_+}$ is monotonically decreasing in $v$ and bounded below on $\mathcal{C}_+$ by Lemma \ref{psi_initially_bounded}. By the monotonicity  of the initial data $\mathring{\psi}$, we have $\mathring{\psi}(u_0,v)\geq \mathring{\psi}(u_0,0)$ and $\mathring{\psi}(u,v_0) \geq \mathring{\psi}(u_1,v_0)$. Hence (\ref{basecase3}) becomes
        \begin{align}
            \frac{r_1}{r_1  - \norm{\partial_u r}_{C(\mathcal{R})}(u-u_0)} \mathring{\psi}(u_0, 0) + \frac{r_0}{r_1  - \norm{\partial_u r}_{C(\mathcal{R})}(u-u_0)} \mathring{\psi}(u_1, v_0) - \bigg( \frac{\norm{\partial_u r}_{C(\mathcal{R})}}{r_1} &(u-u_0) + \frac{r_0}{r_1}\bigg) \mathring{\psi}(u_0,v_0) \leq\nonumber\\
            &\leq\psi(u,v)
        \end{align}
        Combining $u>u_0$ with $r_1 - \norm{\partial_u r}_{C(\mathcal{R})}(u-u_0)\in [r_1/3,r_1]$, it follows that $r_1/(r_1 - \norm{\partial_u r}_{C(\mathcal{R})}(u-u_0))\in [1,3]$. As $r_0>r_1$, it follows that $r_0/(r_1 - \norm{\partial_u r}_{C(\mathcal{R})}(u-u_0))\in [r_0/r_1,3r_0/r_1] $. Finally, by the same bounds we have $\frac{\norm{\partial_u r}_{C(\mathcal{R})}}{r_1}(u-u_0) + \frac{r_0}{r_1} \in[r_0/r_1, r_0/r_1+2/3]$. In any case\footnote{Meaning, whatever the sign of $\mathring{\psi}$ is at the points $(u_0,v_0), (u_0, 0)$ and $(u_1,v_0)$.}
        \begin{align}
            -3 |\mathring{\psi}(u_0, 0)| - \frac{3r_0}{r_1}|\mathring{\psi}(u_1,v_0)| - \bigg(\frac{r_0}{r_1}+\frac{2}{3}\bigg)|\mathring{\psi}(u_0, v_0)| \leq \psi(u,v)\leq \mathring{\psi}(u_0, v_0)
        \end{align}
        where the upper bound follows from the monotonicity of $\psi$ throughout $D^+(\mathcal{C})$ and is included for completeness. This concludes the proof of the base case.\\[5pt]
        \begin{figure}[h]
            \centering
            \includegraphics[width=0.45\linewidth]{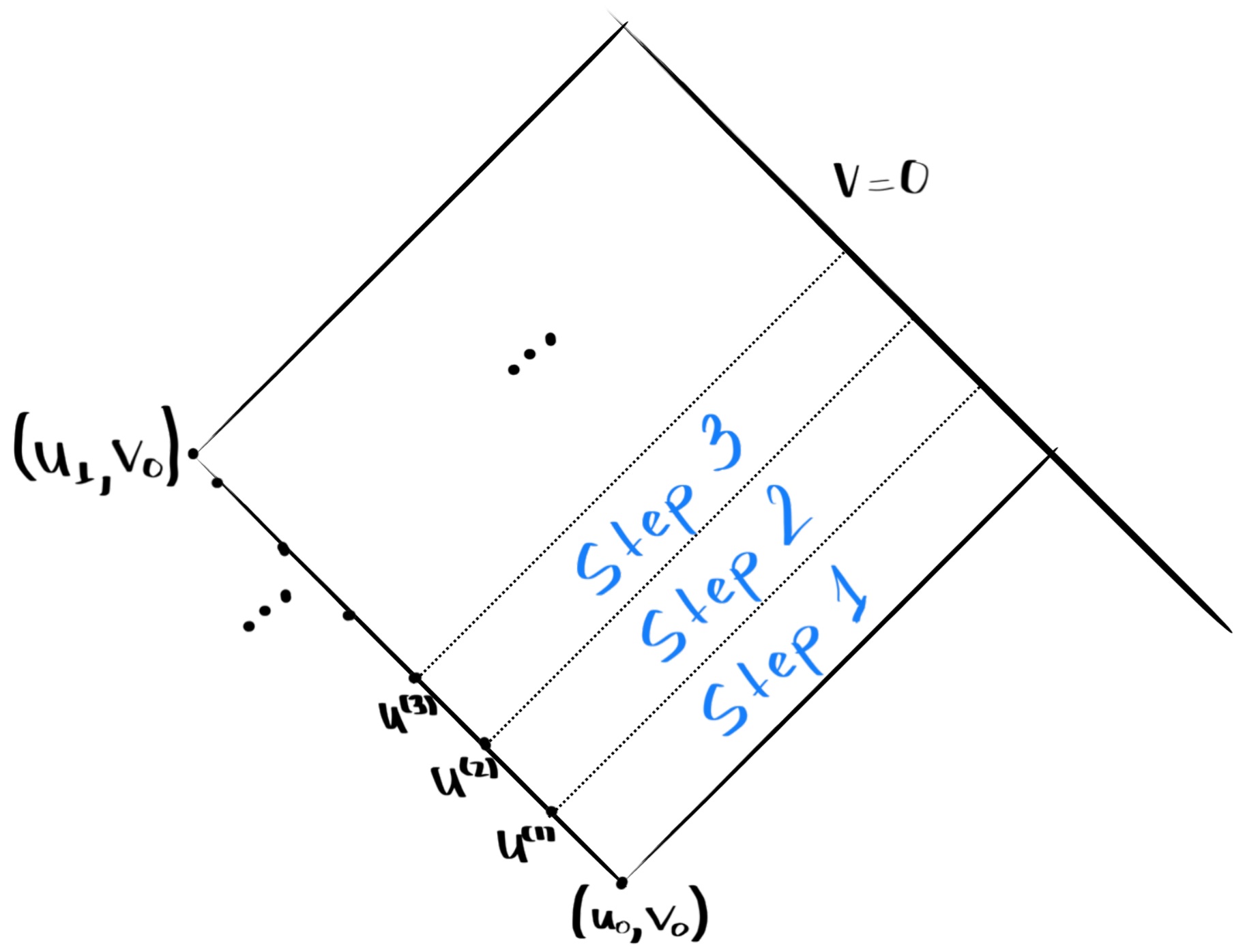}
            \caption{Illustration of the method of proof of Proposition \ref{psi_stays_bounded}. Step 1 is proving that $\psi$ stays bounded in $[u_0, u^{(1)}]$. In the definition of the sequence $\{u^{(n)}\}$, two consecutive elements are close enough so that $(1-\norm{\partial_u r}_{C(\mathcal{R})}(u-u^{(n)})/r_1)$ is positive for $u\in [u^{(n)}, u^{(n+1)}]$. This  allows us to pull $\psi(u,v)$ on the RHS of equation \ref{basecase3} for the case $n=0$ and a perform a similar manipulation for $n>0$. This is crucial for the proof. In Step $n$ we leverage that $\psi$ stays bounded in $[u_0, u^{(n)}]$ and use $\{u=u^{(n)}\}$ as the new initial ingoing hypersurface to prove that $\psi$ stays bounded in $[u_0, u^{(n+1)}]$.}
            \label{induction_illustration}
        \end{figure}
        The induction step is essentially repeating the base case with the slight modification that we swap the original initial ingoing hypersurface $\{u=u_0\}$ for the advanced ingoing hypersurface $\{u=u^{(n)}\}$. An illustrative Penrose diagram is shown in Figure \ref{induction_illustration}. Nevertheless, we include the main steps for completeness. Assume that $\psi$ is bounded as $v\to 0$ for $u \in [u_0, u^{(n)}]$  and let $u \in [u^{(n)}, u^{(n+1)}]$. As in the base case, we begin by expressing the difference:
        \begin{align}
             &\psi(u^{(n)}, v) - \psi(u,v) = \int_{u^{(n)}}^u -\partial_u \psi(u',v) du' = \nonumber\\
            &\int_{u^{(n)}}^u -\frac{1}{r(u',v)} \bigg[r(u', v_0) \partial_u \mathring{\psi}(u', v_0) - \int_{v_0}^v \partial_u r(u',v) \partial_v \psi (u',v') dv'\bigg] \label{inductivestep}
        \end{align}
        Analogously to the base case, the first term in (\ref{inductivestep}) is bounded by 
        \begin{align}
            \frac{r_0}{r_1} \big( \mathring{\psi}(u^{(n)}, v_0) - \mathring{\psi} (u,v_0)\big) 
        \end{align}
        and the second term in (\ref{inductivestep}) is bounded by 
        \begin{align}
             \frac{\norm{\partial_u r}_{C(\mathcal{R})}}{r_1} \big[\mathring{\psi}(u^{(n)}, v_0) - \psi (u,v)\big](u-u^{(n)}) 
        \end{align}
        Since $\partial_v\psi$ and $\partial_u \psi$ are strictly negative and since $\lim_{v\to 0} \psi(u^{(n)}, v)>-\infty$, $\psi({u^{(n)}, v})$ can be extended to the singularity. Define $\lim_{v\to 0} \psi(u^{(n)}, v)=: \psi(u^{(n)},0)$. Using that $\frac{\norm{\partial_ur}_{C(\mathcal{R})}}{r_1}(u-u^{(n)}) \in [0, 2/3]$ and applying exactly the same steps as in the base case but with $u^{(n)}$ instead of $u_0$, we find \begin{align}
            &\mathring{\psi}(u_0,v_0) \geq \psi(u,v) \geq \nonumber\\
            &\geq\frac{r_1}{r_1  - \norm{\partial_u r}_{C(\mathcal{R})}(u-u^{(n)})} {\psi}(u^{(n)}, 0) + \frac{r_0}{r_1  - \norm{\partial_u r}_{C(\mathcal{R})}(u-u^{(n)})} \mathring{\psi}(u_1, v_0) -\nonumber\\
            &-\bigg( \frac{\norm{\partial_u r}_{C(\mathcal{R})}}{r_1} (u-u^{(n)}) + \frac{r_0}{r_1}\bigg) \mathring{\psi}(u_0,v_0) \geq\nonumber\\
            &\geq -3|\psi(u^{(n)}, 0)| - 3\frac{r_0}{r_1} |\mathring{\psi}(u_1,v_0)| - \bigg(\frac{r_0}{r_1} +\frac{2}{3}\bigg)|\mathring{\psi}(u_0, v_0)|
        \end{align} 
        This proves that $\psi$ is bounded as $v\to 0$ for $u \in [u^{(n)}, u^{(n+1)}]$ hence concluding the induction step.
    \end{proof}
\end{prop}
\begin{cor}\label{ratio_of_derivatives_of_psi}
    For every $u \in (u_0, u_1]$,
    \begin{align}
        \frac{\partial_u \psi(u,v)}{ \partial_v \psi (u,v)} \to 0 \ \ \text{ as } \ v \to 0 
    \end{align}
    \begin{proof}
        Let $u\in (u_0, u_1]$. Using (\ref{dupsi_exprn}- \ref{dvpsi_exprn}) and applying Theorem \ref{wave_blow_up_with_monotonicity} we find
        \begin{align}
            &\frac{\partial_u \psi}{\partial_v \psi} (u,v) = \frac{\frac{1}{r(u,v)} \bigg(\partial_u \mathring{\psi} (u,v_0) r(u,v_0) - \int_{v_0}^v \partial_u r(u,v') \partial_v \psi(u,v')\bigg)}{\partial_v \psi(u,v)} \leq\nonumber\\
            &\leq\frac{r(u_0, v_0) {|\partial_u \mathring{\psi}(u,v_0)|} r (u,v_0)} {r(u,v)\inf_{\overline{{\mathcal{C}}_-}}|\partial_u\mathring{\psi}|\int_{u_0}^u(-\partial_v r)(u',v) du'} + \frac{r(u_0, v_0) \norm{\partial_u r}_{C(\mathcal{R})} \big( \psi(u,v_0) - \psi(u,v)\big)}{r(u,v) \inf_{\overline{\mathcal{C}_-}}|\partial_u\mathring{\psi}| \int_{u_0}^u(-\partial_v r)(u',v) du'}
        \end{align}
    The numerator in the first term is obviously bounded. By Proposition \ref{psi_stays_bounded}, the numerator in the second term is bounded. By the assumption that for every $u\in[u_0,u_1]$ we have $\partial_vr(u,v) \to -\infty$ as $v\to 0$ (c.f. Assumptions \ref{instability_assumptions}) and by the proof of Corollary \ref{integral_of_dvr_blows_up} the denominator in both terms goes to infinity as $v\to 0$.
    \end{proof}
\end{cor}
\subsection{Energy blow-up - concluding the proof of Theorem \ref{masterTheorem_2}}
For the stiff fluid velocity it follows by Theorem \ref{wave_blow_up_with_monotonicity} that
\begin{align}
     V^u = \partial^u\psi = g^{uv} \partial_v \psi = -2 e^{-2\omega} \partial_v \psi \to \infty \text{  as  } v\to 0 \text{  for every } u \in(u_0, u_1]
\end{align}
Furthermore, for the fluid energy density we have
\begin{align}
    \rho = - \rho \langle U, U \rangle = - \langle V, V \rangle = - \langle d\psi , d\psi \rangle = 4e^{-2\omega} \psi_{,u} \psi_{,v} \to \infty \text{  as  } v\to 0 \text{ for every $u\in(u_0, u_1]$}
\end{align}
by Theorem \ref{wave_blow_up_with_monotonicity}, Corollary \ref{dupsi_upper_bound} and Assumptions \ref{instability_assumptions}.
This establishes item 1 of Theorem \ref{masterTheorem_2}. Moreover, for the fluid velocity components we have
\begin{align}
    &U^u = \frac{V^u}{\sqrt{\rho}} = - \frac{2e^{-2\omega} \psi_{,v}}{\sqrt{4e^{-2\omega} \psi_{,u} \psi_{,v}}} = e^{-\omega}\sqrt{\frac{\psi_{,v}}{\psi_{,u}}} \to \infty \text{   as  } v\to 0 \text{  for every  } u \in(u_0, u_1]\\
    &U^v = \frac{V^v}{\sqrt{\rho}} = - \frac{2e^{-2\omega} \psi_{,u}}{\sqrt{4e^{-2\omega} \psi_{,u} \psi_{,v}}} = e^{-\omega}\sqrt{\frac{\psi_{,u}}{\psi_{,v}}} \to 0 \text{   as  } v\to 0 \text{  for every  } u \in(u_0, u_1]
\end{align}
where we used Corollary \ref{ratio_of_derivatives_of_psi}. This shows that the stiff fluid velocity approaches a null vector at the weak null singularity, establishing item 2 of Theorem \ref{masterTheorem_2}. The stiff fluid velocity flow is illustrated in Figure \ref{stf_flowlines}. This concludes the proof of our second main result.
\begin{figure}[H]
    \centering
    \includegraphics[width=0.45\linewidth]{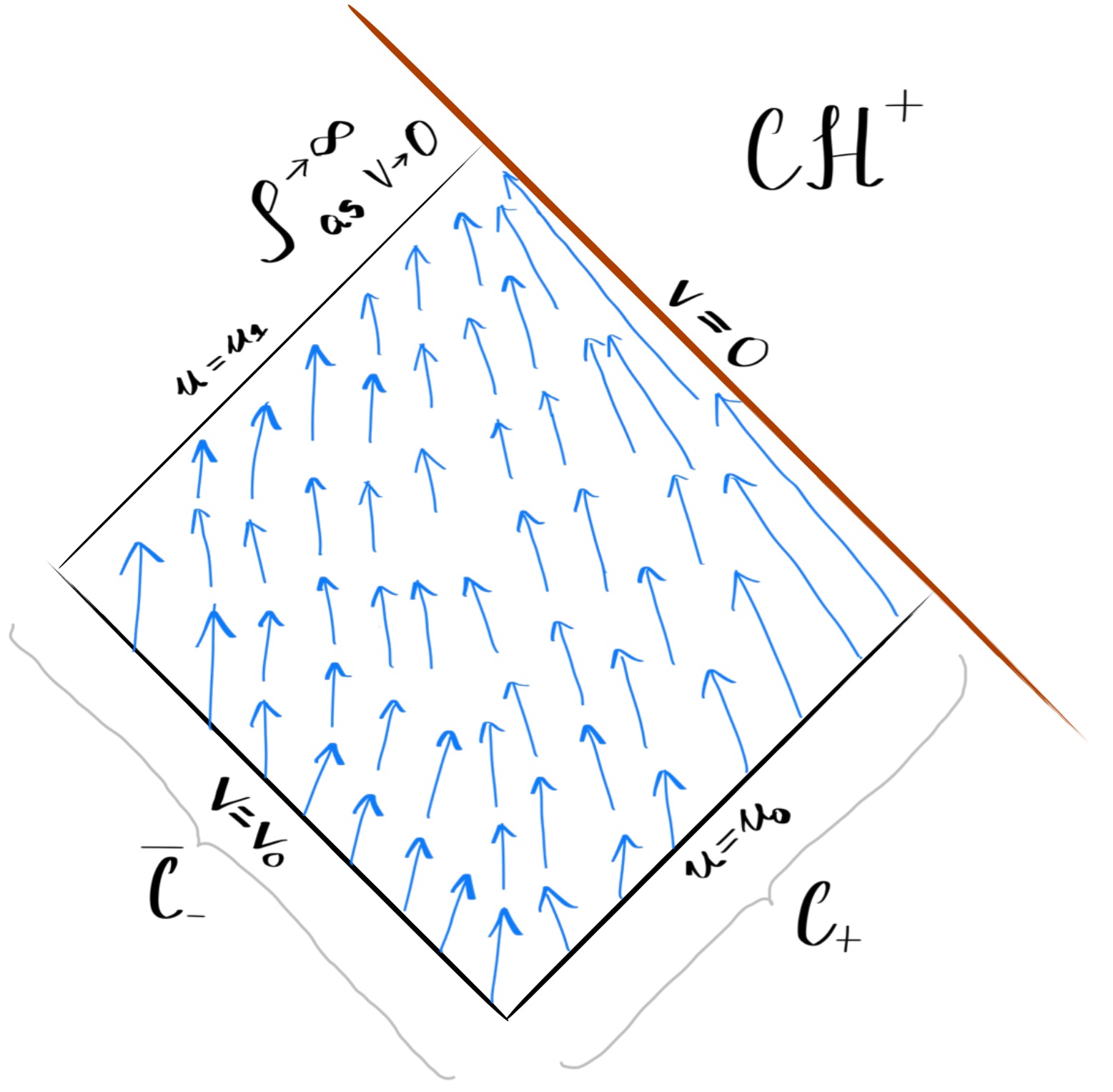}
    \caption{Illustrating the flow lines of the stiff fluid four-velocity $U$. They approach ingoing null curves as we approach the weak null singularity. The stiff fluid energy density goes to infinity at the singularity.}
    \label{stf_flowlines}
\end{figure}

\section{The Einstein--Maxwell--scalar field system giving an 'eligible' black hole spacetime}\label{EMSf_chapter}
In this part we revisit the class of spherically symmetric black hole spacetimes analysed in the work of Luk and Oh \cite{Luk_Oh}, arising from the Einstein--Maxwell--scalar field system. We demonstrate that a codimension-0 submanifold of a spacetime $(\mathcal{M},g)$ of this class satisfies the global geometric assumptions. This shows that there is a wide class of  spacetimes on which our results for the behaviour of spherically symmetric dust versus spherically symmetric stiff fluid hold. 
\subsection{Background spacetime: Scalar field perturbations of\\
Reissner-Nordstr\"{o}m} The following setup is based primarily on Theorem 4.1 in \cite{Luk_Oh} but the description of the Einstein--Maxwell--scalar field system in spherical symmetry from Section 2 of \cite{Luk_Oh} and the conditions on the initial data from Section 3 of \cite{Luk_Oh} will also be used. The spherically symmetric spacetime $(\mathcal{M}, g, F, \phi)$ analysed in \cite{Luk_Oh}, with Penrose diagram in Figure \ref{full_manifold_penrose_diagram}, is the Cauchy development of smooth, admissible, asymptotically flat, spherically symmetric Einstein--Maxwell--scalar field Cauchy data supplied on a spacelike hypersurface $\Sigma_0$ with two asymptotic ends. The admissibility of the initial data is in the sense of Definition 3.1 in \cite{Luk_Oh}. The spacetime is spherically symmetric in the sense of Definition 2.1 of \cite{Luk_Oh}. Due to the spherical symmetry, $\mathcal{M}$ admits a quotient manifold $\mathcal{Q}=\mathcal{M}/SO(3)$. The metric on $\mathcal{M}$ may be written as
\begin{align}
    g=g_{\mathcal{Q}}+r^2d\Omega_2^2
\end{align}
where $g_{\mathcal{Q}}$ is the metric on the quotient manifold, $d\Omega_2^2$ is the unit metric on $S^2$ and $r:\mathcal{Q}\to \mathbb{R}_{\geq 0}$ is the area radius function of the spacetime, defined by the property that the sphere at point $p\in \mathcal{Q}$ has area $4\pi r(p)^2$. The area radius function is smooth and positive on $\mathcal{M}$.
The domain of outer communications $\mathcal{E}$ of $\mathcal{M}$ consists of two connected components $\mathcal{E}_1$ and $\mathcal{E}_2$ which we will call the \textit{right and left exterior regions}, respectively. They correspond to the two asymptotic ends of $\Sigma_0$. Each connected component has a complete future null infinity $\mathcal{I}_i$ (for $i=1,2$) and $\mathcal{E}_1 = J^-(\mathcal{I}_1^+)\cap J^+(\Sigma_0)$ and $\mathcal{E}_2 = J^-(\mathcal{I}_2^+)\cap J^+(\Sigma_0)$ and $\mathcal{E}_i$ approaches a connected component of the exterior region
of a subextremal Reissner–Nordstr\"{o}m spacetime. The future null infinities have the property that the area radius function $r$ diverges to infinity along any future-directed causal curve in $\mathcal{M}$ extending to $\mathcal{I}^+_i$. The past endpoints of $\mathcal{I}_i^+$, coinciding with the endpoints of $\Sigma_0$, correspond to the \textit{spacelike infinities $i_1^0$ and $i_2^0$}. The future endpoints of $\mathcal{I}^+_i$ correspond to the \textit{timelike infinities $i^+_1$ and $i^+_2$.} 
The exterior regions $\mathcal{E}_i$ are bounded to the past by a segment of $\Sigma_0$ and to the future by the event horizons\footnote{Given an open subset $\mathcal{O}$ of $\mathcal{M}$, $\partial^+\mathcal{O}$ denotes the \textit{future boundary} of $\mathcal{O}$} 
\begin{align}
    &\mathcal{H}_1^+ = \partial^+ J^-(\mathcal{I}_1^+); \ & \mathcal{H}_2^+ = \partial^+ J^- (\mathcal{I}_2^+)
\end{align}
The event horizons $\mathcal{H}_1^+$ and $\mathcal{H}_2^+$ are future-complete null hypersurfaces\footnote{$\mathcal{H}_i^+$ ($ i=1,2$) extend to $i_i^+$ to the future.} with past endpoints on $\Sigma_0$. 
Each exterior region $\mathcal{E}_i$ is globally hyperbolic and can be covered by a single double null chart.
The union of the right and left event horizons, $\mathcal{H}:=\mathcal{H}_1\cup\mathcal{H}_2$, together with the segment of $\Sigma_0$ between the endpoints of $\mathcal{H}_1^+$ and $\mathcal{H}_2^+$, constitute the past boundary of a spherically symmetric charged black hole $\mathcal{B}:=\mathcal{M}\backslash\mathcal{E}$. The future boundary of $\mathcal{B}$ consists of the \textit{right Cauchy horizon} $\mathcal{CH}^+_1$, the \textit{left Cauchy horizon} $\mathcal{CH}_2^+$ and a possibly empty achronal set $\mathcal{S}$. The Cauchy horizons are half-open\footnote{The convention is that $\mathcal{CH}_i^+$ ($i=1,2$) include their future endpoints. To the past, they extend to $i_i^+$.} null segments with the property that $r$ extends continuously as a positive function to $\mathcal{CH}^+:=\mathcal{CH}_1^+ \cup \mathcal{CH}_2^+$ except possibly at the future endpoints of $\mathcal{CH}_i^+$. The achronal set $\mathcal{S}$, if non-empty, shares endpoints with the future endpoints of the Cauchy horizons and has the property that $r$ extends continuously to zero on $\mathcal{S}$. For the final\footnote{The spacetime is dynamical and settles to {subextremal Reissner-Nordstr\"{o}m} with parameters $(M,q)$.} black hole mass and charge parameters, subextremality is assumed: $0<|q|<M$. The scalar function $\phi\in C^\infty(\mathcal{M}\to \mathbb{R})$ is the (real) massless scalar field. The spacetime is also characterised by the Maxwell 2-form $F=dA:\mathcal{M}\to \Lambda^2(\mathcal{M})$ where the 1-form $A$ denotes the electromagnetic potential. The set $(\mathcal{M}, g, \phi, F)$ solves the Cauchy problem for the Einstein--Maxwell--scalar field system in spherical symmetry with initial data on $\Sigma_0$ satisfying the admissibility conditions posed in Definition 3.1 of \cite{Luk_Oh}. See Definition 2.1 of \cite{Luk_Oh} or Appendix \ref{appendix_evol_eqns_and_constraints} for the relevant Einstein equations.
\begin{figure}
    \centering
    \includegraphics[width=0.7\linewidth]{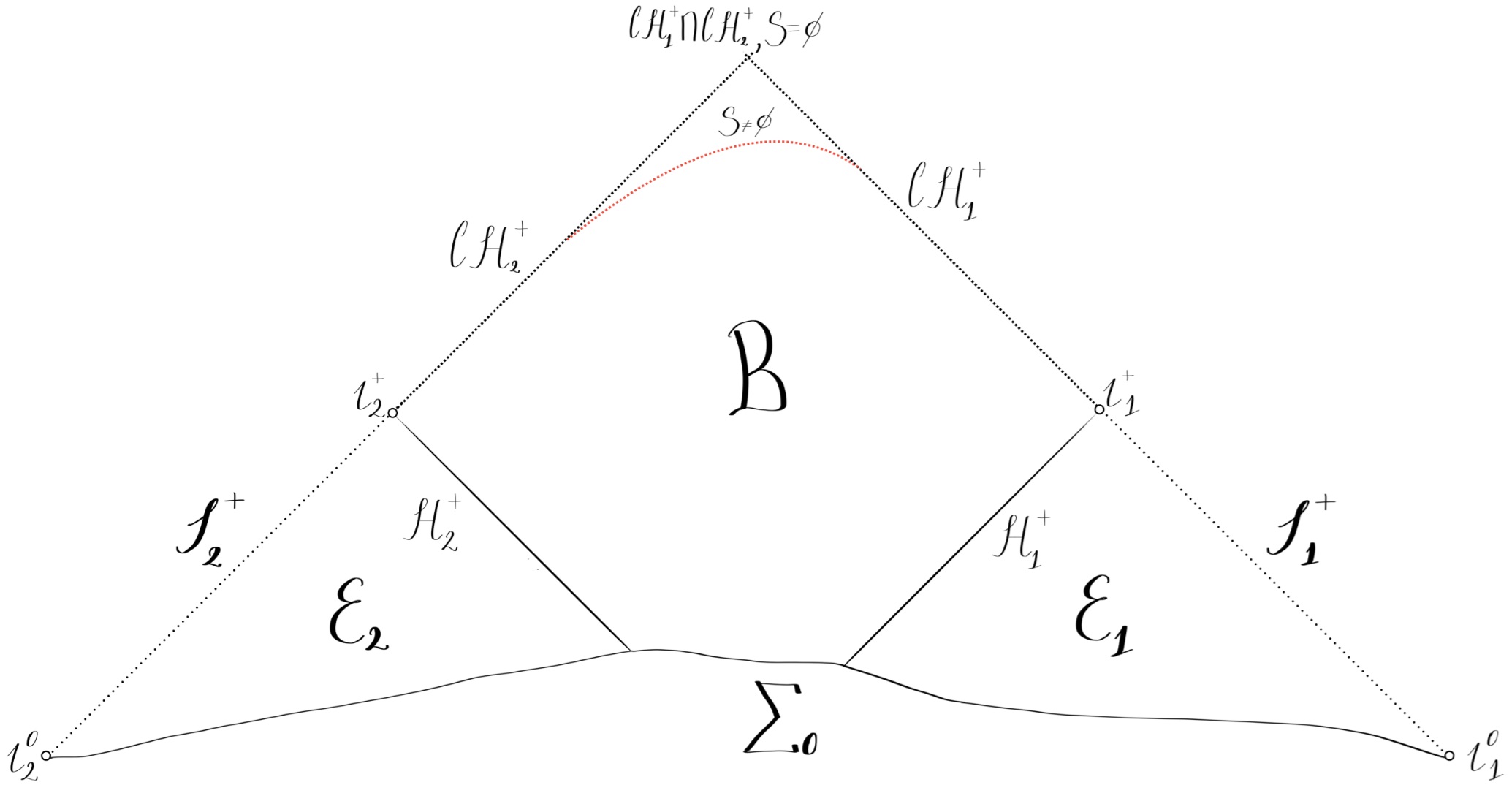}
    \caption{Penrose diagram of $(\mathcal{M},g)$, covering the cases where $\mathcal{S}=\emptyset$ (the null punctured black line of $\mathcal{CH}_1^+$ and $\mathcal{CH}_2^+$ meet) and where $\mathcal{S}\neq \emptyset$ (spacelike punctured red line emanating from the punctured black lines where $\mathcal{CH}_1^+$ and $\mathcal{CH}_2^+$). The initial data for the Cauchy problem studied in \cite{Luk_Oh} is supplied on the 2-ended spacelike asymptotically flat hypersurface $\Sigma_0$. In the final development, two asymptotic regions $\mathcal{E}_1$ and $\mathcal{E}_2$ are present, as well as one black hole region, $\mathcal{B}$.}
    \label{full_manifold_penrose_diagram}
\end{figure}
 \subsubsection{Parameterizing the black hole interior}
The open submanifold $\mathcal{B}$, describing the black hole interior, has topology $\mathcal{B} = \mathcal{Q_B}\times S^2$
where the quotient manifold $\mathcal{Q_B} = \mathcal{B}/SO(3)$ is an open subset of $\mathbb{R}^2$. We define $\overline{\mathcal{B}}:=\mathcal{B}\cup \mathcal{CH}^+\cup \mathcal{S}\cup \mathcal{H}^+\cup\Sigma_{0,\mathcal{B}}$ where by $\Sigma_{0,\mathcal{B}}$ we denote the connected subset of $\Sigma_0$ between the past endpoints of $\mathcal{H}_1$ and $\mathcal{H}_2$ i.e. $\Sigma_{0,\mathcal{B}}:=\Sigma_0\backslash J^-(\mathcal{I}_1^+\cup \mathcal{I}_2^+)$. We also define the quotient $\overline{\mathcal{Q_B}}:=\overline{\mathcal{B}}/SO(3)$.\\[5pt]
$\overline{\mathcal{B}}$ should be viewed as a 4-manifold-with-boundary and $\overline{\mathcal{Q_B}}$ should be viewed as a 2-manifold-with-boundary. Since $\mathcal{Q_B}\subset \mathbb{R}^2$, $\mathcal{B}$ can be covered by a single chart. In particular, the quotient 1+1 Lorentzian manifold $\mathcal{Q_B}$ is equipped with a double null coordinate chart $(\Tilde{u},\Tilde{v})\in\mathbb{R}^2$ in which the spacetime metric takes the form 
\begin{equation}
    g_{\mathcal{Q_B}} = -e^{2\Tilde{\omega}} d\Tilde{u} d\Tilde{v} 
\end{equation}
The lapse function 
$e^{2\Tilde{\omega}}:\mathcal{Q_B} \to \mathbb{R}^+$ is smooth and the coordinate ranges are such that 
\begin{align}\label{tilde_coordinate_ranges}
     &\Tilde{u}\to-\infty \text{ on } \mathcal{H}_1^+, \  \Tilde{v}\to-\infty \text{ on } \mathcal{H}_2^+\nonumber\\
     &\Tilde{v}\to \infty \text{ on } \mathcal{CH}_1^+,\    \Tilde{u}\to \infty \text{ on } \mathcal{CH}_2^+
\end{align}
Technically, $\Tilde{u}$ and $\Tilde{v}$ would be finite and negative on the portion of $\Sigma_{0,\mathcal{B}}$ away from the past endpoints of $\mathcal{H}_i$ and, if $\mathcal{S}\neq \emptyset$, they would be finite and positive on $\mathcal{S}$ away from the future endpoints of $\mathcal{CH}_i$. However, the precise parameterization of the spacelike components of the boundary of $\mathcal{B}$ in double null coordinates is irrelevant for this paper. The full form of the metric on $\mathcal{B}$ in double null gauge $(\Tilde{u},\Tilde{v},\theta,\varphi)$\footnote{$(\theta, \varphi)$ are the standard spherical coordinates on $S^2$} is given by
\begin{equation}
    g = -e^{2\Tilde{\omega}} d\Tilde{u}d\Tilde{v} + r^2(\Tilde{u},\Tilde{v}) d\Omega^2_2 \label{metric_main_expression}
\end{equation}
By the properties described in the previous subsection, $r$ smooth and positive on $\mathcal{Q_B}$, and admits a continuous extension to $\overline{\mathcal{Q_B}}$ which is positive on $\overline{\mathcal{Q_B}}\backslash\mathcal{S}$. The time orientation in $\mathcal{M}$ is such that $\partial_{\Tilde{u}} + \partial_{\Tilde{v}}$ is future-directed.
 
\subsection{Stability of the Cauchy horizon}
The proof that the an open subset of the development of Einstein-Maxwell initial data studied in Luk-Oh \cite{Luk_Oh} satisfies the Stability Assumptions \ref{assumptions_on_L1_norms} greatly depends on Theorem 5.1 in \cite{Luk_Oh} which controls the metric functions $\Tilde{\omega}$ and $r$, and the scalar field $\phi$ and their derivatives inside a characteristic diamond which is bounded by a portion of $\mathcal{H}_1^+$ to the past and a portion of $\mathcal{CH}_1^+$ to the future. From another point of view, this characteristic diamond is the development of the characteristic Einstein--Maxwell--scalar field initial data supplied on a bifurcate null initial hypersurface which includes a portion of the event horizon - and satisfies the corresponding admissibility conditions satisfied by the Cauchy data. The following theorem due to Luk and Oh is a stability statement: roughly speaking it asserts that, provided the characteristic initial data for the Einstein--Maxwell--scalar field system is close to the initial data for Reissner-Nordstr\"{o}m on a bifurcate null hypersurface involving a portion of the event horizon, the solution remains $C^0$-close to perfect Reissner-Nordstr\"{o}m up to the Cauchy horizon and $C^1$-close in an Eddington-Finkelstein-like coordinate system (to be defined below) to perfect Reissner-Nordstr\"{o}m away from the Cauchy horizon.
For convenience, we state Luk-Oh's theorem on Stability of the Cauchy horizon below. It is important to note that the assumptions of the following theorem \textit{always} hold in the global setting (the 2-ended Cauchy development) as noted in Remark 5.3 of \cite{Luk_Oh}. The original theorem is written in a double null gauge which is not regular near $\mathcal{CH}_1^+$. To stay true to the original form of the theorem, we first define the relations between all coordinate charts that will be used throughout this part. \\[8pt]
\textbf{$({u}, {v})$ double null chart}\\[5pt]
Luk and Oh's work \cite{Luk_Oh} primarily uses an 'Eddington-Finkelstein-like' double null coordinate chart in the precise version of their stability statement. The coordinates $(\Tilde{u}, \Tilde{v})$, defined on $\mathcal{Q_B}$, introduced in Section 2 of \cite{Luk_Oh}\footnote{The notation we use for this, and the following coordinate systems, is not the same as the one used in \cite{Luk_Oh}. In particular, In \cite{Luk_Oh} the 'Eddington-Finkelstein-like' coordinate system is denoted by $(u,v)$. In this paper, what we mean by $(u,v)$ is really an offset by $1$ of the chart $(U_{\mathcal{CH}}, V_{\mathcal{CH}})$ in the notation of \cite{Luk_Oh}.} are unbounded at the event horizons and at the Cauchy horizons (Equation (\ref{tilde_coordinate_ranges})). Throughout the other Sections of this article, we have used a 'Kruskal-like' double null gauge $(u,v)$. On the black hole interior $\mathcal{B}$ in \cite{Luk_Oh}, such a gauge is defined by 
\begin{align}
    u = -\frac{1}{2\kappa_-}e^{-2\kappa_-\Tilde{u}},\ \ \
    v = -\frac{1}{2\kappa_-}e^{-2\kappa_-\Tilde{v}}. \label{gauge_relations}
\end{align}
where $\kappa_\pm := (r_+ - r_-)/r_{\pm}^2$ are the \textit{surface gravities} associated to the background Reissner-Nordstr\"{o}m spacetime on $\mathcal{H}^+$ and $\mathcal{CH}^+$, respectively. Here $r_+ := r_{RN}|_{\mathcal{H}^+}$ is the \textit{outer radius} and $r_- := r_{RN}|_{\mathcal{CH}^+}$ is the \textit{inner radius} associated to the background Reissner-Nordstr\"{o}m spacetime, so they are constant.\\[5pt] Let $\omega$ denote the lapse function in these coordinates. Since $g_\mathcal{Q_B} =- e^{2\omega}dudv = -e^{2\Tilde{\omega}}d\Tilde{u}d\Tilde{v}$, we can relate the lapse function $\Tilde{\omega}$ in $(\Tilde{u},\Tilde{v})$ coordinates to the lapse function $\omega$ in $(u,v)$ coordinates by 
\begin{align}
    &e^{2\Tilde{\omega}} = e^{2\omega}\frac{du}{d\Tilde{u}} \frac{dv}{d\Tilde{v}} = e^{2\omega-2\kappa_- \Tilde{u}-2\kappa_-\Tilde{v}}\implies
    {\omega} = \Tilde{\omega} + \kappa_- (\Tilde{u} + \Tilde{v})
\end{align}
This gauge is particularly useful in that $\omega$ extends continuously to $\mathcal{CH}^+\backslash\mathcal{S}$ (that is, away from the future endpoints of $\mathcal{CH}_i^+$).
\\[8pt]
\textbf{$(u_{\mathcal{H}}, \Tilde{v})$ double null chart} 
\\[5pt]
Another coordinate chart on $\mathcal{Q_B}$, introduced in Section 5 of \cite{Luk_Oh}, and in particular the chart in which the original form of the Stability of the Cauchy horizon theorem is stated, is the double null coordinate chart $(u_{\mathcal{H}}, \Tilde{v})$. Here the outgoing coordinate is related to $v$ by (\ref{gauge_relations}), while the ingoing coordinate $u_\mathcal{H}$ is regular on the right event horizon $\mathcal{H}_1^+$ and defined by:
\begin{align}
    u_\mathcal{H} = \frac{1}{2\kappa_+} e^{2\kappa_+ \Tilde{u}}
\end{align}
hence $u_\mathcal{H}\to 0$ as $\Tilde{u}\to {-\infty}$ i.e. $u_\mathcal{H} = 0$ on $\mathcal{H}_1^+$.\\[5pt]
Let $\omega_\mathcal{H}$ denote the lapse function in these coordinates. Since $g_\mathcal{Q} = - e^{2\omega}dudv = -e^{2\omega_\mathcal{H}}du_\mathcal{H} d\Tilde{v}$, the lapse function $\omega_\mathcal{H}$ in $(u_\mathcal{H}, \Tilde{v})$ coordinates is related to the lapse function $\omega$ in $(u,v)$ coordinates by 
\begin{align}
    e^{2\omega_\mathcal{H}} = e^{2\Tilde{\omega}} \frac{d\Tilde{u}}{du_\mathcal{H}} \implies \omega_\mathcal{H} = \Tilde{\omega} - \kappa_+\Tilde{u}
\end{align} 
\\[5pt]
\textbf{$(\Tilde{u}, v_{\mathcal{CH}})$ double null chart} 
\\[5pt]
The last chart used in Luk-Oh's statement of the Stability of the Cauchy Horizon Theorem (See Section 5 of \cite{Luk_Oh}) is the double null chart $(\Tilde{u}, v_{\mathcal{CH}})$ on $\mathcal{Q_B}$ where the ingoing null coordinate $\Tilde{u}$ is defined in (\ref{gauge_relations}), while the outgoing null coordinate is regular at the right Cauchy horizon and equal to 
\begin{equation}
    v_{\mathcal{CH}} = -\frac{1}{2\kappa_-} e^{-2\kappa_- \Tilde{v}} + 1 = v+1
\end{equation}
Hence $v_{\mathcal{CH}} \to 1$ as $\Tilde{v}\to \infty$ i.e. $v_{\mathcal{CH}}=1$ on $\mathcal{CH}_1^+$.\\[5pt]
In the following, $g_{{RN}}$ denotes the Reissner-Nordstr\"{o}m (RN) metric, $\omega_{{RN}}$ denotes the RN lapse function and $r_{{RN}}$ denotes the RN area radius function. Expressions are given in Appendix \ref{RN_int_geometry}.

\begin{thm}\textbf{(Stability of the Cauchy Horizon, original statement)}\label{stability_of_CH_original} Fix the black hole parameters $q,M$ so that $0<|q|<M$ and a real number $s>1$. Consider the characteristic initial value problem for the Einstein--Maxwell--scalar field system (c.f. Proposition 2.5 in Luk-Oh \cite{Luk_Oh}) with $q\neq 0$ and initial data supplied on the bifurcate null hypersurface $\overline{C_1} \cup C_{-\infty}$ where
\begin{align}
    &\overline{C_{1}} = \{(u_\mathcal{H},\Tilde{v}) : u_\mathcal{H}\leq u_{\mathcal{H},0}, \Tilde{v}=1 \} \nonumber\\
    &{C_{-\infty}} = \{(u_\mathcal{H},\Tilde{v}) : u_\mathcal{H}=0, \Tilde{v}\geq 1 \} \nonumber
\end{align}
for $u_{\mathcal{H},0}>0$. Assume the initial data satisfies the constraint equations in spherical symmetry (\ref{constraint_eqn1} - \ref{constraint_eqn2}), and that the following holds on the initial hypersurface for some $E>0$:
\begin{enumerate}
    \item $C_{-\infty}$ (viewed as the event horizon) is an affine complete null hypersurface approaching a subextremal Reissner-Nordstr\"{o}m event horizon with $0<|q|<M$. More precisely, {in the gauge}\footnote{That such a gauge exists is proved in Appendix A of \cite{Luk_Oh}.}
    \begin{equation}
        e^{2\omega_\mathcal{H}} = \frac{4}{r_+^2}e^{-2\kappa_+ r_+}(r_+-r_-)^{1+\kappa_+/\kappa_-}e^{2\kappa_+ \Tilde{v}}
    \end{equation}
    we have
    \begin{enumerate}
        \item $r\to M + \sqrt{M^2-q^2}$ as $\Tilde{v}\to \infty$
        \item The following decay rate holds for $\Tilde{v}\geq 1$ for the scalar field $\phi$ and its $\Tilde{v}$-derivative:
        \begin{equation}
            |\phi(0, \Tilde{v})| + |\partial_{\Tilde{v}} \phi(0, \Tilde{v})|\leq E\Tilde{v}^{-s}  
        \end{equation}
    \end{enumerate}
    \item On $\overline{C_1}$, after normalizing $u_\mathcal{H}$ by 
    \begin{equation}
        \partial_{u_\mathcal{H}} r = -1
    \end{equation}
    the following holds for all $u_\mathcal{H}\leq u_{\mathcal{H},0}$
    \begin{equation}
        |\partial_{u_\mathcal{H}} \phi (u_\mathcal{H}, 1)| \leq E
    \end{equation}
\end{enumerate}
Then, restricting to some nonempty subset $\overline{C_1'} := \{(u_\mathcal{H},\Tilde{v}) \in \overline{\mathcal{Q}}: u_\mathcal{H}\leq u_{\mathcal{H},s}, \Tilde{v}=1 \} \subset \overline{C_1}$ where $u_{\mathcal{H},s}<u_{\mathcal{H},0}$ is sufficiently small depending on the decay rate $s$, the globally hyperbolic future development of the initial data has a Penrose diagram given by {Figure \ref{CH-Stability_illustration}}. 
Also, in the $(\Tilde{u}, v_{\mathcal{CH}}$) we can attach the null boundary $\{v_{\mathcal{CH}}=1\}=\mathcal{CH}^+$ to the spacetime $(\mathcal{M}, g)$ so that the metric extends continuously to $\mathcal{CH}^+$. \\[5pt]
Furthermore, in the $(\Tilde{u},\Tilde{v})$ coordinate system the following estimates hold in $D^+(\overline{C_1'} \cup C_{-\infty})$ for some constant $C=C(M,q,E,s)>0$:
\begin{align}
    |\phi(\Tilde{u},\Tilde{v})| + |(r-r_{RN})(\Tilde{u},\Tilde{v})| + |(\Tilde{\omega}-\Tilde{\omega}_{RN})(\Tilde{u},\Tilde{v})| 
     \leq C\big[ \Tilde{v}^{-s} + |\Tilde{u}|^{-s+1} \big]\label{stability_0}
 \end{align}\vspace{-2em} 
 \begin{align}   
    |\phi_{,\Tilde{v}}(\Tilde{u},\Tilde{v})| + |(r-r_{RN})_{,\Tilde{
    v}}(\Tilde{u},\Tilde{v})| + |(\Tilde{\omega}-\Tilde{\omega}_{RN})_{,\Tilde{v}}(\Tilde{u},\Tilde{v})| 
    \leq C\Tilde{v}^{-s}
\end{align}
Moreover, for every $A\in\mathbb{R}$ there exists a $C>0$, depending on $M,q,E, s$ {and $u_{\mathcal{H},0}$}, such that for all $(\Tilde{u},\Tilde{v})$ with $\Tilde{u}<\Tilde{u}_s$ and $\Tilde{v}>1$ \footnote{The conditions  $\Tilde{u}<\Tilde{u}_s$ and $\Tilde{v}>1$ are equivalent to the condition that a point with coordinates $(\Tilde{u},\Tilde{v},\theta, \varphi)$ lies in $D^+(\overline{C_1'} \cup C_{-\infty})$} the following decay rate holds:
\begin{align}
    |\phi_{,\Tilde{u} }(\Tilde{u},\Tilde{v})| + |(r-r_{RN})_{,\Tilde{u}}(\Tilde{u},\Tilde{v})| + |(\Tilde{\omega}-\Tilde{\omega}_{RN})_{,\Tilde{u}}(\Tilde{u},\Tilde{v})| 
    &\leq C |\Tilde{u}|^{-s} \text{  for  } \Tilde{u} + \Tilde{v} \geq A \  ;\label{stability_v} \\
    &\leq C e^{2\Tilde{\omega}_{RN}}\Tilde{v}^{-s} \text{  for  } \Tilde{u}+ \Tilde{v} \leq A \label{stability_u}
\end{align} 
\label{CauchyStabilityTheorem}\end{thm}
See Figure \ref{CH-Stability_illustration} (left) for a more detailed illustration of Theorem \ref{stability_of_CH_original}. \\[5pt]
The following corollary restates the results of Theorem \ref{stability_of_CH_original} (Theorem 5.1 in \cite{Luk_Oh}) in the $(u,v)$ coordinates which are regular in the vicinity of the Cauchy horizon.
\begin{cor}\textbf{(Stability of the Cauchy Horizon in $(u,v)$ coordinates)}\label{stability_of_CH} 
Under the hypotheses of Theorem \ref{stability_of_CH_original}, 
 the following estimates hold in $\mathcal{R}$ for some constant $C=C(M,q,E,s)>0$, written in the Kruskal-like double null coordinate system $(u,v)$:
\begin{align}
    |\phi(u,v)| + |(r-r_{RN})(u,v)| + |(\omega-\omega_{RN})(u,v)| 
     \leq C\big[ |\ln^{-s}(-2\kappa_-v)| + |\ln^{-s+1}(-2\kappa_- u)| \big]\label{cauchy_stability_eqn_1}
 \end{align}\vspace{-2em} 
 \begin{align}   
    |\phi_{,v}(u,v)| + |(r-r_{RN})_{,v}(u,v)| + |(\omega-\omega_{RN})_{,v}(u,v)| 
    \leq C\bigg|\frac{\ln^{-s}(-2\kappa_- v)}{v}\bigg| \label{stability_v_new}
\end{align}
Moreover, there exists a $C>0$, depending on $M,q,E, s$ \underline{and $u_0:=u(u_{\mathcal{H},0})$}, such that for all $(u,v)$ with $u<u_s$ and $v>v_{\kappa}$ \footnote{The condition  $u<u_s$ and $v>v_{\kappa}$ is equivalent to the condition that a point with coordinates $(u,v,\theta, \varphi)$ lies in $D^+(\overline{C_1'} \cup C_{-\infty})$} and $uv \leq u_s v_\kappa$ the following decay rate holds:
\begin{align}
    |\phi_{,u }(u,v)| + |(r-r_{RN})_{,u}(u,v)| + |(\omega-\omega_{RN})_{,u}(u,v)| 
    \leq C \bigg|\frac{\ln^{-s}(-2\kappa_- u)}{u} \bigg| \label{stability_u_new}
\end{align} 
\begin{proof} Transform the estimates (\ref{stability_0}-\ref{stability_v}) into the $(u,v)$ coordinate system using the relation between the two coordinate systems (\ref{gauge_relations}) in order to obtain equations (\ref{cauchy_stability_eqn_1} - \ref{stability_v_new}). Choosing $A$ such that $\Tilde{u} + \Tilde{v}\geq A \iff uv <u_s v_\kappa$, we get that there exists a $C=C(M,q,E,s, u_0)$ such that $(\ref{stability_u})$ holds. Writing this inequality in the $(u,v)$ coordinate system gives (\ref{stability_u_new}).
\end{proof}\end{cor}

For the remainder of this section we shall work under the hypotheses of Theorem \ref{stability_of_CH_original} (equivalently Corollary \ref{stability_of_CH} in $(u,v)$ coordinates). Ultimately, this will help us establish that Assumptions \ref{assumptions_on_L1_norms} and item 3 in Assumptions \ref{instability_assumptions} hold on a submanifold of $(\mathcal{M}, g)$ that we will specify below. 

\subsubsection{Restricting to a region near $\mathcal{CH}^+$} In this section we restrict to the submanifold of the black hole interior arising from the   Einstein--Maxwell--scalar field system where the Assumptions listed in Section  \ref{preliminaries} are satisfied. This is the codimension-0 submanifold $(\mathcal{R}, g)$ defined below and illustrated in Figure \ref{restriction_to_R}. This in itself is a subset of the interior of the black hole where Stability of the Cauchy horizon, Corollary \ref{stability_of_CH}, holds. As is standard in spherically symmetric problems, the focus is on the (1+1)-dimensional Lorentzian manifold-with-boundary $(\overline{\mathcal{Q_R}} = \overline{\mathcal{R}}/SO(3), g_{\mathcal{Q}})$.\\[5pt]
In view of the gauge relations (\ref{gauge_relations}), the bifurcate null hypersurface $C_{-\infty} \cup \overline{C_1'}$ defined in Theorem \ref{stability_of_CH_original} in the $(u,v)$ coordinate system looks like
\begin{align}
    &\overline{C_{1}'} = \{(u,v) \in \overline{\mathcal{Q_B}}: u\leq u_s, v=v_{\kappa} \} \nonumber\\
    &{C_{-\infty}} = \{(u,v) \in \overline{\mathcal{Q_B}} : u=-\infty, v>v_{\kappa} \} \nonumber
\end{align}
where $v_\kappa := -e^{-2\kappa_-}/(2\kappa_-)$ (the $v$ associated with $\Tilde{v} = 1$) and $u_s := -e^{-2\kappa_- \Tilde{u}_s}/(2\kappa_-)$ (the $u$ associated with $\Tilde{u}_s$). This explains the choice of notation in the definition of the spacetime $(\mathcal{R},g)$ in Section \ref{preliminaries}. 
\begin{defn}
Let $s, u_s, \overline{C_1'}, C_{-\infty}$ be defined as in Theorem \ref{stability_of_CH_original} and let $\kappa_{\pm}$ be the surface gravities at the outer/inner horizons of the Reissner-Nordstr\"{o}m black hole with parameters $0<|q|<M$. In the double null coordinates $(x^\mu)=(u,v, \theta, \varphi)$ which are regular at the right Cauchy horizon\footnote{More precisely, such that $\mathcal{CH}_1^+ = \{v=0 \}$.}, let $\Sigma$ be the spacelike hypersurface in $\mathcal{B}$ defined by the condition $uv=u_s v_\kappa$. Let $(\mathcal{R},g)$ be the open subset of $(\mathcal{B}, g)$, given by\footnote{Given a set $S$, $D^+(S)$ denotes the future domain of dependence of $S$. }
\begin{align}
    &\mathcal{R} = D^+(\overline{C_1}' \cup C_{-\infty}) \cap J^+(\Sigma) = 
    \{(x^\mu) \in \mathcal{B} : -\infty<u\leq u_s, v_\kappa<v<0, uv\leq u_s v_\kappa \}\nonumber\\[5pt]
    &g=-e^{2\omega(u,v)} dudv + r^2 d\Omega_2^2; \label{restricted_submanifold_R}
\end{align}
for  $(u,v)\in \mathcal{Q_R} := \mathcal{R}/SO(3)$. Denote by $\overline{\mathcal{Q_R}}$ the 2-manifold-with-boundary 
\begin{align}
   \overline{\mathcal{Q_R}}= \mathcal{Q_R}\cup [-\infty, u_s]\times\{v=0\}\cup \{u=u_s\}\times[v_\kappa, 0]\cup \{(u,v):u\leq u_s, v_\kappa\leq v, uv\leq u_s v_\kappa \}
\end{align}
and by $\overline{\mathcal{R}}:=\overline{\mathcal{Q_R}}\times S^2$ the 3-manifold-with-boundary associated with $\overline{\mathcal{Q_R}}$.
\end{defn}
The region $\mathcal{R}$ is illustrated in yellow in the Penrose diagram on the right in Figure \ref{CH-Stability_illustration}. Restricting to the region $\mathcal{R}$, we extract the estimates of Theorem \ref{stability_of_CH_original} in the Kruskal-like double null gauge $(u,v)$, which is regular near $\mathcal{CH}_1^+$.
\begin{figure}[h]
    \centering
\begin{subfigure}{0.48\textwidth}
    \includegraphics[width=0.9\textwidth]{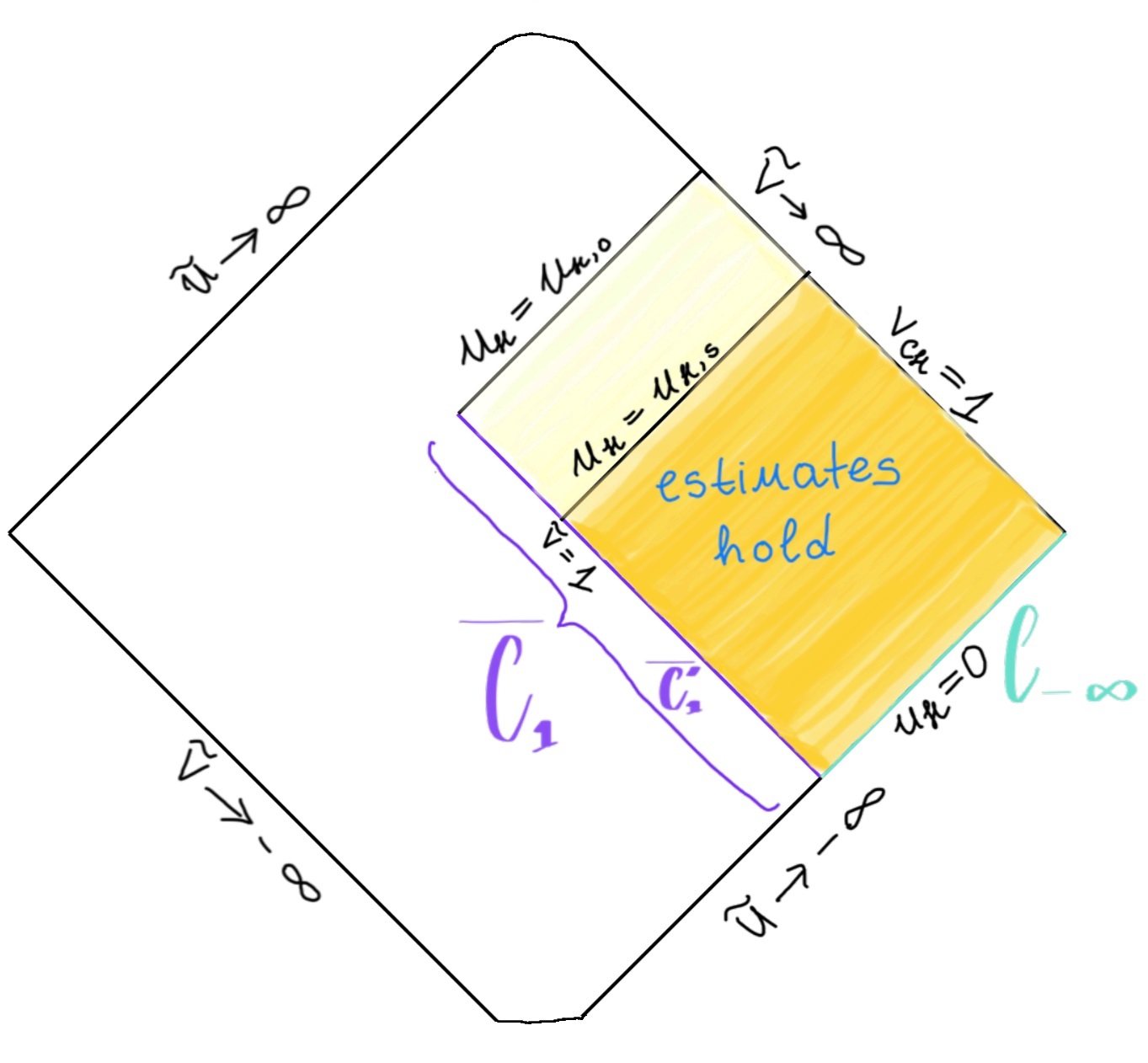} 
    \label{fig:a}
\end{subfigure}
\hfill
\begin{subfigure}{0.48\textwidth}
    \includegraphics[width=0.8\textwidth]{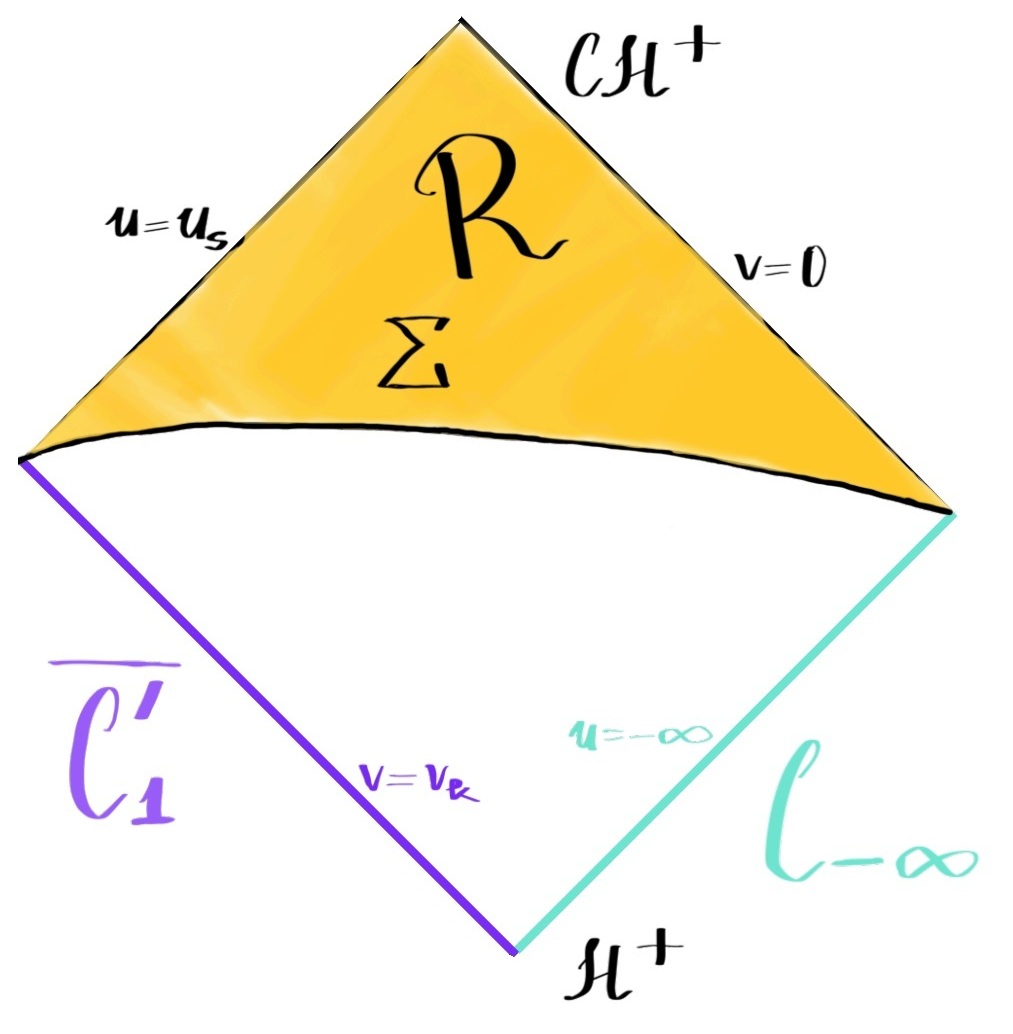}
    \label{fig:b}
\end{subfigure}
     \caption{Left: Illustration of the original statement of Stability of the Cauchy horizon for Einstein--Maxwell--scalar field initial data on a bifurcate null hypersurface, satisfying certain decay conditions. The stability estimates hold in the yellow region with blue letters. Right: Penrose diagram illustrating the open subset $\mathcal{R}\subset \mathcal{M}$, shown in yellow. We restrict to the future of the spacelike hypersurface $\Sigma =\{uv=u_{s}v_{\kappa}\}$, intersected with the domain of dependence of the bifurcate null hypersurface $\overline{C'_{1}}\cup C_{-\infty}$. On $\mathcal{R}$, the statement of Corollary {\ref{stability_of_CH}} holds. The past boundary of $\mathcal{R}$ is $\Sigma$ and the future boundary of $\mathcal{R}$ is $\{u=u_s\}\cup\{v=0\}\cap \overline{\mathcal{R}}$.}
    \label{CH-Stability_illustration}
\end{figure}
 
    {In the rest of this subsection we show that Assumptions \ref{assumptions_on_L1_norms} are satisfied in $\mathcal{R}$. We begin with the simple observation that $\norm{r^{\pm 1}}_{C(\mathcal{R})}$ is bounded by a constant depending on the geometry in $(\mathcal{R},g)$, which is encapsulated in the following lemma:

    \begin{lem}\label{bounds_on_r_in_R} There exist constants $0<r_0<r_1$ such that $r_0\leq r\leq r_1$ in $\mathcal{R}$. Furthermore, $r_1$ can be chosen to depend only on $(M,q,E,s,u_0)$.   
        \begin{proof}
            According to the previous subsection (but more precisely Theorem 4.1 in \cite{Luk_Oh}), the area radius function $r$ is smooth and positive in $\mathcal{Q_B}$ and extends continuously as a positive function on  $\mathcal{CH}_i^+$ away from $\mathcal{S}$ and $i^+_{i}$. Therefore, $r$ is smooth and positive on $\mathcal{R}$ and extends continuously as a positive function to $\overline{\mathcal{R}}$. To prove the existence of $r_0$, we will invoke Lemma \ref{choose_domain_lemma} which is stated and proved in the next subsection. It uses the Raychaudhuri equations (\ref{constraint_eqn1}-\ref{constraint_eqn2}) and the admissibility conditions on the Einstein--Maxwell--scalar field Cauchy data from Definition 3.1 of \cite{Luk_Oh} to show the existence of a semi-global trapped characteristic rectangle in $\mathcal{B}$ . We stress that no part of the proof of Lemma \ref{choose_domain_lemma} requires any bounds on $r$ so there is no circular argument. By Lemma \ref{choose_domain_lemma}, there is a $v_*<0$ such that $\partial_u r<0$ in $\mathcal{R}\cap \{v\geq v_*\}$. Let $u_r\in(-\infty,u_s)$. Note that $\overline{\mathcal{R}}\cap(\{u\geq u_r\}\cup\{v\leq v_*\})$ is a compact set containing $\mathcal{U}_r:={\mathcal{R}}\cap(\{u\geq u_r\}\cup\{v\leq v_*\})$ and so $\inf_{{\mathcal{U}_r}}r>0$. Now fix $u<u_r$ and $v>v_*$ such that $(u,v)\in \mathcal{Q_R}\backslash(\mathcal{U}_r/S^2)$ and $(u',v)\in \mathcal{U}_r/ S^2$. Since $\partial_u r<0$ in $\mathcal{R}\cap \{v\geq v_*\}$, we have $r(u,v)\geq r(u',v)\geq \inf_{\mathcal{U}_r}r$. Since the choice of point $(u,v)$ in $\mathcal{R}\backslash\mathcal{U}_r$ was arbitrary, it follows that $r_0:=\inf_{\mathcal{R}}r=\inf_{\mathcal{U}_r}r>0$. 
            \\[5pt]
            On the other hand, by the first estimate in the statement of the Stability of the Cauchy horizon, Corollary \ref{stability_of_CH},
            \begin{align}
                r(u,v)&\leq |(r-r_{RN})(u,v)|+|r_{RN}(u,v)|\leq C\big[|\ln^{-s}(-2\kappa_-v)|+|\ln^{-s+1}(-2\kappa_-u)|\big] + \sup_{R}|r_{RN}| \nonumber\\
                &\leq C\big[(2\kappa_-)^{-s} + |\ln^{-s+1}(-2\kappa_-u_s)| \big]+r_+=:r_1(M,q,E,s,u_0)
            \end{align}
            holds for $(u,v)\in \mathcal{R}$ for some $C=C(M,q,E,s)>0$. In the last line, we also used that $u_s$ depends only on $u_0$ and $s$ and $v_{\kappa}=e^{-2\kappa_-}/(2\kappa_-)$ depends only on $M$ and $q$.
        \end{proof}
    \end{lem}
\begin{rem}
The lower bound $r_0$ clearly depends on the geometry of the submanifold $(\mathcal{R},g)$. We cannot get $r_0$ to depend only on $(M,q,E,s,u_0)$ unless we also assume smallness of the perturbation (smallness of the constant $C$ in equation (\ref{cauchy_stability_eqn_1})). Indeed, if the perturbation is small enough so that the RHS of (\ref{cauchy_stability_eqn_1}) is $\leq r_-+r_0'$ for some $r_0'>0$, then
\begin{align}
    \inf_{\mathcal{R}}r\geq |(r-r_{RN})(u,v)| - \inf_{\mathcal{R}}r_{RN} \leq r_--r_0'-r_-=r_0'
\end{align}
where $r_0'$ depends only on $(M,q,E,s,u_0)$. However, we will not assume such smallness in what follows.

\end{rem}
We will refer to $(M,q,E,s, u_0, r_0)$ as \textit{the global parameters}. These parameters are specific to the region $\mathcal{R}$, hence any quantity which depends only on the global parameters is constant throughout $\mathcal{R}$. We will define the 'big constant $C_g$' depending on the geometry of the spacetime to be the maximum of a finite set of constants $C=C\big(M,q,E,s, u_0,r_0\big)$ which depend only on the global parameters. Next, we show that in the $(u,v)$ coordinates $e^{\pm 2\omega}$ is uniformly bounded on $\mathcal{R}$ by a constant depending only on the global parameters. This is a straightforward consequence of the $L^\infty$ estimate (\ref{stability_0}) in Corollary \ref{stability_of_CH}, which we demonstrate below.}

\begin{cor}\label{unif_bounded_omega}
In this Kruskal-like double null coordinate system $(u,v, \theta, \varphi)$ which is regular at the Cauchy horizon, there exists a $C>0$ depending on $(M,q,E,s,u_0)$ such that the following uniform bound holds as a result of Corollary \ref{stability_of_CH} throughout the region $\mathcal{R}$:
\begin{equation}
\norm{e^{\pm2\omega}}_{C(\mathcal{R})} \leq C
\end{equation}
\begin{proof} 
We first claim that $r_{RN}=const$ on the hypersurface $\Sigma = \{uv=u_s v_\kappa \}$. To prove this claim, use the expression for the RN tortoise coordinate $r_{RN}^*$ from Appendix \ref{RN_int_geometry}. On $\Sigma$, we have
\begin{align}
    &u_s v_\kappa = uv = \frac{1}{4\kappa_-^2} \exp(-2\kappa_- r_{RN}^*) = \frac{1}{4\kappa_-^2}\exp\big(-2\kappa_- r_{RN} - \frac{\kappa_-}{\kappa_+} \ln(r_+ - r_{RN}) + \ln(r_{RN} - r_-)\big) = \nonumber \\
    & =\frac{1}{4\kappa_-^2} e^{-2\kappa_- r_{RN}} (r_+ - r_{RN})^{-\kappa_-/\kappa_+}(r_{RN} - r_-) \label{r_const_on_Sigma}
\end{align}
By (\ref{r_const_on_Sigma}), the tortoise coordinate $r_{RN}^*$ is constant on $\Sigma$ - this fixes a unique constant $r_{RN, \Sigma} \in (r_-, r_+)$ such that $r_{RN}(u,v) = r_{RN, \Sigma}$ for all $(u,v) \in \Sigma$. \\[5pt]
Next we claim that if $\Sigma'$ is a spacelike hypersurface in $\mathcal{R}$ defined by the properties that $uv=const$ on $\Sigma'$ and $\Sigma' \subset J^+(\Sigma)$, then $r_{RN}(u,v) \equiv r_{RN, \Sigma'}=const$ for all $(u,v)\in\Sigma'$ and $r_{RN, \Sigma'}<r_{RN, \Sigma}$. Let's prove this: By an identical argument as in (\ref{r_const_on_Sigma}), the RN area radius function is constant on any other spacelike hypersurface $\Sigma'$ defined by the property $uv=const$. Now since $\partial_u + \partial_v$ is future-directed, $u$ and $v$ are non-decreasing along future-directed causal curves, so $\Sigma' \subset J^+(\Sigma) \implies uv|_{\Sigma'} < uv|_{\Sigma} = u_s v_\kappa$. This means that $r_{RN}^*|_{\Sigma'}<r_{RN}^*|_{\Sigma}$, hence $r_{RN, \Sigma'}<r_{RN, \Sigma}$. \\[5pt]
Note that the submanifold $\mathcal{R}$ can be foliated by such surfaces:
\begin{equation}
    \mathcal{R} = D^+(\Sigma)\cap\bigcup_{\infty<u<u_s, v_\kappa<v<0} \Sigma'
\end{equation}
Therefore in $\mathcal{R}$ we have $r_{RN, \Sigma}>r_{RN}(u,v)>r_-$.
Referring to Appendix \ref{RN_int_geometry} or Section 5.1 from \cite{Luk_Oh} for the geometry of the Reissner-Nordstr\"{o}m black hole, one finds that in the $(u,v)$ coordinate system which is regular at the Cauchy horizon, one can express the Reissner-Nordstr\"{o}m lapse function in $\mathcal{R}$ as:
\begin{align}
    &e^{2\omega_{RN}} = \frac{-4\frac{r_{RN}^2 - 2M r_{RN} + q^2}{r_{RN}^2}}{4\kappa_-^2 uv} =  \frac{4}{r_{RN}^2} e^{2\kappa_-r_{RN}}(r_+ - r_{RN})^{1+\kappa_-/\kappa_+} \implies \nonumber\\
    &\omega_{RN} = \frac{1}{2} \ln \bigg[ \frac{4}{r_{RN}^2} e^{2\kappa_-r_{RN}}(r_+ - r_{RN})^{1+\kappa_-/\kappa_+} \bigg] \label{omega_RN_in_terms_of_r}
\end{align}
Hence 
\begin{align}
    &\norm{\omega_{RN}}_{C(\mathcal{R})} = \frac{1}{2} \sup_{(u,v)\in \mathcal{Q}_\mathcal{R}} \bigg| \ln \bigg[ \frac{4}{r_{RN}^2} e^{2\kappa_-r_{RN}}(r_+ - r_{RN})^{1+\kappa_-/\kappa_+} \bigg] \bigg|\leq\nonumber\\
    &\leq\max \Bigg\{  \bigg|\ln\bigg( \frac{4}{r_{-}^2} e^{2\kappa_-r_{RN, \Sigma}}(r_+ - r_{RN, \Sigma})^{1+\kappa_-/\kappa_+}\bigg)\bigg|, \bigg|\ln\bigg(\frac{4}{r_{RN, \Sigma}^2} e^{2\kappa_-r_{-}}(r_+ - r_{-})^{1+\kappa_-/\kappa_+}\bigg) \bigg| \Bigg\} \leq C(M,q,E,s, u_0)
\end{align}
where the dependence of $r_{RN, \Sigma}$ on the global parameters $(M,q,E,s)$ is through $u_s$ and $v_\kappa$. Then by the estimate (\ref{cauchy_stability_eqn_1}) from Stability of The Cauchy Horizon, Corollary \ref{stability_of_CH}, 
\begin{align}
   &\norm{\omega}_{C(\mathcal{R})} \leq \norm{\omega - \omega_{RN}}_{C(\mathcal{R})} +\norm{\omega_{RN}}_{C(\mathcal{R})} \leq\nonumber\\
    &\leq 
    \sup_{(u,v)\in \mathcal{Q}_\mathcal{R}} C\big(|\ln^{-s}(-2\kappa_- v)| + |\ln^{-s+1}(-2\kappa_- u)|\big) + \norm{\omega_{RN}}_{C(\mathcal{R})} \leq C(M,q,E,s,u_0) \label{omegabound_forCg}
\end{align}
It follows that $\norm{e^{2\omega}}_{C(\mathcal{R})}$ and $\norm{e^{-2\omega}}_{C(\mathcal{R})}$ are both bounded by a constant which depends only on the global parameters. 
\end{proof}
\end{cor} 

\begin{lem}\label{unif_bounded_u_derivatives}
    Given $v\in (v_\kappa, 0)$, let $u_\Sigma(v) = u_s v_\kappa/v$ i.e. $u_\Sigma$ is the ingoing null coordinate such that the point $(u_\Sigma(v), v)$ lies on $\Sigma$. We have
    \begin{align}
        &\sup_{u\in [u_\Sigma(v), u_s]}|\partial_u \phi(u,v)| \leq C(M,q,E,s,u_0) \\
        &\sup_{u \in [u_\Sigma(v), u_s]}|\partial_u r(u,v)| \leq C(M,q,E,s,u_0) \label{bound_on_dur}
    \end{align}
    \begin{rem}
        In particular, equation \ref{bound_on_dur} implies that $\sup_{(u,v)\in \mathcal{R}} |\partial_u r(u,v)|< \infty$, proving that item 3 of Assumptions \ref{instability_assumptions} holds in $(\mathcal{R}, g)$. 
    \end{rem}
    \begin{proof} By Stability of the Cauchy Horizon Theorem \ref{stability_of_CH},
     \begin{align}
        &\sup_{u\in [u_\Sigma(v), u_s]}|\partial_u \phi(u,v)| + \sup_{u \in [u_\Sigma(v), u_s]}|\partial_u (r-r_{RN})(u,v)| \leq \nonumber\\
        &\leq \sup_{u\in[u_\Sigma(v), u_s]} C\bigg| \frac{\ln^{-s} (-2\kappa_- u)}{u}\bigg|=  C\bigg| \frac{\ln^{-s}(-2\kappa_- u_s)}{u_s} \bigg| \sim C(M,q,E,s, u_0)
    \end{align}
        and by Appendix \ref{RN_int_geometry}, equation (\ref{express_durrn_appendix}) we have 
        \begin{align}
            \sup_{u\in[u_\Sigma(v), u_s]}|\partial_u r_{RN}| &= \sup_{u\in[u_\Sigma(v), u_s]} \bigg|  \frac{1}{2\kappa_- u} \bigg(\ \frac{r_{RN}^2 - 2M r_{RN} + q^2}{r_{RN}^2} \bigg) \bigg| \leq\nonumber\\
            &\leq\bigg|\frac{1}{2\kappa_- u_s r_-^2}(r_{RN, \Sigma}^2 - 2M r_{RN, \Sigma} + q^2)\bigg| \sim C(M,q,E,s,u_0)
        \end{align}
        where $r_{RN,\Sigma}$ is introduced in Corollary \ref{unif_bounded_omega}.
    \end{proof}
\end{lem}

\begin{lem}\label{finite_omegaRN_integral} Given $v\in (v_\kappa, 0)$, let $u_\Sigma(v) = u_s v_\kappa/v$ i.e. $u_\Sigma$ is the ingoing null coordinate such that the point $(u_\Sigma(v), v)$ lies on $\Sigma$. We have 
\begin{align}
    \sup_{v_\kappa<v<0} \int_{u_\Sigma(v)}^{u_s} |\omega_{RN,u}(u,v)|du \leq C(M,q,E,s,u_0)
\end{align}
\begin{proof} Firstly, note that the point $(u,v)$ is in $\mathcal{Q}_\mathcal{R}$ for any $u\in(u_\Sigma, u_s)$. Let us refer to Appendix \ref{RN_int_geometry} or Section 5.1 from Luk-Oh \cite{Luk_Oh} to derive an explicit expression for $\omega_{RN}$ in terms of the RN area radius function $r_{RN}$, the Kruskal-like coordinates $(u,v)$ and the global parameters $(M,q,E,s)$. One has 
\begin{align}
    &\omega_{RN,u} = \frac{1}{2}\partial_u \big(\ln(e^{2\omega_{RN}})\big) = \frac{1}{2}e^{-2\omega_{RN}} \frac{\partial}{\partial u} \bigg(  -\frac{1}{\kappa_-^2 u v} \bigg( 1 - \frac{2M}{r_{RN}} + \frac{q^2}{r_{RN}^2}\bigg) \bigg) = \nonumber\\
    &=\frac{1}{2} e^{-2\omega_{RN}} \bigg[ \frac{1}{\kappa_-^2 u^2 v} \bigg( 1 - \frac{2M}{r_{RN}} + \frac{q^2}{r_{RN}^2}\bigg) - \frac{2}{\kappa_-^2 uv} \bigg( \frac{M}{r_{RN}^2} - \frac{q^2}{r_{RN}^3} \bigg)\partial_u r_{RN}\bigg]  \label{express_duwrn}
\end{align}
Referring to Appendix \ref{RN_int_geometry}, 
\begin{align}
    \partial_u r_{RN} = -\frac{1}{2\kappa_- u} \bigg( 1 - \frac{2M}{r_{RN}} + \frac{q^2}{r_{RN}^2}\bigg) \label{express_durrn}
\end{align}
Rearranging equations (\ref{express_duwrn}) and (\ref{express_durrn}) and employing (\ref{express_OmegaRN_appendix}) yields
\begin{align}
    \omega_{RN,u} = -\frac{1}{2u}\bigg( 1+ \frac{1}{\kappa_-} \bigg( \frac{M}{r_{RN}^2} - \frac{q^2}{r_{RN}^3} \bigg) \bigg) \label{duwrn}
\end{align}
Using that $\kappa_- = 2\sqrt{M^2 - q^2}/(M- \sqrt{M^2-q^2})^2$, one confirms that the RHS of (\ref{duwrn}) is nonnegative for $r_-\leq r_{RN}\leq r_+$ and zero at $r_{RN}=r_-$, so $\omega_{RN,u}>0$ in $\mathcal{R}$. Thus, given $v\in (v_\kappa, 0)$ we can apply the fundamental theorem of calculus:
\begin{align}
\norm{\omega_{RN,u}}_{L^1[u_\Sigma(v), u_s]} = \int_{u_s v_\kappa/v}^{u_s} \partial_u \omega_{RN}(u,v) du = \omega_{RN}(u, v)\big|^{u=u_s}_{u=u_s v_\kappa /v}
\end{align}
This is obviously finite for all $v<0$. By the proof of Corollary \ref{unif_bounded_omega}, $r_{RN}=r_{RN, \Sigma}$ is constant on $\Sigma$ and since $(u_\Sigma(v),v)\in \Sigma$ for every $v_\kappa<v<0$, by expression (\ref{omega_RN_in_terms_of_r}) $\omega_{RN}$ is constant on $\Sigma$. In particular 
\begin{equation}
    \lim_{v\to 0} \omega_{RN}(u_s v_\kappa/v, v) = \omega_{RN} (u_s, v_\kappa)<\infty
\end{equation}
\end{proof}
    
\end{lem}

\subsubsection{$L^1$ bounds on $\omega_{,v}$, $\omega_{,u}$, $r_{,v}$, $r_{,u}$, $\phi_{,v}$, $\phi_{,u}$.}
In this section, we establish the $L^1$ estimates of the null derivatives of geometric quantities along $u$-  and $v$-intervals, hence proving that Assumptions \ref{assumptions_on_L1_norms} needed for the dust regularity are satisfied in the submanifold $\mathcal{R}$ of the Einstein--Maxwell--scalar field black hole interior.  Recall that these were crucial for establishing the proof of Theorem \ref{master_theorem}. Like most of the prerequisites for the proof of dust regularity, they follow from Stability of The Cauchy Horizon, Corollary \ref{stability_of_CH} and hold in this form specifically in the region $\mathcal{R}$. 
\begin{prop} \label{stability_in_w11}Let $ (\mathcal{Q}_\mathcal{R},g_{\mathcal{Q}})$ be as above.
 Then for every $u_1<u_2$ and $v_1<v_2$ such that  $(u_i, v_j)\in \mathcal{Q}_\mathcal{R}$, the following $L^1$ bounds in the $(u,v)$ double null coordinate system hold as a consequence of Corollary \ref{stability_of_CH}:
\begin{align}
    \norm{\omega_{,u} - \omega_{RN,u}}_{L^1[u_1, u_2]} + \norm{r_{,u} - r_{RN,u}}_{L^1[u_1, u_2]} + \norm{\phi_{,u}}_{L^1[u_1, u_2]} \leq C \label{u_derivative_boundss}\\
    \norm{\omega_{,v} - \omega_{RN,v}}_{L^1[v_1,v_2]} + \norm{r_{,v} - r_{RN,v}}_{L^1[v_1, v_2]} + \norm{\phi_{,v}}_{L^1[v_1, v_2]} \leq C' \label{v_derivative_boundss}
\end{align}
For  $C, C'>0$, depending only on the global parameters $(M,q, E, s,u_0)$.
\begin{proof}
   Assume $(u_i, v_j)\in \mathcal{Q}_\mathcal{R}$ for $i,j\in\{1,2\}$. By construction, the hypotheses of the Theorem on Stability of the Cauchy Horizon (Theorem 5.1 in \cite{Luk_Oh}, Corollary \ref{stability_of_CH} above)  are satisfied in $\mathcal{R}$. Rewriting the estimates for the partial derivatives of $\omega, r$ and $\phi$ in the $(u,v)$ double null coordinate system which is regular at $\mathcal{CH}_1^+$ and integrating,
    \begin{align}
         &\norm{\omega_{,u} - \omega_{RN,u}}_{L^1[u_1, u_2]} + \norm{r_{,u} - r_{RN,u}}_{L^1[u_1, u_2]} + \norm{\phi_{,u}}_{L^1[u_1, u_2]} \leq \nonumber \\
         &\leq C\int_{u_1}^{u_2} \frac{1}{|-2\kappa_- u|} \frac{|\ln^{-s}(-2\kappa_- u)|}{|-2\kappa_- |^{-s}} du \leq C(2\kappa_-)^s \int_{-\infty}^{u_s} \frac{|\ln^{-s}(-2\kappa_- u)|}{|-2\kappa_-u|} du = C (2\kappa_-)^{s-1}\frac{ |\ln^{1-s}{(-2\kappa_- u_s)}|}{s-1} \label{blaa_u}
    \end{align}
    for some $C=C(M,q,E,s,u_0)>0$. We used that the decay rate $s$, defined in Corollary \ref{stability_of_CH} is bigger than $1$. Since the Reissner-Nordstr\"{o}m surface gravity on the inner horizon $\kappa_-$ depends only on $M$ and $q$, and since $u_s$ depends only on the decay rate $s$, we have proved (\ref{u_derivative_boundss}). \\[5pt]
    In the same way we establish (\ref{v_derivative_boundss}). Recall $v_\kappa = -e^{-2\kappa_-}/(2\kappa_-)$. Rewriting the estimates in the $(u,v)$ coordinate system and integrating over $v$,
   \begin{align}
        &\norm{\omega_{,v} - \omega_{RN,v}}_{L^1[v_1,v_2]} + \norm{r_{,v} - r_{RN,v}}_{L^1[v_1, v_2]} + \norm{\phi_{,v}}_{L^1[v_1, v_2]} \leq \nonumber\\
        &\leq C' \int_{v_1}^{v_2} \frac{1}{|-2\kappa_- v|} \frac{\ln^{-s}(-2\kappa_- v)}{(-2\kappa_- )^{-s}} dv = C' (-2\kappa_-)^s \int_{v_\kappa}^0 \frac{\ln^{-s}(-2\kappa_- v)}{-2\kappa_{-}v}dv = \nonumber \\
        &=C'\frac{(-2\kappa_-)^{s-1}}{2\kappa_-} \frac{ \ln(e^{-2\kappa_-})}{1-s} = C'\frac{1}{(2\kappa_- (s-1))}
   \end{align}
   which depends only on the global parameters $(M,q, E, s, u_0)$.
\end{proof}
\end{prop}
   
   \subsubsection{$L^1$ bound on $\partial_u \partial_v \omega$.}\label{L1_bds_proof2}
   {The next lemma establishes that the assumption (\ref{geom_bound_as4}) is also satisfied by the Einstein--Maxwell--scalar field black hole interior. Recall that this was important for establishing an $L^1$ estimate for the RHS of the Jacobi equation in Section \ref{regularity_of_timelike_geodesic_variations} and by extension, for the proof of Theorem \ref{master_theorem}. To prove that the assumption holds for the region $\mathcal{R}\subset\mathcal{M}$, we will need to consult the evolution equation for the lapse function in the spherically symmetric Einstein-Maxwell-real scalar field system, (\ref{evol_eqn_omega_appendix}).} Rewritten below for convenience,
\begin{align}
    \omega_{,uv} = - \phi_{,u}\phi_{,v} - \frac{1}{2}e^{2\omega}r^{-2}\bigg(\frac{{q}^2} {r^{2}} - \frac{1}{2}\bigg) + (\ln{r})_{,u}(\ln{r})_{,v} \label{express_ddomega}
\end{align}

    \begin{lem} \label{estimate_dw_dudv}For any $u$ and any $v_1< v_2$ such that $(u, v_1), (u, v_2 )\in \mathcal{Q_R}$, we have
     \begin{equation}
         \norm{\omega_{,uv}(u, v)}_{L^1[v_1, v_2]}\leq {C\big(v_2-v_1 + 1\big)}
     \end{equation}
     Where $C$ depends only on the global parameters $(M,q,E,s,u_0, r_0)$.
\begin{proof}
     Integrating equation (\ref{express_ddomega}) between $v_1$ and $v_2$ and applying Lemma \ref{bounds_on_r_in_R}, Proposition \ref{stability_in_w11} and Lemma \ref{unif_bounded_u_derivatives} we have
\begin{align}
     \norm{\omega_{,uv}}_{L^1[v_1, v_2]} &\leq \norm{\phi_{,u}\phi_{,v}}_{L^1[v_1, v_2]} + \frac{ \norm{e^{2\omega}}_{C(\mathcal{R})}|v_2 - v_1|}{2}\bigg( \frac{q^2}{r_0^4} + \frac{1}{r_0^2}\bigg) + \frac{1}{r_0^2}\norm{r_{,u}r_{,v}}_{L^1[v_1, v_2]} \leq \nonumber\\&\leq\norm{\phi_{,u}}_{C(\mathcal{R})}\norm{\phi_{,v}}_{L^1[v_1, v_2]} + \frac{ \norm{e^{2\omega}}_{C(\mathcal{R})}|v_2 - v_1|}{2}\bigg( \frac{q^2}{r_0^4} + \frac{1}{r_0^2}\bigg) + \frac{1}{r_0^2}\norm{r_{,u}}_{C(\mathcal{R})} \norm{r_{,v}}_{L^1[v_1, v_2]} \leq \nonumber\\
     &\leq C(1+v_2-v_1) \label{dwdudv_bound}
\end{align}
Where $C=C(M,q, E, s, u_0, r_0)$. We conclude that the $L^1$ norm of the factor $\omega_{,uv}$ along any $v$-interval in $\mathcal{Q_R}$ is bounded by a constant depending on the global parameters, multiplied by the length $|v_2 - v_1|$ of the interval.
\end{proof}
\begin{rem}
    As $|v_\kappa| \geq |v_2 - v_1|$ and in view of the fact that $v_\kappa$ is defined in terms of the global parameters $M$ and $q$ of the black hole interior, the estimate can be made uniform:
    \begin{equation}
        \norm{\omega_{,uv}}_{L^1[v_1, v_2]} \leq C(1+|v_\kappa|) \leq C(M,q,E,s, u_0,r_0)\label{misc}
    \end{equation}
\end{rem}
\end{lem}
\textbf{Choosing the geometric constant $C_g$}
\\[5pt]
The constant $C_g$ for the subset $(\mathcal{R}, g)$ of the  Einstein--Maxwell--scalar field black hole interior can then be chosen to be the maximum between the constants $C(M,q,E,s, u_0,r_0)$ in (\ref{u_derivative_boundss} - \ref{v_derivative_boundss}), (\ref{misc}) and the exponential of the constant $C(M,q,E,s,u_0)$ in (\ref{omegabound_forCg}) plus $r_0^{-1}+r_1$.
\\[5pt]
\textbf{The Assumption $\norm{\partial_u r}_{C(\mathcal{R})}<\infty$}
\\[5pt] 
The third item in Assumptions \ref{instability_assumptions} is actually a direct consequence of Lemma \ref{unif_bounded_u_derivatives} which is itself a consequence of Stability of the Cauchy Horizon - Theorem 5.1 in \cite{Luk_Oh}.

 \subsection{Instability of the Cauchy Horizon Theorem}
\textbf{The Assumption $\partial_v r(u,v)\to -\infty$ as $v\to -\infty$.}\\[5pt] 
The first line of Assumptions \ref{instability_assumptions} requires that the area radius function $r_{,v}$ goes to $-\infty$ as $v\to 0$ in double null coordinates which are regular at the Cauchy horizon. This is satisfied by the Einstein--Maxwell--scalar field black hole interior according to the Instability of the Cauchy Horizon - Theorem 5.5 in Luk-Oh, \cite{Luk_Oh}. In particular, equation (5.19) in \cite{Luk_Oh}. \\[5pt]
The second condition in Assumptions \ref{instability_assumptions} asserts the existence of a trapped causal diamond in $\mathcal{R}$ which borders the Cauchy horizon. More precisely, the condition asserts that for every $u_0\in(-\infty,u_s)$ a local trapped characteristic diamond $[u_0-\varepsilon, u_0+\varepsilon]\times[v_0,0)$ exists for some $v_0<0$ and $\varepsilon>0$, both depending on $u_0$. Before we prove the above result, we show the existence of a semi-global region in $\mathcal{B}$ by applying the admissibility conditions on the initial data assumed in Definition 3.1 in \cite{Luk_Oh}. 
\begin{lem}\label{choose_domain_lemma} 
There exists a $v_*<0 $ such that $\partial_u r<0$ everywhere in $\mathcal{Q_B}\cap\{v>v_*\}$ and a $u_*<0$\footnote{Not necessarily $<u_s$. If $u_*<u_s$ then we have a stronger condition than the one in Assumptions \ref{instability_assumptions}. Namely, we get a semi-global trapped region in $\mathcal{R}$. See Lemma \ref{large_trapped_rgn} on that topic.} such that $\partial_v r<0$ everywhere in $\mathcal{Q_B}\cap\{u>u_*\}$.
\begin{rem}
    Since $\mathcal{Q_R}\subset\mathcal{Q_B}$, the first statement above implies $\partial_u r<0$ in $\mathcal{Q_R}\cap \{v>v_*\}$ and the second statement implies $\partial_v r<0$ in $\mathcal{Q_R}\cap \{u>u_*\}$. As $v_*<0$, the set $\mathcal{Q_R}\cap \{v>v_*\}$ is never empty, but the same is not true for $\mathcal{Q_R}\cap \{u>u_*\}$. Namely, this intersection is only nonempty if $u_*<u_s$, hence the additional condition in on the global parameter $E$.  
    If this additional condition is satisfied then an obvious consequence of this Lemma is that $\partial_u r, \partial_v r<0$ everywhere in $\mathcal{Q_R}\cap \{u>u_*\}\cap\{v>v_*\}$, the latter set being nonempty. 
\end{rem}
\begin{proof}
   
     Let $L$ and $\underline{L}$ be radial null geodesic vector fields such that $L$ generates the outgoing null hypersurfaces in $\mathcal{M}$ and $\underline{L}$ generates the ingoing null geodesic hypersurfaces in $\mathcal{M}$. Then in an arbitrary local double null chart $(U,V)$ defined on some subset $\mathcal{O}$ of $\mathcal{Q}$ we have $L\propto \partial_V$ and $\underline{L}\propto \partial_U$ in $\mathcal{O}\times S^2$. In particular, $L\propto \partial_v$ and $\underline{L}\propto \partial_u$ in $\mathcal{R}$ \\[5pt] Furthermore, the Raychaudhuri equation (\ref{raychaudhuri_general_form}) applies for the null congruences generated by $L$ and $\underline{L}$. For radial null geodesic congruences in a spherically symmetric spacetime (\ref{raychaudhuri_general_form}) simplifies to
     \begin{align}
         &L\hat{\theta}=-\frac{1}{2} \hat{\theta}^2-R_{ab}L^aL^b\label{constraint_tensorial_2}\  \text{ for the outgoing geodesics}\\
       &\underline{L\hat{\theta}}=-\frac{1}{2}\underline{\hat{\theta}}^2-R_{ab}\underline{L}^a\underline{L}^b \ \text{ for the ingoing geodesics}
     \end{align}
In the above, $R_{ab}$ denotes the Ricci tensor on $\mathcal{M}$ and  $\hat{\theta}:=\nabla\cdot L$ and $\underline{\hat{\theta}}:=\nabla\cdot\underline{L}$ denote the expansions of the outgoing and ingoing null geodesic congruences. The shear and rotation vanish due to the spherical symmetry and since $L$ and $\underline{L}$ are orthogonal to the spheres of symmetry. We calculate
\begin{align}
    &\hat{\theta}=\frac{2Lr}{r},  &\underline{\hat{\theta}}=\frac{2\underline{L}r}{r}
\end{align}Due to the Einstein equations (\ref{EEs_general})\footnote{We use the convention $R_{ab}-\frac{1}{2}g_{ab}R=2T_{ab}$ - the same convention is used in \cite{Luk_Oh}.} and using that $L$ and $\underline{L}$ are null, we have
\begin{align}
    &R_{ab}L^aL^b = 2 (T_{ab}^{\text{[sf]}}+T_{ab}^{\text{[EM]}})L^aL^b,  &R_{ab}\underline{L}^a\underline{L}^b = 2(T_{ab}^{\text{[sf]}}+T_{ab}^{\text{[EM]}})\underline{L}^a\underline{L}^b 
\end{align}
Applying the massless scalar field energy-momentum tensor (\ref{sf_stressenergy_tensor}) to $L^aL^b$ and $\underline{L}^a\underline{L}^b$, we find
\begin{align}
    &T_{ab}^{\text{[sf]}}L^a L^b = (L\phi)^2, &T_{ab}^{\text{[sf]}}\underline{L}^a \underline{L}^b = (\underline{L}\phi)^2
\end{align} 
where we again used that $L,\underline{L}$ are null. Applying the Maxwell field energy-momentum tensor given in (\ref{EM_stressenergy_tensor}) we have 
\begin{align}
    &T_{ab}^{\text{[EM]}}L^aL^b=0, &T_{ab}^{\text{[EM]}}\underline{L}^a\underline{L}^b=0
\end{align}
 because the first term of $T_{ab}^{\text{[EM]}}$ is antisymmetric and the second term is proportional to $g_{ab}$. It follows that the Raychaudhuri equations take the global form
\begin{align}
    &L\bigg(\frac{2Lr}{r}\bigg)= -\frac{1}{2}\bigg(\frac{2Lr}{r}\bigg)^2-2(L\phi)^2 \iff L(Lr)=-r(L\phi)^2\label{constraint_tensorial_outgoing}\\
    &\underline{L}\bigg(\frac{2\underline{L}r}{r}\bigg)= -\frac{1}{2}\bigg(\frac{2\underline{L}r}{r}\bigg)^2-2(\underline{L}\phi)^2 \iff \underline{L}(\underline{L}r)=-r(\underline{L}\phi)^2\label{constraint_tensorial_ingoing}
\end{align}
These are tensor equations on $\mathcal{M}$ so they are independent of the coordinates. 
     \\[5pt]
   By the future admissibility condition on the Cauchy data in Definition 3.1 \cite{Luk_Oh} (See also Lemma 2.3 and Remark 3.2 in \cite{Luk_Oh}), in the coordinates $(\rho,\theta,\varphi)$ parameterizing the initial spacelike hypersurface $\Sigma_0$, there exist $\rho_1<\rho_2$ such that\footnote{Technically, the condition given in Remark 3.2 of \cite{Luk_Oh} is that $\partial_V r|_{\Sigma_0} <0$ for $\rho\leq\rho_2$ and $\partial_U r|_{\Sigma_0} <0$ for $\rho\geq\rho_1$. It is meant that this condition holds in an \textit{arbitrary local double null chart} $(U,V)$ covering the region of interest.} 
    \begin{align}
           &Lr <0 \text{ on }\{\rho\leq\rho_2\}\subset\Sigma_0 \\
            &\underline{L}r <0 \text{ on }\{\rho\geq\rho_1\} \subset\Sigma_0
    \end{align}
   As mentioned in Remark 3.2 of \cite{Luk_Oh}, this implies that $\{\rho_1\leq \rho\leq \rho_2\}\subset\Sigma_0$ is trapped i.e. $\underline{L}r, L r<0$ in that region. Let $\rho\leq \rho_2$ and let $S_{\rho}$ be the sphere at $\rho$ on $\Sigma_0$, so $Lr|_{S_\rho}<0$. Let $\gamma$ be an integral curve of $L$ parameterised by affine parameter $s$ such that $\gamma(0)\in S_\rho$. By the Raychaudhuri equation (\ref{constraint_tensorial_outgoing}),    
   \begin{align}
       &Lr(\gamma(s)) = Lr(\gamma(0)) -\int_{0}^{s}r (L\phi)^2 (\gamma(s'))ds'<0
   \end{align}
   where we used that the integrand is nonnegative in the second term of the sum. Since $\rho\leq \rho_2$ was arbitrary, this shows that $Lr<0$ along the outgoing null generators originating from the region $\{\rho\leq\rho_2\}$ on $\Sigma_0$, which is equivalent to $Lr<0$ in $J^+(\Sigma_0\cap\{\rho\leq \rho_2\})$\footnote{To see this, let $p\in J^+(\Sigma_0\cap\{\rho\leq \rho_2\})$ and cover $J^-(p)$ with a double null chart $(U,V)$. This is possible by Subsection 5.1.2 of \cite{Luk_Oh}. There is a past-directed causal curve $\sigma$ from $p$ to some point $q$ on $\Sigma_0$ with $\rho(q)\leq\rho_2$. As $\sigma$ is past-directed, $U$ and $V$ are non-increasing along $\sigma$, so in particular $U(q)\leq U(p)$. As $\Sigma_0$ is a Cauchy surface for the spacetime, the past-directed outgoing radial null geodesic $\gamma$ through $p$ eventually hits $\Sigma_0$ at some point $q'$. Furthermore, $V$ is decreasing and $U$ is constant along $\gamma$, so in particular $U(p)=U(q')\geq U(q)$ and so $\rho(q')\leq \rho(q)\leq\rho_2$. Therefore $p$ lies on a future-directed radial outgoing nul geodesic emanating from $\Sigma_0\cap\{\rho\leq \rho_2\}$. The converse statement is obvious.}. 
  \\[5pt]
   Identically, choose $\underline{\rho}\geq \rho_1$ and let $S_{\underline{\rho}}$ denote the corresponding sphere in $\Sigma_0$, so $\underline{L}r|_{S_{\underline{\rho}}}<0$. Let $\underline{\gamma}$ be an integral curve of $\underline{L}$, affinely parameterised by $s$ and such that $\underline{\gamma}(0)\in S_{\underline{\rho}}$. By the Raychaudhuri equation (\ref{constraint_tensorial_ingoing}),
   \begin{align}
       &\underline{L}r(\underline{\gamma}(s)) = \underline{L}r(\underline{\gamma}(0)) -\int_{0}^{s}r (\underline{L}\phi)^2 (\underline{\gamma}(s'))ds'<0
   \end{align}
We used again that the integrand is nonnegative in the second term of the sum. Since $\underline{\rho}\geq \rho_1$ was arbitrary, this shows that $\underline{L}r<0$ along the ingoing null generators originating from the region $\{\rho\geq\rho_1\}$ on $\Sigma_0$, which is equivalent to $\underline{L}r<0$ in $J^+(\Sigma_0\cap\{\rho\geq \rho_1\})$\footnote{The reasoning is equivalent to that in the above footnote.}.\\[5pt]
Combining these results, we find that $Lr, \underline{L}r<0$ in $J^+(\{\rho_1\leq \rho\leq\rho_2\}\cap \Sigma_0)$. We cannot directly conclude that this implies the existence of the desired region in $\mathcal{Q_R}$, because we don't know where $\rho_1$ and $\rho_2$ are on $\Sigma_0$. But we can invoke Lemma A.1 from \cite{Luk_Oh}, which states that there are no trapped surfaces in the black hole exterior\footnote{In other terms, the \textbf{apparent horizon} - the boundary of trapped surfaces - always lies on or inside the event horizon.}. Therefore, $J^+(\{\rho_1\leq \rho\leq\rho_2\}\cap \Sigma_0)$ must lie in the black hole interior $\mathcal{B}$. In particular, this implies $\{(\rho, \theta,\varphi)\in \Sigma_0:\rho_1\leq\rho\leq\rho_2\}\subset\Sigma_{0,\mathcal{B}}$. On $\Sigma_{0,\mathcal{B}}/S^2$, in the Kruskal-like interior coordinates $(u,v)$ spanning from $-\infty$ to $0$, let $v_*$ be the $v$ coordinate of the sphere at $\rho_1$ and let $u_*$ be the $u$ coordinate of the sphere at $\rho_2$. Then $Lr<0$ in $J^+(\Sigma_0\cap\{\rho\leq \rho_2\})$ implies $\partial_v r<0$ in $\{u>u_*\}\cap \mathcal{Q_B}$. Similarly, $\underline{L}r<0$ in $J^+(\Sigma_0\cap \{\rho\geq \rho_1\})$ implies $\partial_u r<0$ in $\{v>v_*\}\cap \mathcal{Q_B}$.
\end{proof}
\end{lem}
In the following Lemma, we show that the second condition in Assumptions \ref{instability_assumptions} is satisfied in $\mathcal{R}$. Given the first condition in Assumptions \ref{instability_assumptions} and Lemma \ref{choose_domain_lemma}, the proof of the existence of the local trapped diamond around every $u_0$ is a simple continuity argument.
\begin{lem}\label{thin_trapped_strip_lemma}
    For each $u_0\in (-\infty, u_s)$ there exists an $\varepsilon>0$ and a $v_0\in(v_\kappa,0)$ such that $[u_0-\varepsilon, u_0+\varepsilon]\times[v_0,0)$ is a trapped region.
    \begin{proof}
        If $u_0\in (-\infty,u_s)$. Since $\partial_v r(u,v)\to -\infty$ as $v\to 0$, there exists a $v_0<0$ such that $\partial_v r(u,v)\leq -1 \ \forall v\geq v_0$. By continuity of $\partial_v r$, $\exists \varepsilon>0$ such that $\partial_v r<0$ on $[u-\varepsilon,u+\varepsilon]\times\{v_0\}$. From the Raychaudhuri equation (\ref{constraint_eqn1}) it follows that $\partial_v r<0$ in $[u-\varepsilon,u+\varepsilon]\times[v_0,0)$.
        Next, choose a point $q\in[u-\varepsilon,u+\varepsilon]\times[v_0,0)\times S^2$. Consider the past-directed radial ingoing null geodesic through $q$. This hits $\Sigma_0$ at some local coordinate $\rho$\footnote{Recall that $\rho$ parameterises $\Sigma_0/S^2$.}. This point must lie in the black hole exterior so by the proof of Lemma \ref{choose_domain_lemma}, $\rho \geq \rho_1$ (as $\rho_1$ must be in the black hole interior according to the previous Lemma). By Lemma A.1 from \cite{Luk_Oh}, $\underline{L}r<0$ at the point parameterised by $\rho$ on $\Sigma_0$, 
        so by an application of the Raychaudhuri equations (an argument identical to those used in
        the proof of Lemma \ref{choose_domain_lemma}), $\underline{L}r<0$ at $q$ and so $\partial_u r<0$ at $q$. Since $q$ was arbitrary, $\partial_u r<0$ in $[u-\varepsilon,u+\varepsilon]\times[v_0,0)$. Therefore, this strip is trapped. 
    \end{proof}
\end{lem}
Lemma \ref{choose_domain_lemma} proves that the region $J^+(\{\rho_1\leq \rho\leq\rho_2\}\cap \Sigma_0)/S^2 = \{(u,v)\in \mathcal{Q_B}: u>u_* \text{ and } v>v_*\}$ is trapped. Under a mild additional condition on the Einstein--Maxwell--scalar field initial data in Definition 3.1 of \cite{Luk_Oh}, namely that the trapped subset of the initial hypersurface $\Sigma_0$ is sufficiently large, we can ensure that the trapped subset of $\mathcal{Q_R}$ given by $\{(u,v)\in \mathcal{Q_R}:u>u_* \text{ and }v>v_*\}$ is nonempty. This is encapsulated in the final result of this paper, Lemma \ref{large_trapped_rgn} below.
\begin{lem}\label{large_trapped_rgn}  Let $u_*$ be as in Lemma \ref{choose_domain_lemma}. If the coordinate $\rho_2$ in Definition 3.1 of \cite{Luk_Oh} is sufficiently large, then $u_*<u_s$.
\begin{proof}
    Recall that $u_*$ is defined to be the ingoing coordinate of the sphere at $\rho_2$ in the Kruskal-like coordinates $(u,v)$, defined on $\mathcal{Q_B}$. Let $\rho_{\mathcal{H}_1}$ denote the coordinate where the event horizon $\mathcal{H}_1$ intersects the Cauchy hypersurface $\Sigma_0$. 
    As a consequence of Lemma A.1 of \cite{Luk_Oh},  $\rho_2\leq\rho_{\mathcal{H}_1}$. In $(u,v)$ coordinates we have $u\to -\infty$ as $\rho\to \rho_{\mathcal{H}_1}^-$. So in particular there exists a $\rho_{s}$ such that $\rho\geq\rho_s \implies u(\rho)<u_s$.  Hence if the initial data for the global Cauchy problem in \cite{Luk_Oh} satisfies $\rho_2\geq \rho_s$, then $u_*=u(\rho_2)<u_s$.
\end{proof}
\end{lem}

\begin{rem}
    We expect that an alternative to requiring that $\rho_2$ is large is requiring that the spherically symmetric scalar field perturbation on the Reissner-Nordstr\"{o}m initial data is small. This would be characterised by a small constant $E$ and a large constant $s$ in Theorem \ref{stability_of_CH_original}, which would make the existence time $u_s$ larger. However, we do not investigate this conjecture further in the current article. 
\end{rem}

\begin{figure}[H]
    \centering
    \includegraphics[width=0.8\linewidth]{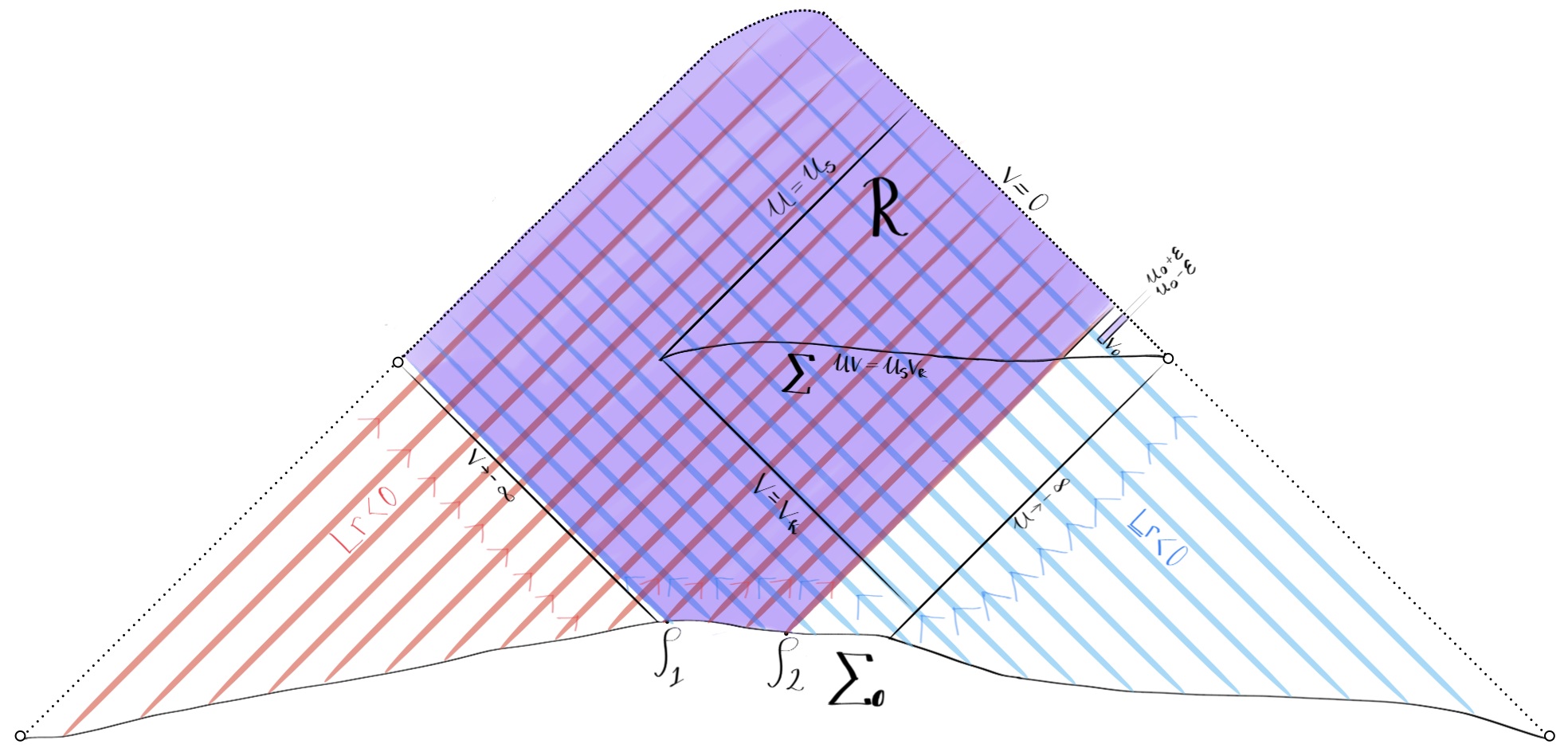}
    \caption{Illustration of the proofs of Lemma \ref{choose_domain_lemma} and Lemma \ref{thin_trapped_strip_lemma} in the case $u_*<u_s$. By the Raychaudhuri equations, the sign of $Lr$ is preserved along outgoing null geodesics (red) and the sign of $\underline{L}r$ is preserved along ingoing null geodesics (blue), emanating from $\Sigma_0$. By Definition 3.1 of \cite{Luk_Oh}, $Lr<0$ for $\rho\leq \rho_2$ and $\underline{L}r<0$ for $\rho\geq\rho_1$. Since the segment of $\Sigma_0$ defined by $\rho_1\leq\rho\leq \rho_2$ is trapped, it must be between $\mathcal{H}_1$ and $\mathcal{H}_2$. The semi-global trapped region where the blue and red lines intersect is highlighted in purple. There is a small purple strip around $u_0$ outside of this region - this represents the thin trapped strip whose existence is proved in Lemma \ref{thin_trapped_strip_lemma}. }
    \label{trapped_region_illustration}
\end{figure}

\subsection{Concluding remark}
    Recall that as a result of the abovementioned blow-up results from Luk-Oh, the estimates we establish for the solution $(\rho, U)$ of the characteristic initial value problem for the stiff fluid in $\mathcal{R}$ indeed imply that the energy density of any such solution blows up at the singularity. Combining this result  with Lemma \ref{unif_bounded_u_derivatives} shows that the $u$-componenet of the stiff fluid velocity blows up at the singularity in $(\mathcal{R}, g)$. This is in contrast with the spherically symmetric dust model, where we showed that the $L^1$ bounds on the $v$-derivatives from the Stability of the Cauchy Horizon Theorem lead to bounded dust energy density and fluid velocity components at the singularity.

\appendix\label{appendix:a}\section{Appendix}

\subsection{Einstein equations with a Maxwell field and a massless scalar field}
The spacetimes in \cite{Luk_Oh} are spherically symmetric solutions of the Einstein equations
\begin{align}
    R_{ab}-\frac{1}{2}g_{ab}R=2\big(T_{ab}^{[\text{sf}]}+T_{ab}^{\text{[EM]}}\big) \label{EEs_general}
\end{align}
where $R_{ab}$ denotes the Ricci tensor and $R=R_c^{\ c}$ denotes the Ricci scalar of the spacetime $(\mathcal{M},g)$, and the massless scalar field energy-momentum tensor is
\begin{align}  T_{ab}^{\text{[sf]}}:=\nabla_a \phi\nabla_b\phi - \frac{1}{2}g_{ab}\nabla_c\phi\nabla^c\phi \label{sf_stressenergy_tensor}
\end{align}
and the Maxwell field energy momentum tensor is 
\begin{align}
    T_{ab}^{\text{[EM]}}:=F_{ac}F_b^{\ c}-\frac{1}{4}g_{ab}F_{cd}F^{cd} \label{EM_stressenergy_tensor}
\end{align}
where $F=dA$, the one-form $A$ being the electromagnetic potential.
\begin{subsection}{The Raychaudhuri equation for a null geodesic congruence}
Let $L$ be an affinely parameterised null geodesic vector field generating a null geodesic congruence on $(\mathcal{M},g)$. Then the Raychaudhuri equation holds:
\begin{align}
    L\hat{\theta} = -\frac{1}{2}\hat{\theta}^2 -\hat{\sigma}_{ab}\hat{\sigma}^{ab} + \hat{\omega}_{ab}\hat{\omega}^{ab}-R_{ab}L^a L^b \label{raychaudhuri_general_form}
\end{align}
where the scalar function $\hat{\theta}=\nabla\cdot L$ is the expansion of $L$, the $(0,2)$-tensor  $\hat{\sigma}_{ab}$ denotes the shear of $L$ and the 2-form $\hat{\omega}_{ab}$ denotes the rotation of $L$. 
\end{subsection}

 \subsection{Spherically symmetric Einstein--Maxwell--scalar field evolution equations and constraints in double null coordinates} \label{appendix_evol_eqns_and_constraints}
The Einstein--Maxwell--scalar field evolution equations in spherical symmetry take the following form in double null coordinates:
\begin{align}
&r_{,uv}= -\frac{e^{2\omega}}{4r} - \frac{r_{,u}r_{,v}}{r} +\frac{e^{2\omega}q^2}{4r^3} \label{evol_eqn_r_appendix}\\
&\phi_{,uv} = - \phi_{,u} (\ln r)_{,v} - \phi_{,v} (\ln r)_{,u} \\
&\omega_{,uv} = - \phi_{,u} \phi_{,v} + \frac{e^{2\omega}}{2r^2}\bigg(\frac{1}{2} - \frac{q^2}{r^2}\bigg)
+ (\ln r)_{,u} (\ln r)_{,v} \label{evol_eqn_omega_appendix}
\end{align}
And the Raychaudhuri constraints for the same system in double null coordinates take the form:
\begin{align}
    & \partial_v(e^{-2\omega}r_{,v}) = - re^{-2\omega} \phi_{,v}^2 \label{constraint_eqn1} \\
    & \partial_u(e^{-2\omega}r_{,u}) = - re^{-2\omega} \phi_{,u}^2 \label{constraint_eqn2}
\end{align}

\subsection{Christoffel symbols and Riemann tensor components in spherical symmetry}\label{appendix_riemann_components}
For convenience we list the Christoffel symbols of $g$ in double null coordinates $(u,v,\theta, \varphi)$:
\begin{align} 
   & \Gamma^u_{\mu\nu} = 
    \begin{pmatrix}
        2\omega_{,u} & 0 & 0 & 0 \\
        0 & 0 & 0 & 0 \\
        0 & 0 & 2e^{-2\omega}r r_{,v} & 0 \\
        0 & 0 & 0 & 2e^{-2\omega} r r_{,v} \sin^{2}(\theta)
    \end{pmatrix}; \ \ 
    \Gamma^v_{\mu\nu} = 
    \begin{pmatrix}
        0 & 0 & 0 & 0 \\
        0 & 2\omega_{,v} & 0 & 0 \\ 
        0 & 0 & 2e^{-2\omega}r r_{,u} & 0\\
        0 & 0 & 0 & 2e^{-2\omega} r r_{,u} \sin^{2}(\theta) 
    \end{pmatrix}; \nonumber \\[5pt]
    & \Gamma^\theta_{\mu\nu} =
    \begin{pmatrix}
        0 & 0 & \frac{r_{,u}}{r}  & 0 \\
        0 & 0 & \frac{r_{,v}}{r}  & 0 \\
        \frac{r_{,u}}{r}  & \frac{r_{,v}}{r}  & 0 & 0 \\
        0 & 0 & 0 & -\sin{(\theta)}\cos{(\theta)}
    \end{pmatrix}; \ \ 
    \Gamma^\varphi_{\mu\nu} = 
    \begin{pmatrix}
        0 & 0 & 0 & \frac{r_{,u}}{r}  \\
        0 & 0 & 0 & \frac{r_{,v}}{r}  \\
        0 & 0 & 0 & \cot{(\theta)} \\
        \frac{r_{,u}}{r}  & \frac{r_{,v}}{r}  & \cot{(\theta)} & 0 
    \end{pmatrix} \label{Christoffel_Symbols}
\end{align}
The Christoffel symbols of $g_{\mathcal{Q_R}}$ in $(u,v)$ coordinates are just the upper-left $2\times2$ submatrices of $\Gamma^\mu_{\nu\rho}$,
\begin{align}
     & \mathring{\Gamma}^u_{\mu\nu} = 
    \begin{pmatrix}
        2\omega_{,u} & 0 \\
        0 & 0 
    \end{pmatrix}; \ \ 
    \mathring{\Gamma}^v_{\mu\nu} = 
    \begin{pmatrix}
        0 & 0   \\
        0 & 2\omega_{,v}
    \end{pmatrix} 
\end{align}

Hence a calculation shows that the Riemann tensor components of $g_{\mathcal{Q_R}}$ in the $(u,v)$ double null coordinates are:
\begin{align}
     &(\mathring{R}^u_{\ u\mu\nu}) =  
    \begin{pmatrix} 
    0 & -2\omega_{,uv}  \\
    2\omega_{,uv} & 0 
    \end{pmatrix} =  - (\mathring{R}^v_{\ v\mu\nu});
\end{align}

\subsection{Reissner-Nordstr\"{o}m interior geometry}\label{RN_int_geometry}
Following Section 5 from Luk-Oh \cite{Luk_Oh}, in the (Kruskal-like) coordinates $(u,v)$ which are regular near $\mathcal{R}\cap \mathcal{CH}^+$, the Reissner-Nordstr\"{o}m metric takes the form
\begin{equation}
    g_{RN} = -e^{2\omega_{RN}} du dv + r^2 g_{S^2}
\end{equation}
where the RN lapse function is 
\begin{align}
    e^{2\omega_{RN}} = - \frac{1}{\kappa^2_- uv} \bigg( 1 - \frac{2M}{r_{RN}} + \frac{q^2}{r_{RN}^2} \bigg) \label{express_OmegaRN_appendix}
\end{align}
where $r_{RN} = r_{RN}(u,v)$ is the RN area radius function. The RN surface gravities are $\kappa_\pm = (r_+ - r_-)r_{\pm}^{-2}$ where
the inner and outer horizon area radii are given by the roots of (\ref{express_OmegaRN_appendix}) i.e. $r_\pm = M \pm \sqrt{M^2-q^2}$. \\[5pt]
The $u,v$ double null coordinates are related to the RN tortoise coordinate $r_{RN}^*$ by 
$4\kappa_{-}^2 uv = e^{-2\kappa_- r_{RN}^*}$, where the tortoise coordinate is given by
\begin{align}
    r_{RN}^* = r_{RN} + \frac{1}{2\kappa_+} \ln (r_+ - r_{RN}) - \frac{1}{2\kappa_-} \ln(r_{RN} - r_-)
\end{align}
Hence the RN lapse function in this double null chart can be expressed in terms of the area radius function as
\begin{equation}
    e^{2\omega_{RN}} = 4r_{RN}^{-2}e^{2\kappa_- r_{RN}}(r_{+}-r_{RN})^{1+\kappa_-/\kappa_+}
\end{equation}
We also have the following expressions for the ingoing and outgoing derivatives of the RN area radius function in the Kruskal-like coordinates, c.f. Section 5.1 in \cite{Luk_Oh}:
\begin{align}
    &\partial_u r_{RN} = -\frac{1}{2\kappa_- u} \bigg( 1 - \frac{2M}{r_{RN}} + \frac{q^2}{r_{RN}^2}\bigg) \label{express_durrn_appendix} \\
    &\partial_v r_{RN} = -\frac{1}{2\kappa_- v} \bigg( 1 - \frac{2M}{r_{RN}} + \frac{q^2}{r_{RN}^2}\bigg) \label{express_dvrrn_appendix}
\end{align}

\printbibliography\nocite{*}

@article{Sbierski2021,
  author  = {Sbierski, Jan},
  title   = {On holonomy singularities in general relativity and the {$C^{0,1}_{loc}$}-inextendibility of spacetimes},
  journal = {arXiv preprint},
  year    = {2020},
  note    = {arXiv:2007.12049 [gr-qc]},
}

@article{Luk_Oh,
  author  = {Luk, Jonathan and Oh, Sung-Jin},
  title   = {Strong cosmic censorship in spherical symmetry for two-ended asymptotically flat initial data {I}. {T}he interior of the black hole region},
  journal = {Inventiones mathematicae},
  volume  = {217},
  pages   = {1--111},
  year    = {2019},
  note    = {arXiv:1702.05715 [gr-qc]},
}

@article{Sbierski_Fournodavlos,
  author  = {Fournodavlos, Grigorios and Sbierski, Jan},
  title   = {Generic blow-up results for the wave equation in the interior of a {S}chwarzschild black hole},
  journal = {Archive for Rational Mechanics and Analysis},
  volume  = {235},
  pages   = {2141--2182},
  year    = {2020},
  note    = {arXiv:1804.01941 [gr-qc]},
}

@article{Chr_SelfGravitatingFluid,
  author  = {Christodoulou, Demetrios},
  title   = {Self-gravitating relativistic fluids: {A} two-phase model},
  journal = {Archive for Rational Mechanics and Analysis},
  volume  = {130},
  number  = {4},
  pages   = {343--400},
  year    = {1995},
  doi     = {10.1007/BF00375144},
}

@article{Dafermos2012,
  author  = {Dafermos, Mihalis},
  title   = {Black holes without spacelike singularities},
  journal = {Communications in Mathematical Physics},
  volume  = {332},
  number  = {2},
  pages   = {729--757},
  year    = {2014},
  note    = {arXiv:1201.1797 [gr-qc]},
}

@article{Gavassino,
  author  = {Gavassino, L.},
  title   = {Solitons and singularities in relativistic ultrastiff fluids},
  journal = {arXiv preprint},
  year    = {2025},
  note    = {arXiv:2504.20332 [gr-qc]},
}

@article{acoustic_spacetimes,
  author  = {Visser, Matt and Molina-Paris, Carmen},
  title   = {Acoustic geometry for general relativistic barotropic irrotational fluid flow},
  journal = {New Journal of Physics},
  volume  = {12},
  number  = {9},
  pages   = {095014},
  year    = {2010},
  note    = {arXiv:1001.1310 [gr-qc]},
}

@article{Luk_WNS,
  author  = {Luk, Jonathan},
  title   = {Weak null singularities in general relativity},
  journal = {Journal of the American Mathematical Society},
  volume  = {30},
  number  = {3},
  pages   = {847--927},
  year    = {2017},
  note    = {arXiv:1311.4970 [gr-qc]},
}

@article{Song,
  author  = {Song, Y.},
  title   = {Weak null singularity for the {E}instein--{E}uler system},
  journal = {arXiv preprint},
  year    = {2025},
  note    = {arXiv:2506.16635 [gr-qc]},
}

@article{Dafermos03,
  author  = {Dafermos, Mihalis},
  title   = {Stability and instability of the {C}auchy horizon for the spherically symmetric {E}instein--{M}axwell--scalar field equations},
  journal = {Annals of Mathematics},
  volume  = {158},
  number  = {3},
  pages   = {875--928},
  year    = {2003},
  doi     = {10.4007/annals.2003.158.875},
}

@article{Dafermos05,
  author  = {Dafermos, Mihalis},
  title   = {The interior of charged black holes and the problem of uniqueness in general relativity},
  journal = {Communications on Pure and Applied Mathematics},
  volume  = {58},
  number  = {4},
  pages   = {445--504},
  year    = {2005},
}

@article{Sbierski2024,
  author  = {Sbierski, Jan},
  title   = {Lipschitz inextendibility of weak null singularities from curvature blow-up},
  journal = {arXiv preprint},
  year    = {2024},
  note    = {arXiv:2409.18838 [gr-qc]},
}

@article{Chr10,
  author  = {Christodoulou, Demetrios},
  title   = {The formation of black holes and singularities in spherically symmetric gravitational collapse},
  journal = {Communications on Pure and Applied Mathematics},
  volume  = {44},
  number  = {3},
  pages   = {339--373},
  year    = {1991},
}

@article{LinearStabilityOfSchwartzschild,
  author  = {Dafermos, Mihalis and Holzegel, Gustav and Rodnianski, Igor},
  title   = {The linear stability of the {S}chwarzschild solution to gravitational perturbations},
  journal = {Acta Mathematica},
  volume  = {222},
  number  = {1},
  pages   = {1--214},
  year    = {2019},
}

@article{NonlinearStabilityOfSchwartzschild,
  author  = {Dafermos, Mihalis and Holzegel, Gustav and Rodnianski, Igor and Taylor, Martin},
  title   = {The non-linear stability of the {S}chwarzschild family of black holes},
  journal = {arXiv preprint},
  year    = {2021},
  note    = {arXiv:2104.08222 [gr-qc]},
}

@article{Joshi_Saraykar,
  author  = {Joshi, P. S. and Saraykar, R. V.},
  title   = {Shell-crossings in gravitational collapse},
  journal = {International Journal of Modern Physics D},
  volume  = {22},
  number  = {11},
  pages   = {1350027},
  year    = {2013},
}

@article{Szekeres_Lun,
  author  = {Szekeres, P. and Lun, A.},
  title   = {What is a shell-crossing singularity?},
  journal = {Publications of the Astronomical Society of Australia},
  volume  = {16},
  number  = {1},
  pages   = {1--11},
  year    = {1999},
}

@article{Fr_Klein,
  author  = {Frauendiener, J. and Klein, C.},
  title   = {On crossing of dust shells},
  journal = {Journal of Mathematical Physics},
  volume  = {36},
  number  = {7},
  pages   = {3632--3643},
  year    = {1995},
}

@article{GC_CS_ITF,
  author  = {Zhang, X. and Gao, S.},
  title   = {Geodesic completeness, curvature singularities and infinite tidal forces},
  journal = {arXiv preprint},
  year    = {2025},
  note    = {arXiv:2507.04616 [gr-qc]},
}

@article{Rendall,
  author  = {Rendall, A. D.},
  title   = {Reduction of the characteristic initial value problem to the {C}auchy problem and its applications to the {E}instein equations},
  journal = {Proceedings of the Royal Society of London. Series A},
  volume  = {427},
  pages   = {221--239},
  year    = {1990},
}

@book{Ringstrom,
  author    = {Ringstr{\"o}m, Hans},
  title     = {The Cauchy Problem in General Relativity},
  publisher = {European Mathematical Society},
  series    = {ESI Lectures in Mathematics and Physics},
  year      = {2009},
}

@unpublished{Dafermos_ETHlectures,
  author = {Dafermos, Mihalis},
  title  = {The geometry and analysis of black hole spacetimes in general relativity},
  note   = {ETH Nachdiplom lectures},
  year   = {2012},
}

\end{document}